\documentclass{eptcs}

\usepackage{bbm}

\usepackage{titling}
\usepackage{hhline}
\usepackage{bookmark}
\usepackage[table]{xcolor}
\usepackage[all,cmtip]{xy}

\newcommand{\R}{\mathbb{R}}

\usepackage{tikzit}

\tikzstyle{H}=[draw, color=black, fill={rgb:black,1;white,3}, shape=rectangle, tikzit fill=yellow]
\tikzstyle{thick}=[draw, line width=0.5mm]

\tikzstyle{Z}=[rounded rectangle, minimum height=12pt, font={\boldmath}, text centered,inner sep=1.5mm, outer sep=-2mm, scale=.8, draw=black, fill=white, tikzit draw=black, tikzit fill=whites]

\tikzstyle{X}=[rounded rectangle, minimum height=12pt, font={\large\boldmath},text=white,text centered, inner sep=1.5mm, outer sep=-2mm, scale=.8, draw=black, fill=black!55, tikzit draw=black, tikzit fill={rgb,255: red,191; green,191; blue,191}]

\tikzstyle{classicalwire}=[color=black,decorate,decoration={coil,segment length=3pt,aspect=1,amplitude=.8pt}]

\tikzstyle{electricalcontrol}=[color=blue!20!black!80, dash pattern=on 3pt off 2pt]

\tikzstyle{Zthick}=[rounded rectangle, minimum height=12pt, font={\large},text centered, inner sep=1.5mm, outer sep=-2mm, scale=.8, draw={rgb:black,3;white,1}, fill=white,line width=0.75mm, tikzit draw=black, tikzit fill=whites]
\tikzstyle{Xthick}=[rounded rectangle, minimum height=12pt, font={\large\boldmath}, text centered, text=white, inner sep=1.5mm, outer sep=-2mm, scale=.8, draw={rgb:black,3;white,1},fill=black!55,line width=0.75mm, tikzit draw=black, tikzit fill={rgb,255: red,191; green,191; blue,191}]

\tikzstyle{phase}=[draw, fill=white, diamond, scale=1, inner sep=0pt, minimum size=10pt]
\tikzstyle{discard}=[draw, xscale=2.2, ground, rotate=90]
\tikzstyle{mmixed}=[draw, quantum, yscale=-2.2, ground, rotate=180]
\tikzstyle{quantum}=[line width=.6mm]
\tikzstyle{map}=[draw, color=black, fill=white, rectangle]
\tikzstyle{mapperp}=[draw, color=white, fill=black, text=white, rectangle]
\tikzstyle{s}=[draw, color=black, fill=gray, rectangle]
\tikzstyle{mapthick}=[draw, color=black, fill=white, rectangle, inner sep=0pt, minimum size=15pt, line width=0.5mm]
\tikzstyle{otimes}=[draw, fill=white, rotate=45, scale=0.9, minimum height=.1cm, circle, append after command={[shorten >=\pgflinewidth, shorten <=\pgflinewidth]
(\tikzlastnode.north) edge (\tikzlastnode.south)
(\tikzlastnode.east) edge (\tikzlastnode.west)}]
\tikzstyle{dot}=[thick, fill=black, circle, scale=1, inner sep=.05cm]
\tikzstyle{oplus}=[draw, scale=0.9, minimum height=.1cm, circle, append after command={[shorten >=\pgflinewidth, shorten <=\pgflinewidth]
(\tikzlastnode.north) edge (\tikzlastnode.south)
(\tikzlastnode.east) edge (\tikzlastnode.west)
}]
\tikzstyle{andin}=[draw, and gate US, rotate=90, scale=1, fill=white, label={center:{\it \&}}]
\tikzstyle{mulin}=[draw, and gate US, rotate=90, scale=1, fill=white, label={center:{}}]
\tikzstyle{andout}=[draw, and gate US, rotate=-90, scale=1, fill=white, label={center:{\it \&}}]
\tikzstyle{scalar}=[draw, rounded corners=1ex, rectangle round south west=false, rectangle round south east=false, tikzit fill={rgb,255: red,129; green,253; blue,255}, tikzit shape=rectangle]
\tikzstyle{scalarop}=[draw, rounded corners=1ex, rectangle round north west=false, rectangle round north east=false, tikzit fill={rgb,255: red,116; green,172; blue,255}, tikzit shape=rectangle]
\tikzstyle{fanin}=[draw, shape border rotate=30, regular polygon, regular polygon sides=3, fill=white, inner sep=.1cm]
\tikzstyle{fanout}=[draw, shape border rotate=-30, regular polygon, regular polygon sides=3, fill=white, inner sep=.1cm]
\tikzstyle{onein}=[draw, shape border rotate=30, regular polygon, regular polygon sides=3, fill=black, inner sep=.04cm, scale=1.2]
\tikzstyle{oneout}=[draw, shape border rotate=-30, regular polygon, regular polygon sides=3, fill=black, inner sep=.04cm, scale=1.2]
\tikzstyle{zeroin}=[draw, shape border rotate=30, regular polygon, regular polygon sides=3, fill=white, inner sep=.04cm, scale=1.2]
\tikzstyle{zeroout}=[draw, shape border rotate=-30, regular polygon, regular polygon sides=3, fill=white, inner sep=.04cm, scale=1.2]

\input{lagrel.tikzdefs}

\usepackage{mdframed}
\usepackage{arydshln}
\usepackage{multicol}

\usepackage{everypage}
\usepackage{lipsum}
\usepackage{amsthm}
\usepackage[inline]{enumitem}   
\usepackage{scalerel,stackengine}
\stackMath
\renewcommand\hat[1]{%
\savestack{\tmpbox}{\stretchto{%
  \scaleto{%
    \scalerel*[\widthof{\ensuremath{#1}}]{\kern-.6pt\bigwedge\kern-.6pt}%
    {\rule[-\textheight/2]{1ex}{\textheight}}
  }{\textheight}%
}{0.5ex}}%
\stackon[1pt]{#1}{\tmpbox}%
}
\makeatletter
\newcommand{\inlineitem}[1][]{%
\ifnum\enit@type=\tw@
    {\descriptionlabel{#1}}
  \hspace{\labelsep}%
\else
  \ifnum\enit@type=\z@
       \refstepcounter{\@listctr}\fi
    \quad\@itemlabel\hspace{\labelsep}%
\fi}
\makeatother
\renewcommand{\phi}{\varphi}

\newcommand{\xrightarrowtail}[1]{\!\!{\xymatrix@C=1em{\ar@{>->}[r]^{#1}&}}\!\!\!}
\newcommand{\xleftarrowtail}[1]{\!\!\!{\xymatrix@C=1em{&\ar@{>->}[l]_{#1}}}\!\!}

\newcommand{\xrightarrowiso}[1]{\!\!{\xymatrix@C=1em{\ar@{->}[r]^{#1}_\cong&}}\!\!\!}
\newcommand{\xleftarrowiso}[1]{\!\!\!{\xymatrix@C=1em{&\ar@{->}[l]_{#1}^\cong}}\!\!}

\theoremstyle{theorem}
  \newtheorem{theorem}{Theorem}[section]
  \newtheorem{corollary}[theorem]{Corollary}
  \newtheorem{lemma}[theorem]{Lemma}

\theoremstyle{definition}

  \newtheorem{definition}[theorem]{Definition}
  \newtheorem{example}[theorem]{Example}
  
  \newtheorem{remark}[theorem]{Remark}

\newcommand{\Mat}{\mathsf{Mat}}

\newcommand{\Not}{\mathsf{not}}

\newcommand{\Rel}{\mathsf{Rel}}
\newcommand{\op}{\mathsf{op}}

\newcommand{\FdHilb}{\mathsf{FHilb}}
\newcommand{\FHilb}{\mathsf{FHilb}}
\newcommand{\CPM}{\mathsf{CPM}}

\newcommand{\Aff}{\mathsf{Aff}}

\newcommand{\cb}{{\sf cb}}

\newcommand{\F}{\mathbb{F}}
\newcommand{\X}{\mathbb{X}}

\newcommand{\Z}{\mathbb{Z}}
\newcommand{\N}{\mathbb{N}}

\newcommand{\Lag}{\mathsf{Lag}}
\newcommand{\im}{\mathsf{im}}
\newcommand{\ZX}{\mathsf{ZX}}

\DeclareMathSymbol{\bot}{\mathord}{symbols}{"3F}

\renewcommand{\epsilon}{\varepsilon}
\renewcommand{\bar}[1]{\overline{#1}\hspace*{.01cm}}

\newcommand{\Stab}{{\sf Stab}}
\newcommand{\LinRel}{\sf LinRel}

\newcommand{\Isot}{{\sf Isot}}
\newcommand{\Co}{{\sf Co}}

\usepackage{amsmath}
\usepackage{pict2e}

\newdir{|>}{-<5pt,0pt>{
\begin{tikzpicture}[scale=.7]
	\begin{pgfonlayer}{nodelayer}
		\node [style=none] (0) at (0, 0) {};
		\node [style=none] (1) at (1, 0) {};
		\node [style=none] (2) at (-1, -0.25) {};
	\end{pgfonlayer}
	\begin{pgfonlayer}{edgelayer}
		\draw (2.center) to (0.center);
		\draw (0.center) to (1.center);
	\end{pgfonlayer}
\end{tikzpicture}
}}
\newdir{|<}{-<5pt,0pt>{
\begin{tikzpicture}[scale=.9]
	\begin{pgfonlayer}{nodelayer}
		\node [style=none] (0) at (0, -0.25) {};
		\node [style=none] (1) at (-1, -0.25) {};
		\node [style=none] (2) at (1, 0) {};
	\end{pgfonlayer}
	\begin{pgfonlayer}{edgelayer}
		\draw (2.center) to (0.center);
		\draw (0.center) to (1.center);
	\end{pgfonlayer}
\end{tikzpicture}
}}

\DeclareFontFamily{U}{mathx}{\hyphenchar\font45}
\DeclareFontShape{U}{mathx}{m}{n}{
      <5> <6> <7> <8> <9> <10>
      <10.95> <12> <14.4> <17.28> <20.74> <24.88>
      mathx10
      }{}
\DeclareSymbolFont{mathx}{U}{mathx}{m}{n}
\DeclareFontSubstitution{U}{mathx}{m}{n}
\DeclareMathAccent{\widecheck}{0}{mathx}{"71}
\DeclareMathAccent{\wideparen}{0}{mathx}{"75}

\usepackage{amsmath}

\usepackage{hyperref}

\newcommand\numeq[2]%
  {\label{#2}\stackrel{\scriptscriptstyle(\mkern-1.5mu#1\mkern-1.5mu)}{=}}

\tikzset{meter/.append style={draw, inner sep=10, rectangle, font=\vphantom{A}, minimum width=30, line width=.8,
 path picture={\draw[black] ([shift={(.1,.3)}]path picture bounding box.south west) to[bend left=50] ([shift={(-.1,.3)}]path picture bounding box.south east);\draw[black,-latex] ([shift={(0,.1)}]path picture bounding box.south) -- ([shift={(.3,-.1)}]path picture bounding box.north);}}}

\xymatrixrowsep{.5cm}
\xymatrixcolsep{.65cm}

\title{The Algebra for Stabilizer Codes}
\def\titlerunning{The Algebra for Stabilizer Codes}
\date{\today}
\author{Cole Comfort \\ Department of Computer Science, University of Oxford}
\def\authorrunning{Cole Comfort}

\begin{document}

\title{The Algebra for Stabilizer Codes}

\maketitle

\begin{abstract}
There is a bijection between odd prime dimensional qudit pure stabilizer states modulo invertible scalars and affine Lagrangian subspaces of finite dimensional symplectic $\F_p$-vector spaces.  In the language of the stabilizer formalism, full rank stabilizer tableaux are exactly the bases for affine Lagrangian subspaces. This correspondence extends to an isomorphism of props: the composition of stabilizer circuits corresponds to the relational composition of affine subspaces spanned by the tableaux, the tensor product corresponds to the direct sum.  In this paper, we extend this correspondence between stabilizer circuits and tableaux to the mixed setting;  regarding stabilizer codes as affine coisotropic subspaces (again only in odd prime qudit dimension/for qubit CSS codes).  We show that by splitting the projector for a stabilizer code we recover the error detection protocol and the error correction protocol with affine classical processing power.
\end{abstract}

\section{Introduction}
\label{sec:intro}

The connection between the stabilizer formalism and symplectic geometry has been known for quite a while, at least as early as the papers of Calderbank, Rains, Shor and Sloane
\cite{css,cssone}.
This was also discovered by Gross at a later date \cite{gross}. However, its role in the stabilizer formalism is often underplayed, for example, it is not explicitly mentioned in Gottesman's highly influential PhD thesis \cite{gottesman}.  Perhaps a reason for this is that despite their dominance in the quantum computer literature,
qubit stabilizer circuits do not conform so nicely to the symplectic geometric framework as do all other prime qudit dimensions (one could argue, because 2 is the oddest prime).
\footnote{Most research in quantum computing has traditionally studied qubits, which has hindered this formalism. However qudit quantum computing has started to become more researched, especially in the ZX-calculus community, which aims to create graphical languages to be used as tools for optimizing and for performing calculations involving quantum circuits.  For example, the odd prime qudit stabilizer ZX-calculus has recently enjoyed a traditional completeness theorem \cite{booth}, shortly following a ``doubled completeness theorem'' in a preceding paper \cite{lagrel}, being preceded itself by a proper completeness theorem for qutrit stabilizer circuits \cite{wangqutrit}. There has also been work on the qufinite ZX-calculus, which is universal for all finite dimensional Hilbert spaces \cite{qufinite}.  The ZXW calculus  also has recently enjoyed a completeness theorem for qudits \cite{zxw}.  Similarly recent research has established a universal (but not necessarily complete) presentation of the qudit ZH calculus \cite{roy}.  We argue that such research into qudit circuits is not in vein; it can illuminate nontrivial structure, which can potentially be very useful!}

Recently, research into Spekkens' toy model has shed more light on the connection between symplectic geometry and stabilizer circuits.  The qubit version of the model was discovered by Spekkens, where the connection to symplectic geometry was not initially recognized \cite{spekkens}.  Spekkens initially remarked that the toy model had nice properties not satisfied by full quantum theory: for example,  the states and dynamics can be interpreted in terms of  incomplete knowledge about local, noncontextual hidden variables. He observed that this model also has many similarities to quantum theory; however, he proves that the obvious interpretation sending toy bits to qubits isn't sound.   Spekkens' toy model was later formalized in terms of symplectic geometry when it was extended to other qudit dimensions \cite{spekkens2016quasi}. It has since been shown to have the same measurement statistics as odd prime dimensional stabilizer quantum mechanics \cite{catani}.  Even more recently, this correspondence between Spekkens' toy model and stabilizer circuits in odd prime qudit dimensions has been extended to an isomorphism (modulo nonzero scalars), when one regards Spekkens' toy model as the prop of affine Lagrangian relations over $\F_p$ and stabilizer circuits as a monoidal subcategory of Hilbert spaces \cite{lagrel}.  That is to say, in odd prime dimensions, up to scalars, stabilizer circuits and Spekkens' toy model are the same.  This gives a nice categorical semantics for odd prime dimensional qudit stabilizer circuits.  In physical terms,  this reveals that odd prime dimensional qudit stabilizer circuits  are a {\em subtheory} (not merely a toy theory) of quantum theory whose states and dynamics can be interpreted in terms of  incomplete knowledge about local, noncontextual hidden variables.  This yields an sound, universal and complete interpretation of Spekkens' qubit toy model in stabilizer circuits modulo nonzero scalars, sending a toy qubit to two qubits.

Sections \ref{sec:linrel} and \ref{sec:lagrel} give the necessary background needed to understand this correspondence between Spekkens' toy model and stabilizer circuits. Section \ref{sec:linrel} reviews graphical linear and affine algebra.  That is, we review the  graphical presentations linear and affine relations: the monoidal categories whose maps are linear and affine subspaces, respectively.  Section \ref{sec:lagrel} builds on graphical linear algebra and reviews the graphical calculus for linear and affine Lagrangian relations: the monoidal categories whose morphisms are linear and affine Lagrangian subspaces of symplectic vector spaces. This lays the foundation for the rest of the paper, giving a graphical, categorical semantics for Spekkens' toy model and odd prime stabilizer circuits when working over the field $\F_p$, as mentioned in the previous paragraph.

In Section \ref{sec:coisotrel}, we extend the correspondence between affine Lagrangian relations over $\F_p$ and stabilizer circuits to capture quantum mixing.  We show that affine {\em coisotropic} relations over $\F_p$ correspond to stabilizer codes.  Affine coisotropic relations are obtained by adding the discard relation (which acts as the quantum discard for stabilizer circuits) to affine Lagrangian relations.  The interpretation in terms of affine relations/the $X$-gate fragment of the ZX-calculus is as follows (where we draw composition from bottom to top and the tensor from left to right):
$$
\left\llbracket\
\begin{tikzpicture}
	\begin{pgfonlayer}{nodelayer}
		\node [style=none] (0) at (21, 5) {};
		\node [style=none] (1) at (22, 5) {};
		\node [style=none] (2) at (21, 2.5) {};
		\node [style=none] (3) at (22, 2.5) {};
		\node [style=Z] (4) at (21.5, 3.75) {$a,b$};
		\node [style=none] (5) at (21.5, 4.5) {$\cdots$};
		\node [style=none] (6) at (21.5, 3) {$\cdots$};
		\node [style=none] (7) at (21.5, 4.75) {};
		\node [style=none] (8) at (21.5, 2.75) {};
	\end{pgfonlayer}
	\begin{pgfonlayer}{edgelayer}
		\draw [in=150, out=-90, looseness=0.75] (0.center) to (4);
		\draw [in=90, out=-150, looseness=0.75] (4) to (2.center);
		\draw [in=-30, out=90, looseness=0.75] (3.center) to (4);
		\draw [in=-90, out=30, looseness=0.75] (4) to (1.center);
	\end{pgfonlayer}
\end{tikzpicture}
\ \right\rrbracket
=
\begin{tikzpicture}
	\begin{pgfonlayer}{nodelayer}
		\node [style=none] (78) at (402.5, 0.75) {};
		\node [style=none] (79) at (402.5, -2.75) {};
		\node [style=Z] (80) at (402.1, -1.75) {};
		\node [style=none] (81) at (402.145, 0.65) {$\cdots$};
		\node [style=none] (82) at (402.125, -2.65) {$\cdots$};
		\node [style=none] (83) at (400.75, 0.75) {};
		\node [style=none] (84) at (400, -2.75) {};
		\node [style=X] (85) at (400.425, -0.25) {$a$};
		\node [style=none] (86) at (400.375, 0.65) {$\cdots$};
		\node [style=none] (87) at (400.395, -2.625) {$\cdots$};
		\node [style=none] (88) at (400.75, -2.75) {};
		\node [style=none] (89) at (401.75, -2.75) {};
		\node [style=none] (90) at (400, 0.75) {};
		\node [style=none] (91) at (401.75, 0.75) {};
		\node [style=scalar] (92) at (401.25, -1) {$b$};
	\end{pgfonlayer}
	\begin{pgfonlayer}{edgelayer}
		\draw [in=-45, out=90] (79.center) to (80);
		\draw [in=-90, out=60, looseness=0.75] (80) to (78.center);
		\draw [in=90, out=-120, looseness=0.75] (85) to (84.center);
		\draw [in=-90, out=45] (85) to (83.center);
		\draw [in=-60, out=90, looseness=0.75] (88.center) to (85);
		\draw [in=-135, out=90] (89.center) to (80);
		\draw [in=-90, out=150] (85) to (90.center);
		\draw [in=270, out=120, looseness=0.75] (80) to (91.center);
		\draw [in=-75, out=150] (80) to (92);
		\draw [in=330, out=90] (92) to (85);
	\end{pgfonlayer}
\end{tikzpicture},
\hspace*{.3cm}
\left\llbracket \
\begin{tikzpicture}
	\begin{pgfonlayer}{nodelayer}
		\node [style=none] (0) at (21, 5) {};
		\node [style=none] (1) at (22, 5) {};
		\node [style=none] (2) at (21, 2.5) {};
		\node [style=none] (3) at (22, 2.5) {};
		\node [style=X] (4) at (21.5, 3.75) {$a,b$};
		\node [style=none] (5) at (21.5, 4.5) {$\cdots$};
		\node [style=none] (6) at (21.5, 3) {$\cdots$};
		\node [style=none] (7) at (21.5, 4.75) {};
		\node [style=none] (8) at (21.5, 2.75) {};
	\end{pgfonlayer}
	\begin{pgfonlayer}{edgelayer}
		\draw [in=150, out=-90, looseness=0.75] (0.center) to (4);
		\draw [in=90, out=-150, looseness=0.75] (4) to (2.center);
		\draw [in=-30, out=90, looseness=0.75] (3.center) to (4);
		\draw [in=-90, out=30, looseness=0.75] (4) to (1.center);
	\end{pgfonlayer}
\end{tikzpicture}
\ \right\rrbracket
=
\begin{tikzpicture}
	\begin{pgfonlayer}{nodelayer}
		\node [style=none] (93) at (403.5, 0.75) {};
		\node [style=none] (94) at (403.5, -2.75) {};
		\node [style=Z] (95) at (403.9, -1.75) {};
		\node [style=none] (96) at (403.88, 0.65) {$\cdots$};
		\node [style=none] (97) at (403.85, -2.65) {$\cdots$};
		\node [style=none] (98) at (405.25, 0.75) {};
		\node [style=none] (99) at (406, -2.75) {};
		\node [style=X] (100) at (405.575, -0.25) {$a$};
		\node [style=none] (101) at (405.625, 0.65) {$\cdots$};
		\node [style=none] (102) at (405.605, -2.625) {$\cdots$};
		\node [style=none] (103) at (405.25, -2.75) {};
		\node [style=none] (104) at (404.25, -2.75) {};
		\node [style=none] (105) at (406, 0.75) {};
		\node [style=none] (106) at (404.25, 0.75) {};
		\node [style=scalar] (107) at (404.75, -1) {$b$};
	\end{pgfonlayer}
	\begin{pgfonlayer}{edgelayer}
		\draw [in=-135, out=90] (94.center) to (95);
		\draw [in=-90, out=120, looseness=0.75] (95) to (93.center);
		\draw [in=90, out=-60, looseness=0.75] (100) to (99.center);
		\draw [in=-90, out=135] (100) to (98.center);
		\draw [in=-120, out=90, looseness=0.75] (103.center) to (100);
		\draw [in=-45, out=90] (104.center) to (95);
		\draw [in=-90, out=30] (100) to (105.center);
		\draw [in=-90, out=60, looseness=0.75] (95) to (106.center);
		\draw [in=-105, out=30] (95) to (107);
		\draw [in=-150, out=90] (107) to (100);
	\end{pgfonlayer}
\end{tikzpicture}
,
\hspace*{.3cm}
\left\llbracket \
\begin{tikzpicture}
	\begin{pgfonlayer}{nodelayer}
		\node [style=none] (0) at (23, 4.5) {};
		\node [style=none] (1) at (23, 3) {};
		\node [style=scalar] (2) at (23, 3.75) {$a$};
	\end{pgfonlayer}
	\begin{pgfonlayer}{edgelayer}
		\draw (1.center) to (2);
		\draw (2) to (0.center);
	\end{pgfonlayer}
\end{tikzpicture}
\ \right\rrbracket
=
\begin{tikzpicture}
	\begin{pgfonlayer}{nodelayer}
		\node [style=none] (3) at (24.25, 4.5) {};
		\node [style=none] (4) at (24.25, 3) {};
		\node [style=scalarop] (5) at (24.25, 3.75) {$a$};
		\node [style=none] (6) at (25, 4.5) {};
		\node [style=none] (7) at (25, 3) {};
		\node [style=scalar] (8) at (25, 3.75) {$a$};
	\end{pgfonlayer}
	\begin{pgfonlayer}{edgelayer}
		\draw (4.center) to (5);
		\draw (5) to (3.center);
		\draw (7.center) to (8);
		\draw (8) to (6.center);
	\end{pgfonlayer}
\end{tikzpicture}\ ,
\hspace*{.3cm}
\left\llbracket \
\begin{tikzpicture}[yscale=-1]
	\begin{pgfonlayer}{nodelayer}
		\node [style=none] (0) at (0.25, 0) {};
		\node [ground] (1) at (0.25, -0.5) {};
	\end{pgfonlayer}
	\begin{pgfonlayer}{edgelayer}
		\draw (1) to (0.center);
	\end{pgfonlayer}
\end{tikzpicture}\
\right\rrbracket
=
\begin{tikzpicture}
	\begin{pgfonlayer}{nodelayer}
		\node [style=Z] (2) at (4, 0) {};
		\node [style=Z] (3) at (4.5, 0) {};
		\node [style=none] (4) at (4, -1) {};
		\node [style=none] (5) at (4.5, -1) {};
	\end{pgfonlayer}
	\begin{pgfonlayer}{edgelayer}
		\draw (4.center) to (2);
		\draw (3) to (5.center);
	\end{pgfonlayer}
\end{tikzpicture}
$$
Where the spiders have the following interpretations as density matrices:

\

\hfil\scalebox{.9}{$ \displaystyle
\left\llbracket\
\begin{tikzpicture}
	\begin{pgfonlayer}{nodelayer}
		\node [style=none] (0) at (21, 5) {};
		\node [style=none] (1) at (22, 5) {};
		\node [style=none] (2) at (21, 2.5) {};
		\node [style=none] (3) at (22, 2.5) {};
		\node [style=Z] (4) at (21.5, 3.75) {$\hspace*{.05cm}n,m\hspace*{.05cm}$};
		\node [style=none] (5) at (21.5, 4.5) {$\cdots$};
		\node [style=none] (6) at (21.5, 3) {$\cdots$};
		\node [style=none] (7) at (21.5, 4.75) {};
		\node [style=none] (8) at (21.5, 2.75) {};
	\end{pgfonlayer}
	\begin{pgfonlayer}{edgelayer}
		\draw [in=150, out=-90, looseness=0.75] (0.center) to (4);
		\draw [in=90, out=-150, looseness=0.75] (4) to (2.center);
		\draw [in=-30, out=90, looseness=0.75] (3.center) to (4);
		\draw [in=-90, out=30, looseness=0.75] (4) to (1.center);
	\end{pgfonlayer}
\end{tikzpicture}\
\right\rrbracket
\propto 
\sum_{a=0}^{p-1}  e^{\pi\cdot i /p (n\cdot a+m\cdot a^2)}| |a, \ldots, a \rangle\rangle \langle \langle a, \ldots, a||
\ , \ \
\left\llbracket\
\begin{tikzpicture}
	\begin{pgfonlayer}{nodelayer}
		\node [style=none] (0) at (21, 5) {};
		\node [style=none] (1) at (22, 5) {};
		\node [style=none] (2) at (21, 2.5) {};
		\node [style=none] (3) at (22, 2.5) {};
		\node [style=X] (4) at (21.5, 3.75) {$\hspace*{.05cm}n,m\hspace*{.05cm}$};
		\node [style=none] (5) at (21.5, 4.5) {$\cdots$};
		\node [style=none] (6) at (21.5, 3) {$\cdots$};
		\node [style=none] (7) at (21.5, 4.75) {};
		\node [style=none] (8) at (21.5, 2.75) {};
	\end{pgfonlayer}
	\begin{pgfonlayer}{edgelayer}
		\draw [in=150, out=-90, looseness=0.75] (0.center) to (4);
		\draw [in=90, out=-150, looseness=0.75] (4) to (2.center);
		\draw [in=-30, out=90, looseness=0.75] (3.center) to (4);
		\draw [in=-90, out=30, looseness=0.75] (4) to (1.center);
	\end{pgfonlayer}
\end{tikzpicture}\
\right\rrbracket
\propto 
\sum_{a,b=0}^{p-1} 
e^{\pi\cdot i/ p \cdot (n\cdot b-m\cdot b^2)} ||b, \ldots,b \rangle\rangle \langle\langle a, \ldots, a|| 
$}

\

We show that by splitting the decoherence map for the $Z$/$X$ observables we get a calculus for measurement and state preparation. The classical channels then become single wires where the measurement and state preparation are interpreted as injecting and projecting between the single and doubled wires.  The classical wires are drawn coiled for clarity:
$$
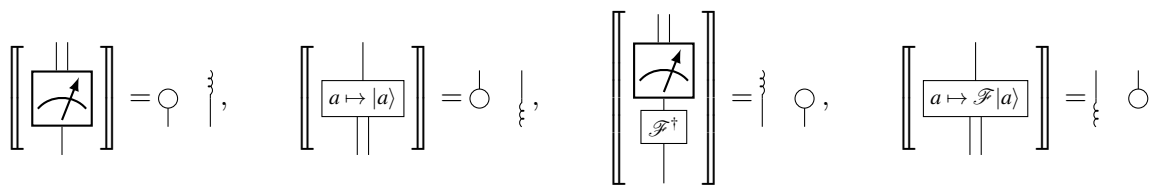
\left\llbracket \
\begin{circuitikz}
\node[meter] (meter) at (0,0) {};
\draw (.1,.5) to (.1,1);
\draw (-.1,.5) to (-.1,1);
\draw (0,-1) to (0,-.5);
\end{circuitikz} \ 
\right\rrbracket 
=

$$
In Example \ref{ex:teleportation}, we prove the correctness of the quantum teleportation algorithm in this graphical calculus using only spider fusion, which the author finds particularly illuminating.  We will preview it here, without explanation:
$$
%
$$
In Section \ref{sec:qec} we apply this formalism to quantum error detection and correction; showing how error detection and affine error correction protocols can be drawn using the string diagrams mentioned in the previous paragraph.  We draw the error correction protocol as follows (of course to be properly understood, one must read the whole paper):
$$
%
$$
In Example \ref{ex:rep} we use the threefold repetition code as a specific example.

\section{Linear/affine relations and the ZX-calculus}
\label{sec:linrel}

In this section, we review the necessary theory to reason diagrammatically about linear  and affine subspaces.  This will allow us to be able to treat stabilizer tableaux in a compositional manner using string diagrams in following sections.  For reference to literature on stabilizer tableaux: see \cite{aaronson} for the ``ordinary'' treatment of  stabilizer tableaux in the qubit setting and \cite{niel} for a generalization to all finite qudit dimensions.

\begin{lemma}
Given a field $k$, the prop of matrices over $k$ under the direct sum, $\Mat_k$, is presented by the following generators and equations, for $a,b \in k$:
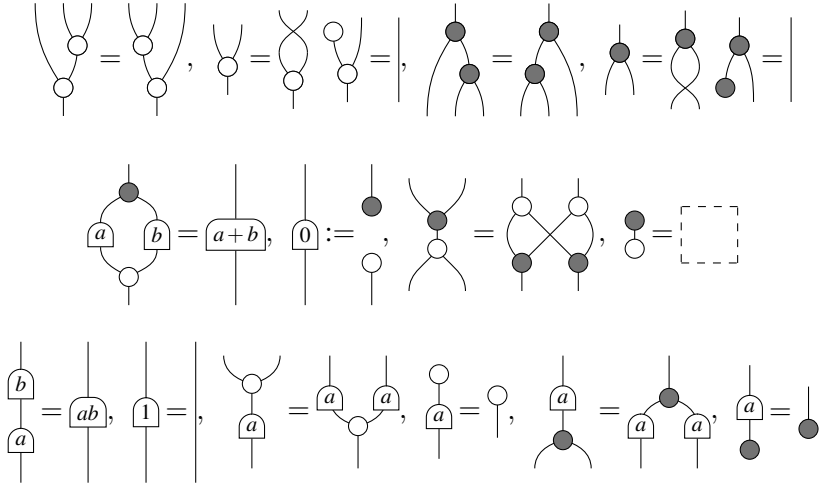
$$
\begin{tikzpicture}
	\begin{pgfonlayer}{nodelayer}
		\node [style=none] (0) at (3.5, -2.75) {};
		\node [style=Z] (1) at (3.5, -2.25) {};
		\node [style=none] (2) at (3, -0.75) {};
		\node [style=Z] (3) at (3.5, -2.25) {};
		\node [style=none] (4) at (3.75, -1.5) {};
		\node [style=Z] (6) at (3.75, -1.5) {};
		\node [style=none] (7) at (3.5, -0.75) {};
		\node [style=Z] (8) at (3.75, -1.5) {};
		\node [style=none] (9) at (4, -0.75) {};
	\end{pgfonlayer}
	\begin{pgfonlayer}{edgelayer}
		\draw (0.center) to (1);
		\draw [in=124, out=-90] (2.center) to (3);
		\draw [in=-90, out=56] (3) to (4.center);
		\draw [in=124, out=-90] (7.center) to (8);
		\draw [in=-90, out=56] (8) to (9.center);
	\end{pgfonlayer}
\end{tikzpicture}
=
\begin{tikzpicture}[xscale=-1]
	\begin{pgfonlayer}{nodelayer}
		\node [style=none] (0) at (3.5, -2.75) {};
		\node [style=Z] (1) at (3.5, -2.25) {};
		\node [style=none] (2) at (3, -0.75) {};
		\node [style=Z] (3) at (3.5, -2.25) {};
		\node [style=none] (4) at (3.75, -1.5) {};
		\node [style=Z] (6) at (3.75, -1.5) {};
		\node [style=none] (7) at (3.5, -0.75) {};
		\node [style=Z] (8) at (3.75, -1.5) {};
		\node [style=none] (9) at (4, -0.75) {};
	\end{pgfonlayer}
	\begin{pgfonlayer}{edgelayer}
		\draw (0.center) to (1);
		\draw [in=124, out=-90] (2.center) to (3);
		\draw [in=-90, out=56] (3) to (4.center);
		\draw [in=124, out=-90] (7.center) to (8);
		\draw [in=-90, out=56] (8) to (9.center);
	\end{pgfonlayer}
\end{tikzpicture},
\hspace*{.2cm}
\begin{tikzpicture}[yscale=-1]
	\begin{pgfonlayer}{nodelayer}
		\node [style=none] (12) at (3.5, -0.75) {};
		\node [style=Z] (15) at (3.5, -1.25) {};
		\node [style=none] (18) at (3.25, -2) {};
		\node [style=Z] (19) at (3.5, -1.25) {};
		\node [style=none] (20) at (3.75, -2) {};
	\end{pgfonlayer}
	\begin{pgfonlayer}{edgelayer}
		\draw (12.center) to (15);
		\draw [in=-124, out=90] (18.center) to (19);
		\draw [in=90, out=-56] (19) to (20.center);
	\end{pgfonlayer}
\end{tikzpicture}
=
\begin{tikzpicture}[yscale=-1]
	\begin{pgfonlayer}{nodelayer}
		\node [style=none] (12) at (3.5, -0.75) {};
		\node [style=Z] (15) at (3.5, -1.25) {};
		\node [style=none] (18) at (3.25, -1.75) {};
		\node [style=Z] (19) at (3.5, -1.25) {};
		\node [style=none] (20) at (3.75, -1.75) {};
		\node [style=none] (21) at (3.75, -2.5) {};
		\node [style=none] (22) at (3.25, -2.5) {};
	\end{pgfonlayer}
	\begin{pgfonlayer}{edgelayer}
		\draw (12.center) to (15);
		\draw [in=-124, out=90] (18.center) to (19);
		\draw [in=90, out=-56] (19) to (20.center);
		\draw [in=270, out=90] (22.center) to (20.center);
		\draw [in=270, out=90] (21.center) to (18.center);
	\end{pgfonlayer}
\end{tikzpicture}
\hspace*{.2cm}
\begin{tikzpicture}
	\begin{pgfonlayer}{nodelayer}
		\node [style=none] (0) at (3.5, -2.25) {};
		\node [style=Z] (1) at (3.5, -1.75) {};
		\node [style=none] (2) at (3.75, -0.75) {};
		\node [style=Z] (3) at (3.5, -1.75) {};
		\node [style=none] (4) at (3.25, -1) {};
		\node [style=Z] (6) at (3.25, -1) {};
	\end{pgfonlayer}
	\begin{pgfonlayer}{edgelayer}
		\draw (0.center) to (1);
		\draw [in=56, out=-90] (2.center) to (3);
		\draw [in=-90, out=124] (3) to (4.center);
	\end{pgfonlayer}
\end{tikzpicture}
=
\begin{tikzpicture}
	\begin{pgfonlayer}{nodelayer}
		\node [style=none] (0) at (3.75, -2.25) {};
		\node [style=none] (2) at (3.75, -0.75) {};
	\end{pgfonlayer}
	\begin{pgfonlayer}{edgelayer}
		\draw (0.center) to (2.center);
	\end{pgfonlayer}
\end{tikzpicture},
\hspace*{.2cm}
\begin{tikzpicture}[yscale=-1]
	\begin{pgfonlayer}{nodelayer}
		\node [style=none] (0) at (3.5, -2.75) {};
		\node [style=X] (1) at (3.5, -2.25) {};
		\node [style=none] (2) at (3, -0.75) {};
		\node [style=X] (3) at (3.5, -2.25) {};
		\node [style=none] (4) at (3.75, -1.5) {};
		\node [style=X] (6) at (3.75, -1.5) {};
		\node [style=none] (7) at (3.5, -0.75) {};
		\node [style=X] (8) at (3.75, -1.5) {};
		\node [style=none] (9) at (4, -0.75) {};
	\end{pgfonlayer}
	\begin{pgfonlayer}{edgelayer}
		\draw (0.center) to (1);
		\draw [in=124, out=-90] (2.center) to (3);
		\draw [in=-90, out=56] (3) to (4.center);
		\draw [in=124, out=-90] (7.center) to (8);
		\draw [in=-90, out=56] (8) to (9.center);
	\end{pgfonlayer}
\end{tikzpicture}
=
\begin{tikzpicture}[xscale=-1,yscale=-1]
	\begin{pgfonlayer}{nodelayer}
		\node [style=none] (0) at (3.5, -2.75) {};
		\node [style=X] (1) at (3.5, -2.25) {};
		\node [style=none] (2) at (3, -0.75) {};
		\node [style=X] (3) at (3.5, -2.25) {};
		\node [style=none] (4) at (3.75, -1.5) {};
		\node [style=X] (6) at (3.75, -1.5) {};
		\node [style=none] (7) at (3.5, -0.75) {};
		\node [style=X] (8) at (3.75, -1.5) {};
		\node [style=none] (9) at (4, -0.75) {};
	\end{pgfonlayer}
	\begin{pgfonlayer}{edgelayer}
		\draw (0.center) to (1);
		\draw [in=124, out=-90] (2.center) to (3);
		\draw [in=-90, out=56] (3) to (4.center);
		\draw [in=124, out=-90] (7.center) to (8);
		\draw [in=-90, out=56] (8) to (9.center);
	\end{pgfonlayer}
\end{tikzpicture},
\hspace*{.2cm}
\begin{tikzpicture}
	\begin{pgfonlayer}{nodelayer}
		\node [style=none] (12) at (3.5, -0.75) {};
		\node [style=X] (15) at (3.5, -1.25) {};
		\node [style=none] (18) at (3.25, -2) {};
		\node [style=X] (19) at (3.5, -1.25) {};
		\node [style=none] (20) at (3.75, -2) {};
	\end{pgfonlayer}
	\begin{pgfonlayer}{edgelayer}
		\draw (12.center) to (15);
		\draw [in=-124, out=90] (18.center) to (19);
		\draw [in=90, out=-56] (19) to (20.center);
	\end{pgfonlayer}
\end{tikzpicture}
=
\begin{tikzpicture}
	\begin{pgfonlayer}{nodelayer}
		\node [style=none] (12) at (3.5, -0.75) {};
		\node [style=X] (15) at (3.5, -1.25) {};
		\node [style=none] (18) at (3.25, -1.75) {};
		\node [style=X] (19) at (3.5, -1.25) {};
		\node [style=none] (20) at (3.75, -1.75) {};
		\node [style=none] (21) at (3.75, -2.5) {};
		\node [style=none] (22) at (3.25, -2.5) {};
	\end{pgfonlayer}
	\begin{pgfonlayer}{edgelayer}
		\draw (12.center) to (15);
		\draw [in=-124, out=90] (18.center) to (19);
		\draw [in=90, out=-56] (19) to (20.center);
		\draw [in=270, out=90] (22.center) to (20.center);
		\draw [in=270, out=90] (21.center) to (18.center);
	\end{pgfonlayer}
\end{tikzpicture}
\hspace*{.2cm}
\begin{tikzpicture}[yscale=-1]
	\begin{pgfonlayer}{nodelayer}
		\node [style=none] (0) at (3.5, -2.25) {};
		\node [style=X] (1) at (3.5, -1.75) {};
		\node [style=none] (2) at (3.75, -0.75) {};
		\node [style=X] (3) at (3.5, -1.75) {};
		\node [style=none] (4) at (3.25, -1) {};
		\node [style=X] (6) at (3.25, -1) {};
	\end{pgfonlayer}
	\begin{pgfonlayer}{edgelayer}
		\draw (0.center) to (1);
		\draw [in=56, out=-90] (2.center) to (3);
		\draw [in=-90, out=124] (3) to (4.center);
	\end{pgfonlayer}
\end{tikzpicture}
=
\begin{tikzpicture}[yscale=-1]
	\begin{pgfonlayer}{nodelayer}
		\node [style=none] (0) at (3.75, -2.25) {};
		\node [style=none] (2) at (3.75, -0.75) {};
	\end{pgfonlayer}
	\begin{pgfonlayer}{edgelayer}
		\draw (0.center) to (2.center);
	\end{pgfonlayer}
\end{tikzpicture}
$$

$$
\begin{tikzpicture}
	\begin{pgfonlayer}{nodelayer}
		\node [style=scalar] (0) at (6, -0.25) {$b$};
		\node [style=scalar] (1) at (5, -0.25) {$a$};
		\node [style=Z] (2) at (5.5, -1) {};
		\node [style=X] (3) at (5.5, 0.5) {};
		\node [style=none] (4) at (5.5, 1) {};
		\node [style=none] (5) at (5.5, -1.5) {};
	\end{pgfonlayer}
	\begin{pgfonlayer}{edgelayer}
		\draw (5.center) to (2);
		\draw [in=-90, out=150] (2) to (1);
		\draw [in=-150, out=90] (1) to (3);
		\draw (3) to (4.center);
		\draw [in=90, out=-30] (3) to (0);
		\draw [in=30, out=-90] (0) to (2);
	\end{pgfonlayer}
\end{tikzpicture}
=
\begin{tikzpicture}
	\begin{pgfonlayer}{nodelayer}
		\node [style=scalar] (1) at (5.5, -0.25) {$a+b$};
		\node [style=none] (4) at (5.5, 1) {};
		\node [style=none] (5) at (5.5, -1.5) {};
	\end{pgfonlayer}
	\begin{pgfonlayer}{edgelayer}
		\draw (5.center) to (1);
		\draw (1) to (4.center);
	\end{pgfonlayer}
\end{tikzpicture},
\hspace*{.2cm}
\begin{tikzpicture}
	\begin{pgfonlayer}{nodelayer}
		\node [style=scalar] (1) at (5.5, -0.25) {$0$};
		\node [style=none] (4) at (5.5, 1) {};
		\node [style=none] (5) at (5.5, -1.5) {};
	\end{pgfonlayer}
	\begin{pgfonlayer}{edgelayer}
		\draw (5.center) to (1);
		\draw (1) to (4.center);
	\end{pgfonlayer}
\end{tikzpicture}
:=
\begin{tikzpicture}
	\begin{pgfonlayer}{nodelayer}
		\node [style=Z] (0) at (5.5, -0.75) {};
		\node [style=X] (1) at (5.5, 0.25) {};
		\node [style=none] (2) at (5.5, 1) {};
		\node [style=none] (3) at (5.5, -1.5) {};
	\end{pgfonlayer}
	\begin{pgfonlayer}{edgelayer}
		\draw (3.center) to (0);
		\draw (1) to (2.center);
	\end{pgfonlayer}
\end{tikzpicture},
\hspace*{.2cm}
\begin{tikzpicture}
	\begin{pgfonlayer}{nodelayer}
		\node [style=none] (4) at (0.5, -2) {};
		\node [style=Z] (5) at (1, -1.25) {};
		\node [style=none] (6) at (1.5, -2) {};
		\node [style=X] (7) at (1, -0.75) {};
		\node [style=none] (8) at (0.5, 0) {};
		\node [style=none] (9) at (1.5, 0) {};
	\end{pgfonlayer}
	\begin{pgfonlayer}{edgelayer}
		\draw [in=-124, out=90] (4.center) to (5);
		\draw [in=90, out=-56] (5) to (6.center);
		\draw (5) to (7);
		\draw [in=-90, out=45] (7) to (9.center);
		\draw [in=-90, out=135] (7) to (8.center);
	\end{pgfonlayer}
\end{tikzpicture}
=
\begin{tikzpicture}
	\begin{pgfonlayer}{nodelayer}
		\node [style=X] (13) at (2.5, -1.5) {};
		\node [style=X] (16) at (3.5, -1.5) {};
		\node [style=Z] (20) at (2.5, -0.5) {};
		\node [style=Z] (23) at (3.5, -0.5) {};
		\node [style=none] (24) at (2.5, 0) {};
		\node [style=none] (25) at (3.5, 0) {};
		\node [style=none] (26) at (2.5, -2) {};
		\node [style=none] (27) at (3.5, -2) {};
	\end{pgfonlayer}
	\begin{pgfonlayer}{edgelayer}
		\draw (16) to (20);
		\draw (23) to (13);
		\draw [bend right=45, looseness=1.25] (16) to (23);
		\draw [bend right=45, looseness=1.25] (20) to (13);
		\draw (26.center) to (13);
		\draw (27.center) to (16);
		\draw (23) to (25.center);
		\draw (20) to (24.center);
	\end{pgfonlayer}
\end{tikzpicture},
\hspace*{.2cm}
\begin{tikzpicture}
	\begin{pgfonlayer}{nodelayer}
		\node [style=Z] (1) at (1, -1.25) {};
		\node [style=X] (3) at (1, -0.75) {};
	\end{pgfonlayer}
	\begin{pgfonlayer}{edgelayer}
		\draw (1) to (3);
	\end{pgfonlayer}
\end{tikzpicture}
=
\begin{tikzpicture}
	\begin{pgfonlayer}{nodelayer}
		\node [style=none] (0) at (2, 0) {};
		\node [style=none] (1) at (2, -1) {};
		\node [style=none] (2) at (3, -1) {};
		\node [style=none] (3) at (3, 0) {};
	\end{pgfonlayer}
	\begin{pgfonlayer}{edgelayer}
		\draw[style=dashed] (3.center) to (0.center);
		\draw[style=dashed] (0.center) to (1.center);
		\draw[style=dashed] (1.center) to (2.center);
		\draw[style=dashed] (2.center) to (3.center);
	\end{pgfonlayer}
\end{tikzpicture}
$$

$$
\begin{tikzpicture}
	\begin{pgfonlayer}{nodelayer}
		\node [style=scalar, fill=white] (6) at (6.5, 0.25) {$b$};
		\node [style=none] (7) at (6.5, 1) {};
		\node [style=none] (8) at (6.5, -1.5) {};
		\node [style=scalar, fill=white] (9) at (6.5, -0.75) {$a$};
	\end{pgfonlayer}
	\begin{pgfonlayer}{edgelayer}
		\draw (8.center) to (6);
		\draw (6) to (7.center);
	\end{pgfonlayer}
\end{tikzpicture}
=
\begin{tikzpicture}
	\begin{pgfonlayer}{nodelayer}
		\node [style=none] (7) at (6.5, 1) {};
		\node [style=none] (8) at (6.5, -1.5) {};
		\node [style=scalar, fill=white] (9) at (6.5, -0.25) {$ab$};
	\end{pgfonlayer}
	\begin{pgfonlayer}{edgelayer}
		\draw (8.center) to (7.center);
	\end{pgfonlayer}
\end{tikzpicture},
\hspace*{.2cm}
\begin{tikzpicture}
	\begin{pgfonlayer}{nodelayer}
		\node [style=none] (7) at (6.5, 1) {};
		\node [style=none] (8) at (6.5, -1.5) {};
		\node [style=scalar, fill=white] (9) at (6.5, -0.25) {$1$};
	\end{pgfonlayer}
	\begin{pgfonlayer}{edgelayer}
		\draw (8.center) to (7.center);
	\end{pgfonlayer}
\end{tikzpicture}=
\begin{tikzpicture}
	\begin{pgfonlayer}{nodelayer}
		\node [style=none] (7) at (6.5, 1) {};
		\node [style=none] (8) at (6.5, -1.5) {};
	\end{pgfonlayer}
	\begin{pgfonlayer}{edgelayer}
		\draw (8.center) to (7.center);
	\end{pgfonlayer}
\end{tikzpicture},
\hspace*{.2cm}
\begin{tikzpicture}
	\begin{pgfonlayer}{nodelayer}
		\node [style=Z] (0) at (6, 0) {};
		\node [style=none] (1) at (5.5, 0.5) {};
		\node [style=none] (2) at (6.5, 0.5) {};
		\node [style=scalar] (3) at (6, -0.75) {$a$};
		\node [style=none] (4) at (6, -1.5) {};
	\end{pgfonlayer}
	\begin{pgfonlayer}{edgelayer}
		\draw (4.center) to (3);
		\draw (3) to (0);
		\draw [in=-90, out=165] (0) to (1.center);
		\draw [in=-90, out=15] (0) to (2.center);
	\end{pgfonlayer}
\end{tikzpicture}
=
\begin{tikzpicture}
	\begin{pgfonlayer}{nodelayer}
		\node [style=Z] (5) at (8, -0.75) {};
		\node [style=none] (6) at (7.5, -0.25) {};
		\node [style=none] (7) at (8.5, -0.25) {};
		\node [style=none] (9) at (8, -1.5) {};
		\node [style=scalar,fill=white] (10) at (8.5, -0.25) {$a$};
		\node [style=scalar,fill=white] (11) at (7.5, -0.25) {$a$};
		\node [style=none] (12) at (7.5, 0.5) {};
		\node [style=none] (13) at (8.5, 0.5) {};
	\end{pgfonlayer}
	\begin{pgfonlayer}{edgelayer}
		\draw [in=-90, out=165] (5) to (6.center);
		\draw [in=-90, out=15] (5) to (7.center);
		\draw (11) to (12.center);
		\draw (10) to (13.center);
		\draw (9.center) to (5);
	\end{pgfonlayer}
\end{tikzpicture},
\hspace*{.2cm}
\begin{tikzpicture}
	\begin{pgfonlayer}{nodelayer}
		\node [style=Z] (14) at (10, 0) {};
		\node [style=scalar] (17) at (10, -0.75) {$a$};
		\node [style=none] (18) at (10, -1.5) {};
	\end{pgfonlayer}
	\begin{pgfonlayer}{edgelayer}
		\draw (18.center) to (17);
		\draw (17) to (14);
	\end{pgfonlayer}
\end{tikzpicture}
=
\begin{tikzpicture}
	\begin{pgfonlayer}{nodelayer}
		\node [style=Z] (19) at (11.25, -0.5) {};
		\node [style=none] (22) at (11.25, -1.25) {};
	\end{pgfonlayer}
	\begin{pgfonlayer}{edgelayer}
		\draw (22.center) to (19);
	\end{pgfonlayer}
\end{tikzpicture},
\hspace*{.2cm}
\begin{tikzpicture}
	\begin{pgfonlayer}{nodelayer}
		\node [style=X] (23) at (12.75, -1) {};
		\node [style=none] (24) at (12.25, -1.5) {};
		\node [style=none] (25) at (13.25, -1.5) {};
		\node [style=scalar] (26) at (12.75, -0.25) {$a$};
		\node [style=none] (27) at (12.75, 0.5) {};
	\end{pgfonlayer}
	\begin{pgfonlayer}{edgelayer}
		\draw (27.center) to (26);
		\draw (26) to (23);
		\draw [in=90, out=-165] (23) to (24.center);
		\draw [in=90, out=-15] (23) to (25.center);
	\end{pgfonlayer}
\end{tikzpicture}
=
\begin{tikzpicture}
	\begin{pgfonlayer}{nodelayer}
		\node [style=X] (28) at (14.75, -0.25) {};
		\node [style=none] (29) at (14.25, -0.75) {};
		\node [style=none] (30) at (15.25, -0.75) {};
		\node [style=none] (31) at (14.75, 0.5) {};
		\node [style=scalar, fill=white] (32) at (15.25, -0.75) {$a$};
		\node [style=scalar, fill=white] (33) at (14.25, -0.75) {$a$};
		\node [style=none] (34) at (14.25, -1.5) {};
		\node [style=none] (35) at (15.25, -1.5) {};
	\end{pgfonlayer}
	\begin{pgfonlayer}{edgelayer}
		\draw [in=90, out=-165] (28) to (29.center);
		\draw [in=90, out=-15] (28) to (30.center);
		\draw (33) to (34.center);
		\draw (32) to (35.center);
		\draw (31.center) to (28);
	\end{pgfonlayer}
\end{tikzpicture},
\hspace*{.2cm}
\begin{tikzpicture}
	\begin{pgfonlayer}{nodelayer}
		\node [style=X] (36) at (16.75, -1) {};
		\node [style=scalar] (37) at (16.75, -0.25) {$a$};
		\node [style=none] (38) at (16.75, 0.5) {};
	\end{pgfonlayer}
	\begin{pgfonlayer}{edgelayer}
		\draw (38.center) to (37);
		\draw (37) to (36);
	\end{pgfonlayer}
\end{tikzpicture}
=
\begin{tikzpicture}
	\begin{pgfonlayer}{nodelayer}
		\node [style=X] (39) at (18, -0.5) {};
		\node [style=none] (40) at (18, 0.25) {};
	\end{pgfonlayer}
	\begin{pgfonlayer}{edgelayer}
		\draw (40.center) to (39);
	\end{pgfonlayer}
\end{tikzpicture}
$$
\end{lemma}

To interpret a diagram as a matrix, label the wires on the bottom and then apply the following rules inductively:
$$
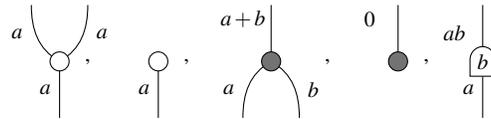
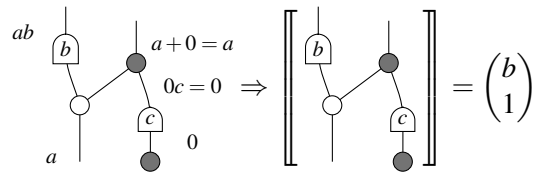
\begin{tikzpicture}
	\begin{pgfonlayer}{nodelayer}
		\node [style=Z] (0) at (-0.5, 0) {};
		\node [style=none] (1) at (-1, 1) {};
		\node [style=none] (2) at (0, 1) {};
		\node [style=none] (3) at (-0.5, -1) {};
		\node [style=X] (4) at (3.25, 0) {};
		\node [style=none] (5) at (2.75, -1) {};
		\node [style=none] (6) at (3.75, -1) {};
		\node [style=none] (7) at (3.25, 1) {};
		\node [style=none] (8) at (-0.75, -0.5) {$a$};
		\node [style=none] (9) at (-1.25, 0.5) {$a$};
		\node [style=none] (10) at (0.25, 0.5) {$a$};
		\node [style=none] (11) at (2.5, -0.5) {$a$};
		\node [style=none] (12) at (4, -0.5) {$b$};
		\node [style=none] (13) at (2.75, 0.75) {$a+b$};
		\node [style=Z] (14) at (1.25, 0) {};
		\node [style=none] (15) at (1.25, -1) {};
		\node [style=none] (16) at (1, -0.5) {$a$};
		\node [style=X] (17) at (5.5, 0) {};
		\node [style=none] (20) at (5.5, 1) {};
		\node [style=none] (23) at (5, 0.75) {$0$};
		\node [style=none] (24) at (0, 0) {,};
		\node [style=none] (25) at (1.75, 0) {,};
		\node [style=none] (26) at (4.25, 0) {,};
		\node [style=none] (27) at (6, 0) {,};
		\node [style=none] (29) at (7, 1) {};
		\node [style=none] (30) at (7, -1) {};
		\node [style=scalar] (31) at (7, 0) {$b$};
		\node [style=none] (32) at (6.75, -0.5) {$a$};
		\node [style=none] (33) at (6.5, 0.5) {$ab$};
	\end{pgfonlayer}
	\begin{pgfonlayer}{edgelayer}
		\draw [in=150, out=-90] (1.center) to (0);
		\draw [in=-90, out=30] (0) to (2.center);
		\draw (0) to (3.center);
		\draw (7.center) to (4);
		\draw [in=90, out=-30] (4) to (6.center);
		\draw [in=90, out=-150] (4) to (5.center);
		\draw (14) to (15.center);
		\draw (20.center) to (17);
		\draw (30.center) to (31);
		\draw (31) to (29.center);
	\end{pgfonlayer}
\end{tikzpicture}
$$

\begin{example}
$$
\begin{tikzpicture}
	\begin{pgfonlayer}{nodelayer}
		\node [style=Z] (0) at (42.75, 0) {};
		\node [style=X] (1) at (43.75, 0.75) {};
		\node [style=none] (2) at (42.75, -1) {};
		\node [style=none] (3) at (44, -0.25) {};
		\node [style=none] (4) at (42.5, 1) {};
		\node [style=none] (5) at (43.75, 1.5) {};
		\node [style=X] (6) at (44, -1) {};
		\node [style=none] (7) at (42.25, -0.95) {$a$};
		\node [style=none] (9) at (44.75, 1.1) {$a+0=a$};
		\node [style=scalar,fill=white] (10) at (42.5, 1) {$b$};
		\node [style=none] (11) at (42.5, 1.75) {};
		\node [style=none] (12) at (41.75, 1.3) {$ab$};
		\node [style=scalar,fill=white] (13) at (44, -0.25) {$c$};
		\node [style=none] (14) at (44.75, -0.65) {$0$};
		\node [style=none] (15) at (44.75, 0.35) {$0c=0$};
	\end{pgfonlayer}
	\begin{pgfonlayer}{edgelayer}
		\draw (2.center) to (0);
		\draw [in=285, out=90] (3.center) to (1);
		\draw (1) to (5.center);
		\draw (0) to (1);
		\draw [in=-90, out=105] (0) to (4.center);
		\draw (10) to (11.center);
		\draw (6) to (13);
	\end{pgfonlayer}
\end{tikzpicture}
\Rightarrow
\left\llbracket\
\begin{tikzpicture}
	\begin{pgfonlayer}{nodelayer}
		\node [style=Z] (0) at (42.75, 0) {};
		\node [style=X] (1) at (43.75, 0.75) {};
		\node [style=none] (2) at (42.75, -1) {};
		\node [style=none] (3) at (44, -0.25) {};
		\node [style=none] (4) at (42.5, 1) {};
		\node [style=none] (5) at (43.75, 1.5) {};
		\node [style=X] (6) at (44, -1) {};
		\node [style=scalar,fill=white] (10) at (42.5, 1) {$b$};
		\node [style=none] (11) at (42.5, 1.75) {};
		\node [style=scalar,fill=white] (13) at (44, -0.25) {$c$};
	\end{pgfonlayer}
	\begin{pgfonlayer}{edgelayer}
		\draw (2.center) to (0);
		\draw [in=285, out=90] (3.center) to (1);
		\draw (1) to (5.center);
		\draw (0) to (1);
		\draw [in=-90, out=105] (0) to (4.center);
		\draw (10) to (11.center);
		\draw (6) to (13);
	\end{pgfonlayer}
\end{tikzpicture}
\ \right\rrbracket
=
\begin{pmatrix}
b\\
1
\end{pmatrix}
$$
\end{example}

There is a very much underappreciated category whose maps are not matrices, but rather, linear subspaces; the exposition of which is crucial for this paper:

\begin{definition}
Given a field $k$, the $\dag$-compact closed prop of linear relations over $k$, $\LinRel_{k}$ has:

\begin{description}
\item[\ \ Objects:] Natural numbers.

\item[\ \ Maps:] A linear relation $n\to m$ is a linear subspace of $k^n \oplus k^m$.

\item[\ \ Composition:] Relational composition, so that for $R \subseteq k^n \oplus k^m$  and $S \subseteq k^m \oplus k^\ell$:
$$
SR := \{  (x,z) \in k^{n} \oplus k^{\ell} : \exists y \in k^{m}, (x,y) \in R \wedge (y,z) \in S \} \subseteq k^n \oplus k^\ell
$$ 

\item[\ \ Tensor product:] Direct sum, so that for $R \subseteq k^n \oplus k^m$ and $S \subseteq k^\ell \oplus k^q$:

$$R\oplus S : =
\left\{
\left(
\begin{pmatrix}
a_1\\a_2
\end{pmatrix},
\begin{pmatrix}
b_1\\b_2
\end{pmatrix}
:
\forall (a_1,b_1) \in R, (a_2,b_2) \in S
\right)
\right\} \subseteq k^{n+\ell}\oplus k^{m+q}
$$

\item[\ \ Dagger:] Relational converse, so that for $R \subseteq k^{n}\oplus k^m$:

$$
R^T := \{ (b,a) : \forall (a,b) \in R \} \subseteq k^{m} \oplus k^n
$$
\end{description}
\end{definition}

\begin{lemma}[\cite{ihpub}]
Given a field $k$, $\LinRel_{k}$ is generated by the generators and equations of the presentation of $\Mat_k$ as well as those of $\Mat_k^{\op}$ (drawn as the vertically flipped generators of $\Mat_k$) modulo the equations for all $a \in k$, $a\neq 0$:

$$
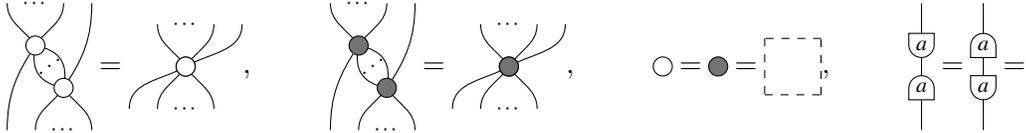
\begin{tikzpicture}
	\begin{pgfonlayer}{nodelayer}
		\node [style=none] (0) at (1.5, -0.5) {};
		\node [style=none] (1) at (0.5, -0.5) {};
		\node [style=none] (2) at (1, -0.5) {$\cdots$};
		\node [style=none] (3) at (0.5, -2.75) {};
		\node [style=Z] (4) at (1, -1.25) {};
		\node [style=none] (5) at (2, -0.5) {};
		\node [style=none] (6) at (1.5, -2.75) {$\cdots$};
		\node [style=none] (7) at (1, -2.75) {};
		\node [style=Z] (8) at (1.5, -2) {};
		\node [style=none] (9) at (2, -2.75) {};
		\node [style=none] (10) at (1.25, -1.5) {\reflectbox{$\ddots$}};
	\end{pgfonlayer}
	\begin{pgfonlayer}{edgelayer}
		\draw [in=-124, out=90] (3.center) to (4);
		\draw [in=-90, out=56] (4) to (0.center);
		\draw [in=124, out=-90] (1.center) to (4);
		\draw [in=-124, out=90] (7.center) to (8);
		\draw [in=90, out=-56] (8) to (9.center);
		\draw [in=-90, out=56] (8) to (5.center);
		\draw [bend left=45, looseness=1.25] (8) to (4);
		\draw [bend left=45, looseness=1.25] (4) to (8);
	\end{pgfonlayer}
\end{tikzpicture}
=
\begin{tikzpicture}
	\begin{pgfonlayer}{nodelayer}
		\node [style=none] (11) at (4, -0.5) {};
		\node [style=none] (12) at (3, -0.5) {};
		\node [style=none] (13) at (3.5, -0.5) {$\cdots$};
		\node [style=none] (14) at (2.5, -2) {};
		\node [style=Z] (15) at (3.5, -1.25) {};
		\node [style=none] (16) at (4.5, -0.5) {};
		\node [style=none] (17) at (3.5, -2) {$\cdots$};
		\node [style=none] (18) at (3, -2) {};
		\node [style=Z] (19) at (3.5, -1.25) {};
		\node [style=none] (20) at (4, -2) {};
	\end{pgfonlayer}
	\begin{pgfonlayer}{edgelayer}
		\draw [in=-150, out=90] (14.center) to (15);
		\draw [in=-90, out=56] (15) to (11.center);
		\draw [in=124, out=-90] (12.center) to (15);
		\draw [in=-124, out=90] (18.center) to (19);
		\draw [in=90, out=-56] (19) to (20.center);
		\draw [in=-90, out=30] (19) to (16.center);
	\end{pgfonlayer}
\end{tikzpicture},
\hspace*{1cm}
\begin{tikzpicture}
	\begin{pgfonlayer}{nodelayer}
		\node [style=none] (0) at (1.5, -0.5) {};
		\node [style=none] (1) at (0.5, -0.5) {};
		\node [style=none] (2) at (1, -0.5) {$\cdots$};
		\node [style=none] (3) at (0.5, -2.75) {};
		\node [style=X] (4) at (1, -1.25) {};
		\node [style=none] (5) at (2, -0.5) {};
		\node [style=none] (6) at (1.5, -2.75) {$\cdots$};
		\node [style=none] (7) at (1, -2.75) {};
		\node [style=X] (8) at (1.5, -2) {};
		\node [style=none] (9) at (2, -2.75) {};
		\node [style=none] (10) at (1.25, -1.5) {\reflectbox{$\ddots$}};
	\end{pgfonlayer}
	\begin{pgfonlayer}{edgelayer}
		\draw [in=-124, out=90] (3.center) to (4);
		\draw [in=-90, out=56] (4) to (0.center);
		\draw [in=124, out=-90] (1.center) to (4);
		\draw [in=-124, out=90] (7.center) to (8);
		\draw [in=90, out=-56] (8) to (9.center);
		\draw [in=-90, out=56] (8) to (5.center);
		\draw [bend left=45, looseness=1.25] (8) to (4);
		\draw [bend left=45, looseness=1.25] (4) to (8);
	\end{pgfonlayer}
\end{tikzpicture}
=
\begin{tikzpicture}
	\begin{pgfonlayer}{nodelayer}
		\node [style=none] (11) at (4, -0.5) {};
		\node [style=none] (12) at (3, -0.5) {};
		\node [style=none] (13) at (3.5, -0.5) {$\cdots$};
		\node [style=none] (14) at (2.5, -2) {};
		\node [style=X] (15) at (3.5, -1.25) {};
		\node [style=none] (16) at (4.5, -0.5) {};
		\node [style=none] (17) at (3.5, -2) {$\cdots$};
		\node [style=none] (18) at (3, -2) {};
		\node [style=X] (19) at (3.5, -1.25) {};
		\node [style=none] (20) at (4, -2) {};
	\end{pgfonlayer}
	\begin{pgfonlayer}{edgelayer}
		\draw [in=-150, out=90] (14.center) to (15);
		\draw [in=-90, out=56] (15) to (11.center);
		\draw [in=124, out=-90] (12.center) to (15);
		\draw [in=-124, out=90] (18.center) to (19);
		\draw [in=90, out=-56] (19) to (20.center);
		\draw [in=-90, out=30] (19) to (16.center);
	\end{pgfonlayer}
\end{tikzpicture},
\hspace*{1cm}
\begin{tikzpicture}
	\begin{pgfonlayer}{nodelayer}
		\node [style=Z] (0) at (3.75, -1) {};
	\end{pgfonlayer}
\end{tikzpicture}
=
\begin{tikzpicture}
	\begin{pgfonlayer}{nodelayer}
		\node [style=X] (0) at (3.75, -1) {};
	\end{pgfonlayer}
\end{tikzpicture}
=
\begin{tikzpicture}
	\begin{pgfonlayer}{nodelayer}
		\node [style=none] (0) at (2, 0) {};
		\node [style=none] (1) at (2, -1) {};
		\node [style=none] (2) at (3, -1) {};
		\node [style=none] (3) at (3, 0) {};
	\end{pgfonlayer}
	\begin{pgfonlayer}{edgelayer}
		\draw[style=dashed] (3.center) to (0.center);
		\draw[style=dashed] (0.center) to (1.center);
		\draw[style=dashed] (1.center) to (2.center);
		\draw[style=dashed] (2.center) to (3.center);
	\end{pgfonlayer}
\end{tikzpicture},
\hspace*{1cm}
\begin{tikzpicture}
	\begin{pgfonlayer}{nodelayer}
		\node [style=none] (3) at (17, 1.5) {};
		\node [style=none] (4) at (17, -0.75) {};
		\node [style=scalarop] (5) at (17, 0.75) {$a$};
		\node [style=scalar] (6) at (17, 0) {$a$};
	\end{pgfonlayer}
	\begin{pgfonlayer}{edgelayer}
		\draw (4.center) to (6);
		\draw (6) to (5);
		\draw (5) to (3.center);
	\end{pgfonlayer}
\end{tikzpicture}
=
\begin{tikzpicture}
	\begin{pgfonlayer}{nodelayer}
		\node [style=none] (3) at (17, 1.5) {};
		\node [style=none] (4) at (17, -0.75) {};
		\node [style=scalar] (5) at (17, 0.75) {$a$};
		\node [style=scalarop] (6) at (17, 0) {$a$};
	\end{pgfonlayer}
	\begin{pgfonlayer}{edgelayer}
		\draw (4.center) to (6);
		\draw (6) to (5);
		\draw (5) to (3.center);
	\end{pgfonlayer}
\end{tikzpicture}
=
\begin{tikzpicture}
	\begin{pgfonlayer}{nodelayer}
		\node [style=none] (3) at (17, 1.5) {};
		\node [style=none] (4) at (17, -0.75) {};
	\end{pgfonlayer}
	\begin{pgfonlayer}{edgelayer}
		\draw (4.center) to (3.center);
	\end{pgfonlayer}
\end{tikzpicture}
$$
\end{lemma}

To interpret a string diagram as a linear subspaces, we must do the same as before; this time labeling the bottom and top wires and generating equations when the labels meet.  These equations carve out the linear subspace:

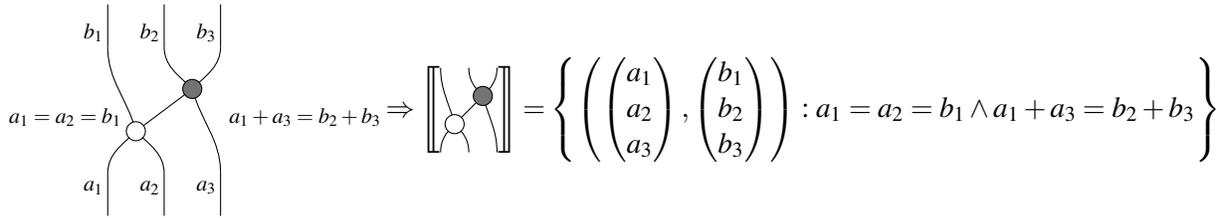
\begin{example}
$$
\begin{tikzpicture}
	\begin{pgfonlayer}{nodelayer}
		\node [style=Z] (0) at (42.75, 0) {};
		\node [style=X] (1) at (43.75, 0.75) {};
		\node [style=none] (2) at (42.25, -0.75) {};
		\node [style=none] (3) at (43.25, -0.75) {};
		\node [style=none] (4) at (44.25, -0.75) {};
		\node [style=none] (5) at (42.25, 1.5) {};
		\node [style=none] (6) at (43.25, 1.5) {};
		\node [style=none] (7) at (44.25, 1.5) {};
		\node [style=none] (8) at (42.25, 2.25) {};
		\node [style=none] (9) at (43.25, 2.25) {};
		\node [style=none] (10) at (44.25, 2.25) {};
		\node [style=none] (11) at (42.25, -1.5) {};
		\node [style=none] (12) at (43.25, -1.5) {};
		\node [style=none] (13) at (44.25, -1.5) {};
		\node [style=none] (14) at (42, -1) {$a_1$};
		\node [style=none] (15) at (43, -1) {$a_2$};
		\node [style=none] (16) at (44, -1) {$a_3$};
		\node [style=none] (17) at (42, 1.75) {$b_1$};
		\node [style=none] (18) at (43, 1.75) {$b_2$};
		\node [style=none] (19) at (44, 1.75) {$b_3$};
		\node [style=none] (20) at (41.5, 0.25) {$a_1=a_2=b_1$};
		\node [style=none] (21) at (45.75, 0.25) {$a_1+a_3=b_2+b_3$};
	\end{pgfonlayer}
	\begin{pgfonlayer}{edgelayer}
		\draw [in=-135, out=90] (2.center) to (0);
		\draw [in=90, out=-45] (0) to (3.center);
		\draw [in=285, out=90] (4.center) to (1);
		\draw [in=-90, out=135] (1) to (6.center);
		\draw [in=-90, out=45] (1) to (7.center);
		\draw (0) to (1);
		\draw [in=-90, out=105] (0) to (5.center);
		\draw (5.center) to (8.center);
		\draw (6.center) to (9.center);
		\draw (7.center) to (10.center);
		\draw (13.center) to (4.center);
		\draw (12.center) to (3.center);
		\draw (11.center) to (2.center);
	\end{pgfonlayer}
\end{tikzpicture}
\Rightarrow
\left\llbracket
\begin{tikzpicture}
	\begin{pgfonlayer}{nodelayer}
		\node [style=Z] (0) at (42.75, 0.25) {};
		\node [style=X] (1) at (43.25, 0.75) {};
		\node [style=none] (2) at (42.5, -0.25) {};
		\node [style=none] (3) at (43, -0.25) {};
		\node [style=none] (4) at (43.5, -0.25) {};
		\node [style=none] (5) at (42.5, 1.25) {};
		\node [style=none] (6) at (43, 1.25) {};
		\node [style=none] (7) at (43.5, 1.25) {};
	\end{pgfonlayer}
	\begin{pgfonlayer}{edgelayer}
		\draw [in=-135, out=90] (2.center) to (0);
		\draw [in=90, out=-45] (0) to (3.center);
		\draw [in=285, out=90] (4.center) to (1);
		\draw [in=-90, out=135] (1) to (6.center);
		\draw [in=-90, out=45] (1) to (7.center);
		\draw (0) to (1);
		\draw [in=-90, out=105] (0) to (5.center);
	\end{pgfonlayer}
\end{tikzpicture}
\right\rrbracket
=
\left\{
\left(
\begin{pmatrix}
           a_{1} \\
           a_{2} \\
           a_{3}
\end{pmatrix}
,
\begin{pmatrix}
           b_{1} \\
           b_{2} \\
           b_{3}
\end{pmatrix}
\right)
:
a_1=a_2=b_1\wedge
a_1+a_3 = b_2+b_3
\right\}
$$
\end{example}

The astute reader may recognize this structure from another setting:

\begin{definition}
Given fixed qudit dimension $d$, the phase-free ZX-calculus is, modulo scalars, generated by two unlabeled spiders, interpreted as the following complex-valued matrices:

$$
\left\llbracket\ 
\begin{tikzpicture}
	\begin{pgfonlayer}{nodelayer}
		\node [style=none] (0) at (4, -0.5) {};
		\node [style=none] (1) at (3, -0.5) {};
		\node [style=none] (2) at (3.5, -0.75) {$\cdots$};
		\node [style=Z] (4) at (3.5, -1.25) {};
		\node [style=none] (6) at (3.5, -1.75) {$\cdots$};
		\node [style=none] (7) at (3, -2) {};
		\node [style=Z] (8) at (3.5, -1.25) {};
		\node [style=none] (9) at (4, -2) {};
		\node [style=none] (10) at (3.5, -2) {$n$};
		\node [style=none] (11) at (3.5, -0.5) {$m$};
	\end{pgfonlayer}
	\begin{pgfonlayer}{edgelayer}
		\draw [in=-90, out=56] (4) to (0.center);
		\draw [in=124, out=-90] (1.center) to (4);
		\draw [in=-124, out=90] (7.center) to (8);
		\draw [in=90, out=-56] (8) to (9.center);
	\end{pgfonlayer}
\end{tikzpicture}
\ \right\rrbracket
\propto
\sum_{a=0}^{p-1} | a, \ldots, a\rangle \langle a,\ldots, a|,
\hspace*{1cm} 
\left\llbracket\ 
\begin{tikzpicture}
	\begin{pgfonlayer}{nodelayer}
		\node [style=none] (0) at (4, -0.5) {};
		\node [style=none] (1) at (3, -0.5) {};
		\node [style=none] (2) at (3.5, -0.75) {$\cdots$};
		\node [style=X] (4) at (3.5, -1.25) {};
		\node [style=none] (6) at (3.5, -1.75) {$\cdots$};
		\node [style=none] (7) at (3, -2) {};
		\node [style=none] (8) at (3.5, -1.25) {};
		\node [style=none] (9) at (4, -2) {};
		\node [style=none] (10) at (3.5, -2) {$n$};
		\node [style=none] (11) at (3.5, -0.5) {$m$};
	\end{pgfonlayer}
	\begin{pgfonlayer}{edgelayer}
		\draw [in=-90, out=56] (4) to (0.center);
		\draw [in=124, out=-90] (1.center) to (4);
		\draw [in=-124, out=90] (7.center) to (8);
		\draw [in=90, out=-56] (8) to (9.center);
	\end{pgfonlayer}
\end{tikzpicture}
\ \right\rrbracket
\propto
\sum_{\sum  a_i = \sum b_j \mod p} | b_1 ,\ldots, b_m \rangle \langle  a_1,\ldots, a_n|
$$
\end{definition}

\begin{definition}
A unitary map $f:\mathcal{H}\to \mathcal{H}$ is a {\bf stabilizer} of a state $|\phi\rangle$ on $\mathcal H$ in case $|\phi\rangle$ is a +1-eigenvector of $f$ so that $f| \phi\rangle = |\phi \rangle$.

The qudit ${\cal X}$-gate (qubit $\Not$-gate)  shifts the X-basis vectors by $a$ mod $d$:
$$
{\cal X} := \sum_{b=0}^{d-1} | b+1\rangle \langle b|
\ , \hspace*{.2cm} \text{where \ \ }
{\cal X}^a = \sum_{b=0}^{d-1} | b+a\rangle \langle b|
$$

Similarly, the  qudit ${\cal Z}$-gate  shifts the Z-basis vectors by $a$ mod $d$:
$$
{\cal Z}
:=
{\cal F}
{\cal X}
{\cal F}^\dag
=
\sum_{b=0}^{d-1}
e^{2\cdot\pi \cdot b/d} | b\rangle \langle b|
\ , \hspace*{.2cm} \text{where \ \ }
{\cal Z}^z
=
{\cal F}
{\cal X}^z
{\cal F}^\dag
=
\sum_{b=0}^{d-1}
e^{2\cdot\pi \cdot z\cdot b/d} | b\rangle \langle b|
$$

An $X$-stabilizer of an $n$-qudit state $\phi$ is a stabilizer of the form:
$$
{\cal X}^{a_0}
\otimes  {\cal X}^{a_1}
\otimes 
\cdots 
\otimes {\cal X}^{a_{n-2}}
\otimes {\cal X}^{a_{n-1}}
$$
Similarly, a $Z$-stabilizer of an $n$-qudit state $\phi$ is a stabilizer of the form:
$$
{\cal Z}^{a_0}
\otimes  {\cal Z}^{a_1} 
\otimes 
\cdots
\otimes {\cal Z}^{a_{n-2}}
\otimes {\cal Z}^{a_{n-1}}
$$

Where the Fourier transform is defined as follows:
$$\mathcal{F} := \dfrac{1}{\sqrt{d}} \sum_{j,k=0}^{d-1} e^{2\pi\cdot i \cdot j \cdot k/d} | k\rangle \langle j | $$
\end{definition}

\begin{lemma}[\cite{cole}]
Given an odd prime $p$, $\LinRel_{\F_p}$ is isomorphic to the $p$-dimensional qudit phase-free ZX-calculus modulo invertible scalars.
\end{lemma}
\begin{proof}
Given $p$-dimensional qudit phase-free ZX-diagram it is easy to see how the $X$-stabilizers form a linear subspace over $\F_p$ as follows:
$$
\left\llbracket
D
\right\rrbracket_X
:=
\left\{ 
\left(

 \right\}
$$
Conversely, given an $\F_p$-linear subspace, take the projector onto the  joint $+1$-eigenspace spanned by the corresponding $X$-stabilizers:
$$
S\mapsto \sum_{(x_1,\cdots, x_n) \in S} \dfrac{1}{|S|} | x_1, \cdots x_n \rangle\langle x_1, \cdots, x_n  |
$$
Regarding this as an $n$-qudit state in  $\CPM(\FHilb)$ (of local dimension $p$), partition the codomain of the state into an input and output.  Bending the input wires down  yields an inverse to the previous mapping.
\end{proof}

\begin{example}
Consider the following phase-free ZX-diagram: \ \ 
$
%
$

Its $X$ stabilizers are parameterized by all the  $a_1,a_2,a_3,b_1,b_2,b_3 \in \F_p$ such that:
$$
%
$$
By labeling the wires with linear equations over $\F_p$, we can calculate these stabilizers:
$$
%
$$
Which gives us a linear subspace of $\F_p^{3} \oplus \F_p^3$:
$$
\left\llbracket
%
\right)
: a_1,a_2,a_3,b_1,b_2,b_3 \in \F_p,
a_1=a_2=b_1\wedge
a_1+a_3 = b_2+b_3
\right\}
$$
\end{example}

Before we continue further, one should observe that it matters which spider the wire is bent with respect to. If the wire is bent with one colour spider and then unbent with another wire, this introduces the scalar $-1$.  This is unfamiliar to many experienced ZX-calculus practitioners who work exclusively in qubits as $-1 \equiv 0 \mod 2$.

\begin{lemma} \hfil $
%
$

\end{lemma}

We could have instead identified phase-free $\ZX$-diagrams with their $Z$-stabilizers.  This can be performed by taking the Fourier transform by colour swapping the phase-free diagram and then calculating its $X$-stabilizers. 

\begin{example}

The $Z$-stabilizers of the diagram in the previous example are parameterized by all the $a_1',a_2',a_3',b_1',b_2',b_3' \in \F_p$ such that:
$$
%
$$
We have to negate the powers of $X$ on the bottom, because stabilizers are defined on states. Therefore, when we turn the following process into a state, the $Z$ stabilizers on the input wires  are inverted.
Again, we label the wires to produce a collection of linear equations which parameterize the stabilizers:
$$
%
$$
This determines a linear subspace of $\F_p^3\oplus \F_p^3$:

\hfil\scalebox{.9}{$
\left\llbracket
%
\right)
: a_1',a_2',a_3',b_1',b_2',b_3' \in \F_p,
b_2'=b_3'=-a_3'\wedge
-a_1'-a_2'=b_1'+b_2'
\right\}
$}

\end{example}

The $Z$ and $X$ stabilizers of a phase-free diagram determine each other.  This is because the Fourier transform of a phase-free ZX-diagram corresponds to the orthogonal complement of the corresponding linear subspace.

\begin{lemma}[\cite{ortho}]
The colour swapping functor sending:
$$
%
$$
is the isomorphism $\LinRel_{k}\cong\LinRel_{k}$ which takes linear subspaces $V \subseteq k^n$ to their orthogonal complement:
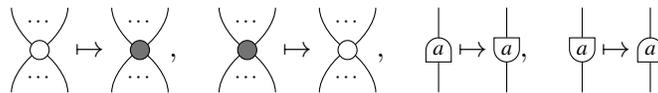
$$
V \mapsto V^\perp := \{  v \in k^n: \forall w \in V ,  v^Tw = 0\}
$$
Notice that the orthogonal complement reverses the order of inclusion of linear subspaces so that $V \subseteq W \iff W^\perp \subseteq V^\perp$ and thus it is involution so that $(V^\perp)^\perp = V$.
\end{lemma}

\begin{example}
Since we know the $Z$ and $X$ stabilizers of our running example, we can check that the two subspaces which they determine are orthogonal to each other.

Take

$$v=\left(
\begin{pmatrix}
           a_{1} \\
           a_{2} \\
           a_{3}
\end{pmatrix}
,
\begin{pmatrix}
           b_{1} \\
           b_{2} \\
           b_{3}
\end{pmatrix}
\right) \in \left\llbracket
\begin{tikzpicture}
	\begin{pgfonlayer}{nodelayer}
		\node [style=Z] (0) at (42.75, 0.25) {};
		\node [style=X] (1) at (43.25, 0.75) {};
		\node [style=none] (2) at (42.5, -0.25) {};
		\node [style=none] (3) at (43, -0.25) {};
		\node [style=none] (4) at (43.5, -0.25) {};
		\node [style=none] (5) at (42.5, 1.25) {};
		\node [style=none] (6) at (43, 1.25) {};
		\node [style=none] (7) at (43.5, 1.25) {};
	\end{pgfonlayer}
	\begin{pgfonlayer}{edgelayer}
		\draw [in=-135, out=90] (2.center) to (0);
		\draw [in=90, out=-45] (0) to (3.center);
		\draw [in=285, out=90] (4.center) to (1);
		\draw [in=-90, out=135] (1) to (6.center);
		\draw [in=-90, out=45] (1) to (7.center);
		\draw (0) to (1);
		\draw [in=-90, out=105] (0) to (5.center);
	\end{pgfonlayer}
\end{tikzpicture}
\right\rrbracket_X,
\hspace*{1cm}
w=\left(
\begin{pmatrix}
           a_{1}' \\
           a_{2}' \\
           a_{3}'
\end{pmatrix}
,
\begin{pmatrix}
           b_{1}' \\
           b_{2}' \\
           b_{3}'
\end{pmatrix}
\right) \in \left\llbracket
\begin{tikzpicture}
	\begin{pgfonlayer}{nodelayer}
		\node [style=Z] (0) at (42.75, 0.25) {};
		\node [style=X] (1) at (43.25, 0.75) {};
		\node [style=none] (2) at (42.5, -0.25) {};
		\node [style=none] (3) at (43, -0.25) {};
		\node [style=none] (4) at (43.5, -0.25) {};
		\node [style=none] (5) at (42.5, 1.25) {};
		\node [style=none] (6) at (43, 1.25) {};
		\node [style=none] (7) at (43.5, 1.25) {};
	\end{pgfonlayer}
	\begin{pgfonlayer}{edgelayer}
		\draw [in=-135, out=90] (2.center) to (0);
		\draw [in=90, out=-45] (0) to (3.center);
		\draw [in=285, out=90] (4.center) to (1);
		\draw [in=-90, out=135] (1) to (6.center);
		\draw [in=-90, out=45] (1) to (7.center);
		\draw (0) to (1);
		\draw [in=-90, out=105] (0) to (5.center);
	\end{pgfonlayer}
\end{tikzpicture}
\right\rrbracket_Z$$

So that
$$
a_1=a_2=b_1,\
a_1+a_3 = b_2+b_3,\ 
b_2'=b_3'=-a_3',\
-a_1'-a_2'=b_1'+b_2'
$$

Then
\begin{align*}
  v^Tw
=& a_1a_1' + a_2a_2'+a_3a_3'+b_1b_1' + b_2b_2'+b_3b_3'\\
=& a_1a_1' + a_1a_2'-a_3b_2'+a_1b_1' + b_2b_2'+b_3b_2'\\
=& a_1a_1' + a_1a_2'-(b_2+b_3-a_1)b_2'+a_1(-a_1'-a_2'-b_2') + b_2b_2'+b_3b_2'\\
=& a_1a_1' + a_1a_2'-b_2b_2'-b_3b_2'+a_1b_2'-a_1a_1'-a_1a_2'-a_1b_2' + b_2b_2'+b_3b_2'\\
=& (a_1a_1'-a_1a_1') + (a_1a_2'-a_1a_2')+(b_2b_2'-b_2b_2')+(b_3b_2'-b_3b_2')+(a_1b_2'-a_1b_2')\\
=&0+0+0+0\\
=0
\end{align*}
\end{example}

We can get an even larger fragment of the ZX-calculus by looking at affine subspaces.  First, we shall look at affine matrices:

\begin{definition}
Given a field $k$, the prop of affine matrices, $\Aff\Mat_k$ is presented by adding the affine shift to the presentation of $\Mat_k$, modulo the following equations
$$

$$
\end{definition}
These equations state that the affine shift can be copied and discarded.

To translate the string diagrams to an affine matrix, we do the same as before, but this time, the new affine shift generator introduces a $1$ (as opposed to its unlabeled counterpart which introduces a $0$).

\begin{example}

$$
%
$$
\end{example}

This leads us to the notion of affine relations:

\begin{definition}
The prop of affine relations over $k$, $\Aff\Rel_{k}$ is constructed in the same way as $\LinRel_k$ except a map $n\to m$ is instead a (possibly empty) affine subspace of $k^n\oplus k^m$.
\end{definition}

This has a string diagrammatic presentation:

\begin{lemma}[\cite{affine}]
$\Aff\Rel_{k}$  is generated by adding an extra generator to $\LinRel_{k}$ modulo the following equations:
$$
%
$$

\end{lemma}
The first equations corresponds to the multiplication by zero (imposed by the the equation $1=0$) and the other two come directly from the presentation of affine matrices.

The phased spiders are used for syntactic sugar:

$$
%
$$

To calculate the affine subspaces corresponding to a string diagram, we perform the same process as for linear relations, but this time, we generate a set of affine equations which defines the affine subspace.
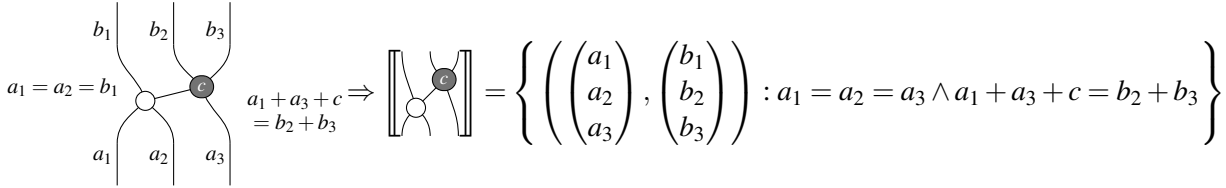
\begin{example}

$$
%
\right)
:
a_1=a_2=a_3\wedge
a_1+a_3+c = b_2+b_3
\right\}
$$
\end{example}

The following example is also very distinct from the linear setting

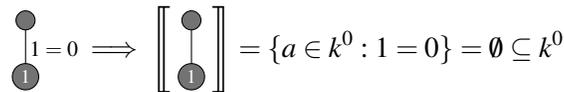
\begin{example}
$$
%
\ \right\rrbracket
=
\{a \in k^0 : 1=0\}
=
\emptyset
\subseteq k^0
$$
\end{example}

From this phased-spider presentation, a ZX-calculus practitioner wouldn't be too surprised that:

\begin{lemma}
\label{lem:affrel}
$\Aff\Rel_{\F_p}$ is isomorphic to the $p$-dimensional qudit ZX-calculus with Pauli $X$ gates modulo invertible scalars.
\end{lemma}
This is given by the interpretation:
$$
\left\llbracket\ 
%
\ \right\rrbracket
\propto
\sum_{c+\sum  a_i = \sum b _j \mod p} | b_1 ,\ldots, b_m\rangle \langle  a_1,\ldots, a_n|
$$

The proof is almost identical to that for linear relations and phase-free ZX-diagrams.

\begin{example}
Consider the following phase-free+$X$ gate ZX-diagram:

$$
%
$$

To compute the $X$ stabilizers is to find the $a_1,a_2,a_3,b_1,b_2,b_3 \in \F_p$ such that

$$
%
$$

In $\Aff\Rel_{\F_p}$, this equation looks like:

$$
%
$$

These $a_1,a_2,a_3,b_1,b_2,b_3$ are parameterized by the elements of the affine subspaces:

$$
%
\right)
:
a_1=a_2=a_3\wedge
a_1+a_3+c = b_2+b_3
\right\}
$$

\end{example}

\section{Linear and Affine Lagrangian relations}
\label{sec:lagrel}

We first develop the basic theory of linear symplectic geometry, then we give compositional account of this theory in terms of linear relations and graphical linear algebra. This will allow us to capture even larger fragments of quantum theory with this semantics of  $\F_p$ linear/affine subspaces.

\subsection{Linear symplectic geometry}
In this section we give a basic introduction to linear symplectic geometry.  A more detailed introduction can be found in \cite{weinsteinsymplectic}.
\begin{definition}
  Given a field  $k$ and a $k$-vector space $V$, a {\bf symplectic form} on $V$ is a bilinear map $\omega:V\times V\to k$ which is:
\begin{description}
 \item[\ \ Alternating] $\forall v \in V$, $\omega(v,v)=0$.
 \item[\ \ Non-degenerate] if $\exists v \in V: \forall w \in V: \omega(v,w)=0$, then $v=0$.
\end{description}
  A {\bf symplectic vector space} is a vector space equipped with a symplectic form. A (linear) {\bf symplectomorphism} is a linear isomorphism between symplectic vector spaces that preserves the symplectic form. Two symplectic vector spaces are {\bf symplectomorphic}, when there is a symplectomorpism between them.
\end{definition}

\begin{lemma}
\label{lemma:sform}
Every vector space $k^{2n}$ with a chosen basis is equipped with a bilinear form given by the following block matrix:
$$
\omega:=
\begin{bmatrix}
0_n & I_n\\
-I_n & 0_n
\end{bmatrix}
$$
so that $\omega(v,w) := v \omega w^T$.
Moreover, every finite dimensional symplectic vector space over $k$ is symplectomorphic to one of the form $k^{2n}$ with such a symplectic form.
\end{lemma}

\begin{definition}

Let $W \subseteq V$ be a linear subspace of a symplectic space $V$.
The {\bf symplectic dual} of the subspace $W$ is defined to be
$
W^\omega:= \{v \in V : \forall w \in W, \omega(v,w)=0 \}
$.
A linear subspace  $W$ of a symplectic vector space $V$ is {\bf isotropic} when $W^\omega \supseteq W$, {\bf coisotropic} when $W^\omega \subseteq W$ and {\bf Lagrangian} when $W^\omega=W$.
\end{definition}
Notice that the symplectic complement reverses the order of inclusion, so that coisotropic subspaces are turned into isotropic subspaces and vice versa.  All of these subspaces generalize the notion of symplectomorphism:

\begin{lemma}
Every symplectomorphism $f:V\to V$ induces a Lagrangian subspace $\Gamma_f:=\{ (fv, v) | v \in V \}$.
\end{lemma}

As a matter of convention, we consider linear subspaces as being represented as the row space of a matrix.
An isotropic subspace can equivalently be characterized as the $k$-linear row space of a matrix $[Z|X]$ so that $[Z|X] \omega [Z|X]^T = 0$.
Moreover, a Lagrangian subspace can be described as a matrix satisfying this equation which additionally has rank $n$.

\begin{definition}
Given a field $k$, the prop of {\bf Lagrangian relations},  $\Lag\Rel_k$ has morphisms $n\to m$ being Lagrangian subspaces of the symplectic vector space $k^{n+m} \oplus k^{n+m}$ with respect to the symplectic form given above.  Composition is given by relational composition and the tensor product is given by the direct sum (where the $Z$ and $X$ gradients are grouped together).

The props of {\bf isotropic relations} and {\bf coisotropic relations}, $\Isot\Rel_{k}$ and $\Co\Isot\Rel_{k}$ are defined in the obvious analogous ways.
\end{definition}

In linear mechanical systems, symplectic vector spaces over $\R$ are interpreted as  the phase space: ie the space of all allowable configurations of position and momentum.  The Lagrangian subspaces are interpreted as the initial configurations of the mechanical system, and the symplectomorphisms describe the reversible evolution of the system.  The category of Lagrangian relations allows one to tensor and compose the various configurations and evolutions of the system.  Coisotropic and isotropic relations allows more or less degrees of freedom of the configurations and evolution of the phase space.

\subsection{String diagrams for Lagrangian relations}
The following subsection reviews the results of \cite{lagrel}.
The CPM construction of \cite{cpm} arises in two ways in this paper. We present it using a conjugation rather than a dagger because this is more natural for our purposes.  Also, we will assume that our categories are self dual compact closed, so that $X^*=X$ for no reason other than convenience:

\begin{definition}
Given a self dual  compact closed category $\X$ with an identity on objects strict symmetric monoidal involution $\bar{(-)}:\X\to \X$, the symmetric monoidal category $\CPM(\X,\bar{(-)})$ has:

\begin{description}
\item[\ \ Objects:] Same as $\X$.

\item[\ \ Maps:] Maps $\hat f:X \to Y$ in $\CPM(\X)$ are given by by equivalence classes of {\bf purifications} \\$(Z,f:X\to Y\otimes Z)$ in $\X$, so that $(Z,f:X\to Y\otimes Z)\sim (Z,f':X\to Y\otimes Z')$ if and only if:

\hfil$
\begin{tikzpicture}
	\begin{pgfonlayer}{nodelayer}
		\node [style=none] (7) at (2.25, 3.25) {};
		\node [style=none] (8) at (3.75, 3.25) {};
		\node [style=none] (9) at (2.5, 1.5) {};
		\node [style=none] (10) at (4, 1.5) {};
		\node [style=Z] (11) at (3.25, 2.75) {};
		\node [style=map] (12) at (2.5, 2) {$\bar f$};
		\node [style=map] (13) at (4, 2) {$ f$};
		\node [style=none] (14) at (2.5, 1.25) {$X$};
		\node [style=none] (15) at (4, 1.25) {$X$};
		\node [style=none] (16) at (2.25, 3.5) {$Y$};
		\node [style=none] (17) at (3.75, 3.5) {$Y$};
		\node [style=none] (18) at (2.75, 2.75) {$Z$};
	\end{pgfonlayer}
	\begin{pgfonlayer}{edgelayer}
		\draw [bend left] (12) to (11);
		\draw [in=60, out=0, looseness=1.50] (11) to (13);
		\draw [in=-90, out=120] (13) to (8.center);
		\draw (10.center) to (13);
		\draw (9.center) to (12);
		\draw [in=270, out=120] (12) to (7.center);
	\end{pgfonlayer}
\end{tikzpicture}
=
\begin{tikzpicture}
	\begin{pgfonlayer}{nodelayer}
		\node [style=none] (7) at (2.25, 3.25) {};
		\node [style=none] (8) at (3.75, 3.25) {};
		\node [style=none] (9) at (2.5, 1.5) {};
		\node [style=none] (10) at (4, 1.5) {};
		\node [style=Z] (11) at (3.25, 2.75) {};
		\node [style=map] (12) at (2.5, 2) {$\bar f'$};
		\node [style=map] (13) at (4, 2) {$ {f'}$};
		\node [style=none] (14) at (2.5, 1.25) {$X$};
		\node [style=none] (15) at (4, 1.25) {$X$};
		\node [style=none] (16) at (2.25, 3.5) {$Y$};
		\node [style=none] (17) at (3.75, 3.5) {$Y$};
		\node [style=none] (18) at (2.75, 2.75) {$Z'$};
	\end{pgfonlayer}
	\begin{pgfonlayer}{edgelayer}
		\draw [bend left] (12) to (11);
		\draw [in=60, out=0, looseness=1.50] (11) to (13);
		\draw [in=-90, out=120] (13) to (8.center);
		\draw (10.center) to (13);
		\draw (9.center) to (12);
		\draw [in=270, out=120] (12) to (7.center);
	\end{pgfonlayer}
\end{tikzpicture}
$

\item[\ \ Composition:] Vertical stacking.

\item[\ \ Symmetric monoidal structure:] Pointwise in $\X$.

\end{description}
The canonical functor $\X\to\CPM(\X,
\bar{(-)})$ taking
$$
\begin{tikzpicture}
	\begin{pgfonlayer}{nodelayer}
		\node [style=map] (15) at (7.25, -0.5) {$f$};
		\node [style=none] (17) at (7.25, -1.25) {};
		\node [style=none] (19) at (7.25, 0.25) {};
	\end{pgfonlayer}
	\begin{pgfonlayer}{edgelayer}
		\draw (17.center) to (15.center);
		\draw (15.center) to (19.center);
	\end{pgfonlayer}
\end{tikzpicture}
\mapsto
\begin{tikzpicture}
	\begin{pgfonlayer}{nodelayer}
		\node [style=map] (10) at (6.25, -0.5) {$\bar f$};
		\node [style=map] (15) at (7.25, -0.5) {$f$};
		\node [style=none] (16) at (6.25, -1.25) {};
		\node [style=none] (17) at (7.25, -1.25) {};
		\node [style=none] (18) at (6.25, 0.25) {};
		\node [style=none] (19) at (7.25, 0.25) {};
	\end{pgfonlayer}
	\begin{pgfonlayer}{edgelayer}
		\draw (16.center) to (10.center);
		\draw (17.center) to (15.center);
		\draw (10.center) to (18.center);
		\draw (15.center) to (19.center);
	\end{pgfonlayer}
\end{tikzpicture}
$$
 is called {\bf doubling}. The maps in the image of this functor are {\bf pure}, and those which aren't are {\bf mixed}.

The map given by connecting both sides is the {\bf discard map}:
$$
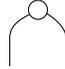
\begin{tikzpicture}
	\begin{pgfonlayer}{nodelayer}
		\node [style=none] (10) at (6.25, -0.5) {};
		\node [style=none] (15) at (7.25, -0.5) {};
		\node [style=none] (16) at (6.25, -1) {};
		\node [style=none] (17) at (7.25, -1) {};
		\node [style=Z] (18) at (6.75, 0) {};
	\end{pgfonlayer}
	\begin{pgfonlayer}{edgelayer}
		\draw (16.center) to (10.center);
		\draw (17.center) to (15.center);
		\draw [in=-15, out=90, looseness=0.75] (15.center) to (18.center);
		\draw [in=90, out=-150] (18.center) to (10.center);
	\end{pgfonlayer}
\end{tikzpicture}
$$
All maps can be obtained by composing pure maps with discard maps.  Given a mixed map $f$ in $\CPM(\X)$ such a factorization into a pure map followed by a discard map is a {\bf purification} of $f$.

\end{definition}

We will omit the conjugation functor when using this construction when it is clear from context.

This construction was first described to give a categorical semantics for density matrices:
\begin{example}

$\CPM(\FdHilb,\bar{(-)})$  is the category of density matrices between finite dimensional Hilbert spaces.
The pure states in this setting are pure quantum states, up to a scalar factor and the discard is the quantum discard.
\end{example}

Surprisingly, Lagrangian relations also arise from this construction, but in a very different way from the orthogonal complement of linear relations; here the discard adds phase:
\begin{theorem}[{\cite[Corollary 3.4]{lagrel}}]
 $\Lag\Rel_{\F_p} \cong \CPM(\LinRel_{\F_p}, (-)^\perp)$
\end{theorem}

This gives us generators for $\Lag\Rel_{\F_p}$ by doubling and tracing out the generators of $\LinRel_{\F_p}$.  This doubling reflects the symmetry between $Z$ and $X$ gradings of the ambient symplectic vector space. We regard the left $n$-most wires as the $Z$ grading and the right $n$ most wires as the $X$ grading.  By adding affine shifts to this structure, this is shown to be isomorphic to stabilizer circuits modulo nonzero scalars.

\begin{definition}
The prop of {\bf affine Lagrangian relations} $\Aff\Lag\Rel_{k}$ has morphisms $n\to m$ either the empty subspace or affine Lagrangian subspaces: affine subspaces  $L+a\subseteq k^{2n}\oplus k^{2m}$ such that $L$ is Lagrangian.

The props of {\bf affine isotropic relations} and {\bf affine coisotropic relations}$\Aff\Isot\Rel_{k}$ and $\Aff\Co\Isot\Rel_{k}$ are defined in the obvious analogous ways.
\end{definition}

We give a very brief overview of the qudit stabilizer formalism (see 
\cite{gota} for reference):
\begin{definition}
Fix some local dimension $d$.
A single qudit {\bf Weyl operator} is an $d$-dimensional unitary of the form, $\mathcal{Z}^{z}\mathcal{X}^{x}$, for $x,z \in \Z/d\Z$.

An $n$-qudit Weyl operator is the $n$-fold tensor product of single qudit Weyl operators.   Given some $\vec z, \vec x \in k^{n}$, the corresponding $n$-qudit Weyl-operator is denoted by:
$$
{\mathcal W}(\vec z,\vec x)=\bigotimes_{j=0}^{n-1}\mathcal{Z}_{(j)}^{z_j}\mathcal{X}_{(j)}^{x_j}
$$
The $n$-qudit Weyl operators form the {\bf Heisenberg-Weyl group} (called the Pauli group for qubits) $\mathcal{P}_d^n$ under matrix multiplication and the Hermitian adjoint.

An $n$-qudit (pure) {\bf stabilizer state} is a state stabilized by a maximal subgroup of the Heisenberg-Weyl group.  Considered all together, the qudit stabilizer states in addition to the cups and caps form the prop of (pure) qudit {\bf stabilizer circuits}. 
Concretely, this is the prop generated by the matrices for state preparation and postselection in the standard basis along with the qudit Clifford gates.

Let $\Stab_d$ denote prop of (pure) qudit stabilizer circuits modulo invertible scalars.  
%
%
\end{definition}

\begin{theorem}[{\cite[Theorem 4.16]{lagrel}}]
For odd prime $p$, $\Aff\Lag\Rel_{\F_p}$ is isomorphic to $\Stab_p$.

Similarly, $\Aff\Lag\Rel_{\F_2}$ is isomorphic to the pure circuits in Spekkens' toy model.
\end{theorem}

A basis for the affine Lagrangian subspace corresponds to the stabilizer tableau for a pure stabilizer state (ie a stabilizer tableau on $n$ qudits with dimension $n$).  Denote the symplectic version of a Weyl operator as follows:

$$
W(\vec z, \vec x)
:=
\begin{tikzpicture}
	\begin{pgfonlayer}{nodelayer}
		\node [style=X] (0) at (17, 4.5) {$\vec z$};
		\node [style=X] (2) at (18, 4.5) {$\vec x$};
		\node [style=none] (3) at (17, 5.25) {};
		\node [style=none] (4) at (18, 5.25) {};
		\node [style=none] (5) at (17, 3.75) {};
		\node [style=none] (6) at (18, 3.75) {};
	\end{pgfonlayer}
	\begin{pgfonlayer}{edgelayer}
		\draw (6.center) to (2.center);
		\draw (2.center) to (4.center);
		\draw (3.center) to (0.center);
		\draw (0.center) to (5.center);
	\end{pgfonlayer}
\end{tikzpicture}
$$

The novelty in interpreting stabilizer states in this categorical framework is that it reveals that the relational composition of tableaux is the composition of stabilizer circuits.

\begin{corollary}
For odd prime $p$, $\Lag\Rel_{\F_p}$ is a presentation for stabilizer circuits without affine phases. 
\end{corollary}

The following set of generators for $\Aff\Lag\Rel_{k}$, reveals the connection to the ZX-calculus:

\begin{theorem}[{\cite[Remark 4.17]{lagrel}}]

$\Aff\Lag\Rel_{k}$ is generated by two spiders both decorated by the additive group of $k^2$ as well as scaling gates when $k$ is not a prime field:
$$
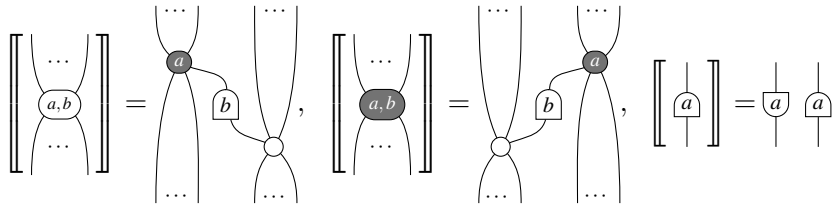
\left\llbracket\
\begin{tikzpicture}
	\begin{pgfonlayer}{nodelayer}
		\node [style=none] (0) at (21, 5) {};
		\node [style=none] (1) at (22, 5) {};
		\node [style=none] (2) at (21, 2.5) {};
		\node [style=none] (3) at (22, 2.5) {};
		\node [style=Z] (4) at (21.5, 3.75) {$a,b$};
		\node [style=none] (5) at (21.5, 4.5) {$\cdots$};
		\node [style=none] (6) at (21.5, 3) {$\cdots$};
		\node [style=none] (7) at (21.5, 4.75) {};
		\node [style=none] (8) at (21.5, 2.75) {};
	\end{pgfonlayer}
	\begin{pgfonlayer}{edgelayer}
		\draw [in=150, out=-90, looseness=0.75] (0.center) to (4);
		\draw [in=90, out=-150, looseness=0.75] (4) to (2.center);
		\draw [in=-30, out=90, looseness=0.75] (3.center) to (4);
		\draw [in=-90, out=30, looseness=0.75] (4) to (1.center);
	\end{pgfonlayer}
\end{tikzpicture}
\ \right\rrbracket
=
\begin{tikzpicture}
	\begin{pgfonlayer}{nodelayer}
		\node [style=none] (78) at (402.5, 0.75) {};
		\node [style=none] (79) at (402.5, -2.75) {};
		\node [style=Z] (80) at (402.1, -1.75) {};
		\node [style=none] (81) at (402.145, 0.65) {$\cdots$};
		\node [style=none] (82) at (402.125, -2.65) {$\cdots$};
		\node [style=none] (83) at (400.75, 0.75) {};
		\node [style=none] (84) at (400, -2.75) {};
		\node [style=X] (85) at (400.425, -0.25) {$a$};
		\node [style=none] (86) at (400.375, 0.65) {$\cdots$};
		\node [style=none] (87) at (400.395, -2.625) {$\cdots$};
		\node [style=none] (88) at (400.75, -2.75) {};
		\node [style=none] (89) at (401.75, -2.75) {};
		\node [style=none] (90) at (400, 0.75) {};
		\node [style=none] (91) at (401.75, 0.75) {};
		\node [style=scalar] (92) at (401.25, -1) {$b$};
	\end{pgfonlayer}
	\begin{pgfonlayer}{edgelayer}
		\draw [in=-45, out=90] (79.center) to (80);
		\draw [in=-90, out=60, looseness=0.75] (80) to (78.center);
		\draw [in=90, out=-120, looseness=0.75] (85) to (84.center);
		\draw [in=-90, out=45] (85) to (83.center);
		\draw [in=-60, out=90, looseness=0.75] (88.center) to (85);
		\draw [in=-135, out=90] (89.center) to (80);
		\draw [in=-90, out=150] (85) to (90.center);
		\draw [in=270, out=120, looseness=0.75] (80) to (91.center);
		\draw [in=-75, out=150] (80) to (92);
		\draw [in=330, out=90] (92) to (85);
	\end{pgfonlayer}
\end{tikzpicture},
\hspace*{.2cm}
\left\llbracket \
\begin{tikzpicture}
	\begin{pgfonlayer}{nodelayer}
		\node [style=none] (0) at (21, 5) {};
		\node [style=none] (1) at (22, 5) {};
		\node [style=none] (2) at (21, 2.5) {};
		\node [style=none] (3) at (22, 2.5) {};
		\node [style=X] (4) at (21.5, 3.75) {$a,b$};
		\node [style=none] (5) at (21.5, 4.5) {$\cdots$};
		\node [style=none] (6) at (21.5, 3) {$\cdots$};
		\node [style=none] (7) at (21.5, 4.75) {};
		\node [style=none] (8) at (21.5, 2.75) {};
	\end{pgfonlayer}
	\begin{pgfonlayer}{edgelayer}
		\draw [in=150, out=-90, looseness=0.75] (0.center) to (4);
		\draw [in=90, out=-150, looseness=0.75] (4) to (2.center);
		\draw [in=-30, out=90, looseness=0.75] (3.center) to (4);
		\draw [in=-90, out=30, looseness=0.75] (4) to (1.center);
	\end{pgfonlayer}
\end{tikzpicture}
\ \right\rrbracket
=
\begin{tikzpicture}
	\begin{pgfonlayer}{nodelayer}
		\node [style=none] (93) at (403.5, 0.75) {};
		\node [style=none] (94) at (403.5, -2.75) {};
		\node [style=Z] (95) at (403.9, -1.75) {};
		\node [style=none] (96) at (403.88, 0.65) {$\cdots$};
		\node [style=none] (97) at (403.85, -2.65) {$\cdots$};
		\node [style=none] (98) at (405.25, 0.75) {};
		\node [style=none] (99) at (406, -2.75) {};
		\node [style=X] (100) at (405.575, -0.25) {$a$};
		\node [style=none] (101) at (405.625, 0.65) {$\cdots$};
		\node [style=none] (102) at (405.605, -2.625) {$\cdots$};
		\node [style=none] (103) at (405.25, -2.75) {};
		\node [style=none] (104) at (404.25, -2.75) {};
		\node [style=none] (105) at (406, 0.75) {};
		\node [style=none] (106) at (404.25, 0.75) {};
		\node [style=scalar] (107) at (404.75, -1) {$b$};
	\end{pgfonlayer}
	\begin{pgfonlayer}{edgelayer}
		\draw [in=-135, out=90] (94.center) to (95);
		\draw [in=-90, out=120, looseness=0.75] (95) to (93.center);
		\draw [in=90, out=-60, looseness=0.75] (100) to (99.center);
		\draw [in=-90, out=135] (100) to (98.center);
		\draw [in=-120, out=90, looseness=0.75] (103.center) to (100);
		\draw [in=-45, out=90] (104.center) to (95);
		\draw [in=-90, out=30] (100) to (105.center);
		\draw [in=-90, out=60, looseness=0.75] (95) to (106.center);
		\draw [in=-105, out=30] (95) to (107);
		\draw [in=-150, out=90] (107) to (100);
	\end{pgfonlayer}
\end{tikzpicture}
,
\hspace*{.2cm}
\left\llbracket \
\begin{tikzpicture}
	\begin{pgfonlayer}{nodelayer}
		\node [style=none] (0) at (23, 4.5) {};
		\node [style=none] (1) at (23, 3) {};
		\node [style=scalar] (2) at (23, 3.75) {$a$};
	\end{pgfonlayer}
	\begin{pgfonlayer}{edgelayer}
		\draw (1.center) to (2);
		\draw (2) to (0.center);
	\end{pgfonlayer}
\end{tikzpicture}
\ \right\rrbracket
=
\begin{tikzpicture}
	\begin{pgfonlayer}{nodelayer}
		\node [style=none] (3) at (24.25, 4.5) {};
		\node [style=none] (4) at (24.25, 3) {};
		\node [style=scalarop] (5) at (24.25, 3.75) {$a$};
		\node [style=none] (6) at (25, 4.5) {};
		\node [style=none] (7) at (25, 3) {};
		\node [style=scalar] (8) at (25, 3.75) {$a$};
	\end{pgfonlayer}
	\begin{pgfonlayer}{edgelayer}
		\draw (4.center) to (5);
		\draw (5) to (3.center);
		\draw (7.center) to (8);
		\draw (8) to (6.center);
	\end{pgfonlayer}
\end{tikzpicture}
$$

\end{theorem}

The spider fusion is pointwise:
$$
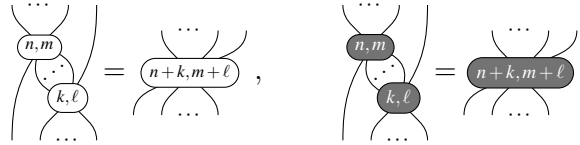
\begin{tikzpicture}
	\begin{pgfonlayer}{nodelayer}
		\node [style=none] (178) at (378.25, -0.425) {};
		\node [style=none] (179) at (377.25, -0.425) {};
		\node [style=none] (180) at (377.75, -0.425) {$\cdots$};
		\node [style=none] (181) at (377.25, -2.825) {};
		\node [style=Z] (182) at (377.75, -1.175) {$n,m$};
		\node [style=none] (183) at (378.75, -0.425) {};
		\node [style=none] (184) at (378.25, -2.825) {$\cdots$};
		\node [style=none] (185) at (377.75, -2.825) {};
		\node [style=Z] (186) at (378.25, -2.075) {$k,\ell$};
		\node [style=none] (187) at (378.75, -2.825) {};
		\node [style=none] (188) at (378, -1.5) {\reflectbox{$\ddots$}};
	\end{pgfonlayer}
	\begin{pgfonlayer}{edgelayer}
		\draw [in=-124, out=90] (181.center) to (182);
		\draw [in=-90, out=56] (182) to (178.center);
		\draw [in=124, out=-90] (179.center) to (182);
		\draw [in=-124, out=90] (185.center) to (186);
		\draw [in=90, out=-56] (186) to (187.center);
		\draw [in=-90, out=56] (186) to (183.center);
		\draw [bend left=45, looseness=1.25] (186) to (182);
		\draw [bend left=45, looseness=1.25] (182) to (186);
	\end{pgfonlayer}
\end{tikzpicture}
=
\begin{tikzpicture}
	\begin{pgfonlayer}{nodelayer}
		\node [style=none] (11) at (4, -0.5) {};
		\node [style=none] (12) at (3, -0.5) {};
		\node [style=none] (13) at (3.5, -0.5) {$\cdots$};
		\node [style=none] (14) at (2.5, -2) {};
		\node [style=none] (15) at (3.5, -1.25) {};
		\node [style=none] (16) at (4.5, -0.5) {};
		\node [style=none] (17) at (3.5, -2) {$\cdots$};
		\node [style=none] (18) at (3, -2) {};
		\node [style=Z] (19) at (3.5, -1.25) {$n+k,m+\ell$};
		\node [style=none] (20) at (4, -2) {};
	\end{pgfonlayer}
	\begin{pgfonlayer}{edgelayer}
		\draw [in=-150, out=90] (14.center) to (15);
		\draw [in=-90, out=56] (15) to (11.center);
		\draw [in=124, out=-90] (12.center) to (15);
		\draw [in=-124, out=90] (18.center) to (19);
		\draw [in=90, out=-56] (19) to (20.center);
		\draw [in=-90, out=30] (19) to (16.center);
	\end{pgfonlayer}
\end{tikzpicture}\ ,
\hspace*{1cm}
\begin{tikzpicture}
	\begin{pgfonlayer}{nodelayer}
		\node [style=none] (167) at (375.75, -0.425) {};
		\node [style=none] (168) at (374.75, -0.425) {};
		\node [style=none] (169) at (375.25, -0.425) {$\cdots$};
		\node [style=none] (170) at (374.75, -2.825) {};
		\node [style=X] (171) at (375.25, -1.175) {$n,m$};
		\node [style=none] (172) at (376.25, -0.425) {};
		\node [style=none] (173) at (375.75, -2.825) {$\cdots$};
		\node [style=none] (174) at (375.25, -2.825) {};
		\node [style=X] (175) at (375.75, -2.075) {$k,\ell$};
		\node [style=none] (176) at (376.25, -2.825) {};
		\node [style=none] (177) at (375.5, -1.5) {\reflectbox{$\ddots$}};
	\end{pgfonlayer}
	\begin{pgfonlayer}{edgelayer}
		\draw [in=-124, out=90] (170.center) to (171);
		\draw [in=-90, out=56] (171) to (167.center);
		\draw [in=124, out=-90] (168.center) to (171);
		\draw [in=-124, out=90] (174.center) to (175);
		\draw [in=90, out=-56] (175) to (176.center);
		\draw [in=-90, out=56] (175) to (172.center);
		\draw [bend left=45, looseness=1.25] (175) to (171);
		\draw [bend left=45, looseness=1.25] (171) to (175);
	\end{pgfonlayer}
\end{tikzpicture}
=
\begin{tikzpicture}
	\begin{pgfonlayer}{nodelayer}
		\node [style=none] (11) at (4, -0.5) {};
		\node [style=none] (12) at (3, -0.5) {};
		\node [style=none] (13) at (3.5, -0.5) {$\cdots$};
		\node [style=none] (14) at (2.5, -2) {};
		\node [style=none] (15) at (3.5, -1.25) {};
		\node [style=none] (16) at (4.5, -0.5) {};
		\node [style=none] (17) at (3.5, -2) {$\cdots$};
		\node [style=none] (18) at (3, -2) {};
		\node [style=X] (19) at (3.5, -1.25) {$n+k,m+\ell$};
		\node [style=none] (20) at (4, -2) {};
	\end{pgfonlayer}
	\begin{pgfonlayer}{edgelayer}
		\draw [in=-150, out=90] (14.center) to (15);
		\draw [in=-90, out=56] (15) to (11.center);
		\draw [in=124, out=-90] (12.center) to (15);
		\draw [in=-124, out=90] (18.center) to (19);
		\draw [in=90, out=-56] (19) to (20.center);
		\draw [in=-90, out=30] (19) to (16.center);
	\end{pgfonlayer}
\end{tikzpicture}
$$
Call the first component of the phase group the {\bf affine phase} and the second component the {\bf linear phase}.  The white spider corresponds to the $Z$-basis and the grey spider corresponds to the $X$-basis.

In Hilbert spaces, the spiders $\ell \to k$ are interpreted as follows:
\begin{align*}
\left\llbracket\
\begin{tikzpicture}
	\begin{pgfonlayer}{nodelayer}
		\node [style=none] (0) at (21, 4.5) {};
		\node [style=none] (1) at (22, 4.5) {};
		\node [style=none] (2) at (21, 3) {};
		\node [style=none] (3) at (22, 3) {};
		\node [style=none] (4) at (21.5, 3.75) {};
		\node [style=Z] (5) at (21.5, 3.75) {$n,m$};
		\node [style=none] (6) at (21.525, 4.425) {$\cdots$};
		\node [style=none] (7) at (21.525, 3.075) {$\cdots$};
	\end{pgfonlayer}
	\begin{pgfonlayer}{edgelayer}
		\draw [in=150, out=-90, looseness=0.75] (0.center) to (4.center);
		\draw [in=90, out=-150, looseness=0.75] (4.center) to (2.center);
		\draw [in=-30, out=90, looseness=0.75] (3.center) to (4.center);
		\draw [in=-90, out=30, looseness=0.75] (4.center) to (1.center);
	\end{pgfonlayer}
\end{tikzpicture}\
\right\rrbracket
= &
\sum_{a=0}^{p-1}  e^{\pi\cdot i /p (n\cdot a+m\cdot a^2)} |a, \ldots, a \rangle \langle a, \ldots, a|\\
\left\llbracket\
\begin{tikzpicture}
	\begin{pgfonlayer}{nodelayer}
		\node [style=none] (0) at (21, 4.5) {};
		\node [style=none] (1) at (22, 4.5) {};
		\node [style=none] (2) at (21, 3) {};
		\node [style=none] (3) at (22, 3) {};
		\node [style=none] (4) at (21.5, 3.75) {};
		\node [style=X] (5) at (21.5, 3.75) {$n,m$};
		\node [style=none] (6) at (21.525, 4.425) {$\cdots$};
		\node [style=none] (7) at (21.525, 3.075) {$\cdots$};
	\end{pgfonlayer}
	\begin{pgfonlayer}{edgelayer}
		\draw [in=150, out=-90, looseness=0.75] (0.center) to (4.center);
		\draw [in=90, out=-150, looseness=0.75] (4.center) to (2.center);
		\draw [in=-30, out=90, looseness=0.75] (3.center) to (4.center);
		\draw [in=-90, out=30, looseness=0.75] (4.center) to (1.center);
	\end{pgfonlayer}
\end{tikzpicture}\
\right\rrbracket
= &
\sum_{a=0}^{p-1}  e^{\pi\cdot i /p (n\cdot a+m\cdot a^2)} \mathcal{F}^{\otimes k}|a, \ldots, a \rangle \langle a, \ldots, a|(\mathcal{F}^\dag)^{\otimes \ell} \\
\end{align*}
The scaling gate, which is a derived generator in this setting is interpreted as:
$$
\left\llbracket \
\begin{tikzpicture}
	\begin{pgfonlayer}{nodelayer}
		\node [style=none] (0) at (23, 4.5) {};
		\node [style=none] (1) at (23, 3) {};
		\node [style=scalar] (2) at (23, 3.75) {$b$};
	\end{pgfonlayer}
	\begin{pgfonlayer}{edgelayer}
		\draw (1.center) to (2);
		\draw (2) to (0.center);
	\end{pgfonlayer}
\end{tikzpicture}
\ \right\rrbracket
=
\sum_{a=0}^{p-1} |a\cdot b \rangle \langle a |
$$

The Fourier transform is derived:
$$
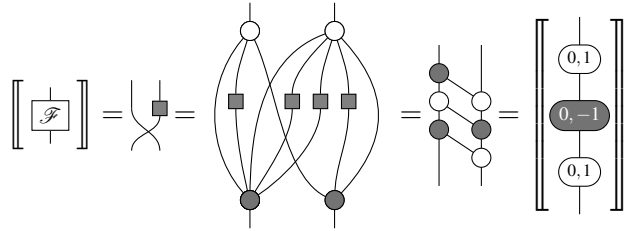
\left\llbracket\
\begin{tikzpicture}
	\begin{pgfonlayer}{nodelayer}
		\node [style=none] (0) at (1.25, -1) {};
		\node [style=map] (1) at (1.25, -1.5) {$\mathcal F$};
		\node [style=none] (2) at (1.25, -2) {};
	\end{pgfonlayer}
	\begin{pgfonlayer}{edgelayer}
		\draw (2.center) to (1);
		\draw (1) to (0.center);
	\end{pgfonlayer}
\end{tikzpicture}\
\right\rrbracket
=
\begin{tikzpicture}
	\begin{pgfonlayer}{nodelayer}
		\node [style=none] (0) at (0.5, 1) {};
		\node [style=none] (1) at (0.5, -0.25) {};
		\node [style=none] (2) at (1, -0.25) {};
		\node [style=none] (3) at (1, 1) {};
		\node [style=s] (4) at (1, 0.5) {};
		\node [style=none] (5) at (0.5, 0.5) {};
	\end{pgfonlayer}
	\begin{pgfonlayer}{edgelayer}
		\draw (4) to (3.center);
		\draw [in=90, out=-90] (4) to (1.center);
		\draw [in=-90, out=90] (2.center) to (5.center);
		\draw (5.center) to (0.center);
	\end{pgfonlayer}
\end{tikzpicture}
=
\begin{tikzpicture}[xscale=-1]
	\begin{pgfonlayer}{nodelayer}
		\node [style=X] (0) at (19.5, 1.25) {};
		\node [style=Z] (1) at (18, 4.25) {};
		\node [style=none] (2) at (18, 4.75) {};
		\node [style=none] (3) at (19.5, 4.75) {};
		\node [style=none] (4) at (18, 0.75) {};
		\node [style=none] (5) at (19.5, 0.75) {};
		\node [style=X] (6) at (19.5, 1.25) {};
		\node [style=Z] (7) at (19.5, 4.25) {};
		\node [style=X] (8) at (18, 1.25) {};
		\node [style=Z] (9) at (18, 4.25) {};
		\node [style=X] (10) at (19.5, 1.25) {};
		\node [style=X] (11) at (18, 1.25) {};
		\node [style=X] (12) at (19.5, 1.25) {};
		\node [style=Z] (13) at (19.5, 4.25) {};
		\node [style=s] (14) at (17.75, 3) {};
		\node [style=s] (15) at (18.75, 3) {};
		\node [style=s] (16) at (19.75, 3) {};
		\node [style=s] (17) at (18.25, 3) {};
	\end{pgfonlayer}
	\begin{pgfonlayer}{edgelayer}
		\draw (2.center) to (1);
		\draw (5.center) to (0);
		\draw [bend right=45] (6) to (7);
		\draw [in=135, out=-135, looseness=1.25] (9) to (8);
		\draw (3.center) to (7);
		\draw (4.center) to (8);
		\draw [in=-120, out=15, looseness=0.75] (11) to (13);
		\draw [in=90, out=-105] (9) to (14);
		\draw [in=90, out=-45, looseness=0.75] (9) to (15);
		\draw [in=90, out=-90] (14) to (11);
		\draw [in=-90, out=120, looseness=0.75] (6) to (15);
		\draw [in=-15, out=90] (12) to (9);
		\draw [in=-90, out=150, looseness=0.75] (12) to (17);
		\draw [in=285, out=90] (17) to (9);
		\draw [in=-90, out=75, looseness=0.75] (12) to (16);
		\draw [in=-75, out=90] (16) to (13);
	\end{pgfonlayer}
\end{tikzpicture}
=
\begin{tikzpicture}
	\begin{pgfonlayer}{nodelayer}
		\node [style=none] (0) at (0, 0.75) {};
		\node [style=none] (1) at (0.75, 0.75) {};
		\node [style=none] (2) at (0, 3.25) {};
		\node [style=none] (3) at (0.75, 3.25) {};
		\node [style=Z] (4) at (0.75, 1.25) {};
		\node [style=X] (5) at (0, 1.75) {};
		\node [style=Z] (6) at (0.75, 2.25) {};
		\node [style=X] (7) at (0, 2.75) {};
		\node [style=Z] (8) at (0, 2.25) {};
		\node [style=X] (9) at (0.75, 1.75) {};
	\end{pgfonlayer}
	\begin{pgfonlayer}{edgelayer}
		\draw (4) to (5);
		\draw (6) to (7);
		\draw (8) to (9);
		\draw (1.center) to (4);
		\draw (4) to (9);
		\draw (9) to (6);
		\draw (6) to (3.center);
		\draw (2.center) to (7);
		\draw (7) to (8);
		\draw (8) to (5);
		\draw (5) to (0.center);
	\end{pgfonlayer}
\end{tikzpicture}
=
\left\llbracket
\begin{tikzpicture}
	\begin{pgfonlayer}{nodelayer}
		\node [style=none] (0) at (1.25, 0) {};
		\node [style=none] (1) at (1.25, -3.5) {};
		\node [style=Z] (2) at (1.25, -2.75) {$0,1$};
		\node [style=Z] (3) at (1.25, -0.75) {$0,1$};
		\node [style=X] (4) at (1.25, -1.75) {$0,-1$};
	\end{pgfonlayer}
	\begin{pgfonlayer}{edgelayer}
		\draw (1.center) to (2);
		\draw (2) to (4);
		\draw (4) to (3);
		\draw (3) to (0.center);
	\end{pgfonlayer}
\end{tikzpicture}
\right\rrbracket
$$

In the ZX-calculus literature, this decomposition of the Fourier transform is known as {\it Euler decomposition} \cite{duncan2009graph}.

This view of odd prime dimensional stabilizer circuits in terms of the ZX-calculus means that the phase groups for the $Z$ and $X$-spiders are the torus $(\Z/p\Z)^2$  as noted \cite[Page 166]{ranchin2016alternative}.  This is in contrast to the qubit case where the phase groups are $\Z/4\Z$; which \cite{coecke2011phase} point out as a crucial difference between Spekkens' toy model and qubit stabilizer theory.
However over $\F_2$ when the phases are restricted to the subgroup
$$\Z/2\Z\subseteq \Z/4\Z; \ a\mapsto 2a$$
and for odd prime $p$, over $\F_p$ when the phases are restricted to the subgroup
$$\Z/p\Z \subseteq (\Z/p\Z)^2; \ a\mapsto (a,0)$$
then both intepretations coincide. Therefore, we can faithfully interpret qubit stabilizer circuits generated by controlled-not gates, Z gates,  X gates as well as the states $|+\rangle$ and postselections $|0\rangle$, $\langle + |$ and $\langle 0|$ in affine Lagrangian relations.

Up until now, we have also neglected to relate the symplectic picture to the \dag-structure of $\FHilb$.  This will be instrumental in the following section:
\begin{definition}
\label{def:conj}
There is a monoidal conjugation functor $\bar{(-)}:\Aff\Lag\Rel_k\to \Aff\Lag\Rel_k$ given by:
$$
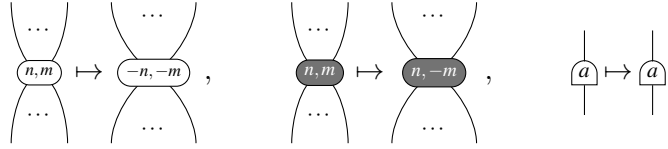
\begin{tikzpicture}
	\begin{pgfonlayer}{nodelayer}
		\node [style=none] (0) at (21, 5) {};
		\node [style=none] (1) at (22, 5) {};
		\node [style=none] (2) at (21, 2.5) {};
		\node [style=none] (3) at (22, 2.5) {};
		\node [style=none] (4) at (21.5, 3.75) {};
		\node [style=Z] (40) at (21.5, 3.75) {$n,m$};
		\node [style=none] (5) at (21.5, 4.5) {$\cdots$};
		\node [style=none] (6) at (21.5, 3) {$\cdots$};
		\node [style=none] (7) at (21.5, 4.75) {};
		\node [style=none] (8) at (21.5, 2.75) {};
	\end{pgfonlayer}
	\begin{pgfonlayer}{edgelayer}
		\draw [in=150, out=-90, looseness=0.75] (0.center) to (4);
		\draw [in=90, out=-150, looseness=0.75] (4) to (2.center);
		\draw [in=-30, out=90, looseness=0.75] (3.center) to (4);
		\draw [in=-90, out=30, looseness=0.75] (4) to (1.center);
	\end{pgfonlayer}
\end{tikzpicture}
\mapsto
\begin{tikzpicture}
	\begin{pgfonlayer}{nodelayer}
		\node [style=none] (0) at (20.75, 5) {};
		\node [style=none] (1) at (22.25, 5) {};
		\node [style=none] (2) at (20.75, 2.5) {};
		\node [style=none] (3) at (22.25, 2.5) {};
		\node [style=none] (4) at (21.5, 3.75) {};
		\node [style=none] (5) at (21.5, 4.75) {$\cdots$};
		\node [style=none] (6) at (21.5, 2.75) {$\cdots$};
		\node [style=Z] (9) at (21.5, 3.75) {$-n,-m$};
	\end{pgfonlayer}
	\begin{pgfonlayer}{edgelayer}
		\draw [in=150, out=-90, looseness=0.75] (0.center) to (4);
		\draw [in=90, out=-150, looseness=0.75] (4) to (2.center);
		\draw [in=-30, out=90, looseness=0.75] (3.center) to (4);
		\draw [in=-90, out=30, looseness=0.75] (4) to (1.center);
	\end{pgfonlayer}
\end{tikzpicture}\ ,
\hspace*{1cm}
\begin{tikzpicture}
	\begin{pgfonlayer}{nodelayer}
		\node [style=none] (0) at (21, 5) {};
		\node [style=none] (1) at (22, 5) {};
		\node [style=none] (2) at (21, 2.5) {};
		\node [style=none] (3) at (22, 2.5) {};
		\node [style=none] (4) at (21.5, 3.75) {};
		\node [style=X] (40) at (21.5, 3.75) {$n,m$};
		\node [style=none] (5) at (21.5, 4.5) {$\cdots$};
		\node [style=none] (6) at (21.5, 3) {$\cdots$};
		\node [style=none] (7) at (21.5, 4.75) {};
		\node [style=none] (8) at (21.5, 2.75) {};
	\end{pgfonlayer}
	\begin{pgfonlayer}{edgelayer}
		\draw [in=150, out=-90, looseness=0.75] (0.center) to (4);
		\draw [in=90, out=-150, looseness=0.75] (4) to (2.center);
		\draw [in=-30, out=90, looseness=0.75] (3.center) to (4);
		\draw [in=-90, out=30, looseness=0.75] (4) to (1.center);
	\end{pgfonlayer}
\end{tikzpicture}
\mapsto
\begin{tikzpicture}
	\begin{pgfonlayer}{nodelayer}
		\node [style=none] (0) at (20.75, 5) {};
		\node [style=none] (1) at (22.25, 5) {};
		\node [style=none] (2) at (20.75, 2.5) {};
		\node [style=none] (3) at (22.25, 2.5) {};
		\node [style=none] (4) at (21.5, 3.75) {};
		\node [style=none] (5) at (21.5, 4.75) {$\cdots$};
		\node [style=none] (6) at (21.5, 2.75) {$\cdots$};
		\node [style=none] (9) at (21.5, 3.75) {};
		\node [style=X] (90) at (21.5, 3.75) {$n,-m$};
	\end{pgfonlayer}
	\begin{pgfonlayer}{edgelayer}
		\draw [in=150, out=-90, looseness=0.75] (0.center) to (4);
		\draw [in=90, out=-150, looseness=0.75] (4) to (2.center);
		\draw [in=-30, out=90, looseness=0.75] (3.center) to (4);
		\draw [in=-90, out=30, looseness=0.75] (4) to (1.center);
	\end{pgfonlayer}
\end{tikzpicture}\ ,
\hspace*{1cm}
\begin{tikzpicture}
	\begin{pgfonlayer}{nodelayer}
		\node [style=none] (0) at (23, 4.5) {};
		\node [style=none] (1) at (23, 3) {};
		\node [style=scalar] (2) at (23, 3.75) {$a$};
	\end{pgfonlayer}
	\begin{pgfonlayer}{edgelayer}
		\draw (1.center) to (2);
		\draw (2) to (0.center);
	\end{pgfonlayer}
\end{tikzpicture}
\mapsto
\begin{tikzpicture}
	\begin{pgfonlayer}{nodelayer}
		\node [style=none] (0) at (23, 4.5) {};
		\node [style=none] (1) at (23, 3) {};
		\node [style=scalar] (2) at (23, 3.75) {$a$};
	\end{pgfonlayer}
	\begin{pgfonlayer}{edgelayer}
		\draw (1.center) to (2);
		\draw (2) to (0.center);
	\end{pgfonlayer}
\end{tikzpicture}
$$
Define the  dagger functor to be the conjugate converse $$(-)^\dag:=(\ \bar{(-)}\ )^*:\Aff\Lag\Rel_k^\op\to\Aff\Lag\Rel_k$$
\end{definition}

Concretely, the conjugation functor negates the elements in the $Z$-grading of the vector space.

In the case of $k=\F_p$ for $p$ an odd prime, the conjugation is transported to the complex conjugation along $\Aff\Lag\Rel_{\F_p} \cong \Stab_p$; and the dagger to the Hermitian adjoint:

\begin{lemma}
For odd prime $p$, $\Aff\Lag\Rel_{\F_p}$ and  $\Stab_p$ are isomorphic as \dag-compact closed categories.
\end{lemma}
\begin{proof}
We know that all stabilizer states are of the form $C|0\rangle^{\otimes n}$ for $C$ Clifford operator.
The affine symplectomorphisms corresponding to the generators of the Clifford group can be easily verified to be unitaries with respect to the dagger.  Moreover, $|0\rangle$ is an isometry in $\Stab_p$ with the Hermitian adjoint just as 
$\begin{tikzpicture}
	\begin{pgfonlayer}{nodelayer}
		\node [style=X] (389) at (205, -2) {};
		\node [style=none] (390) at (205, -1.75) {};
		\node [style=Z] (391) at (204.5, -2) {};
		\node [style=none] (392) at (204.5, -1.75) {};
	\end{pgfonlayer}
	\begin{pgfonlayer}{edgelayer}
		\draw (389) to (390.center);
		\draw (391) to (392.center);
	\end{pgfonlayer}
\end{tikzpicture}$ is an isometry with respect to the dagger.
\end{proof}
\begin{remark}
In \cite{passive} and \cite{network}, they show that certain idealized  electrical circuits can be interpreted in terms of affine Lagrangian relations over  the field of real rational functions $\R(s)$.  This is the smallest field in which the polynomial ring $\R[s]$ embeds.

 We restate the interpretations of electrical circuits in terms of phased spiders for affine Lagrangian relations. An idealized junction, a resistor with resistance $r$, a capacitor with capacitance $c$ and an inductor with inductance $\ell$ are interpreted as:

$$
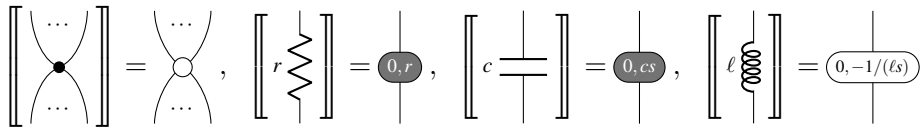
\left\llbracket\ 
\begin{tikzpicture}
	\begin{pgfonlayer}{nodelayer}
		\node [style=none] (0) at (21, 4.75) {};
		\node [style=none] (1) at (22, 4.75) {};
		\node [style=none] (2) at (21, 2.75) {};
		\node [style=none] (3) at (22, 2.75) {};
		\node [style=dot] (4) at (21.5, 3.75) {};
		\node [style=none] (5) at (21.5, 4.5) {$\cdots$};
		\node [style=none] (6) at (21.5, 3) {$\cdots$};
	\end{pgfonlayer}
	\begin{pgfonlayer}{edgelayer}
		\draw [in=150, out=-90, looseness=0.75] (0.center) to (4);
		\draw [in=90, out=-150, looseness=0.75] (4) to (2.center);
		\draw [in=-30, out=90, looseness=0.75] (3.center) to (4);
		\draw [in=-90, out=30, looseness=0.75] (4) to (1.center);
	\end{pgfonlayer}
\end{tikzpicture}
\ \right\rrbracket
=
\begin{tikzpicture}
	\begin{pgfonlayer}{nodelayer}
		\node [style=none] (0) at (21, 4.75) {};
		\node [style=none] (1) at (22, 4.75) {};
		\node [style=none] (2) at (21, 2.75) {};
		\node [style=none] (3) at (22, 2.75) {};
		\node [style=Z] (4) at (21.5, 3.75) {};
		\node [style=none] (5) at (21.5, 4.5) {$\cdots$};
		\node [style=none] (6) at (21.5, 3) {$\cdots$};
	\end{pgfonlayer}
	\begin{pgfonlayer}{edgelayer}
		\draw [in=150, out=-90, looseness=0.75] (0.center) to (4);
		\draw [in=90, out=-150, looseness=0.75] (4) to (2.center);
		\draw [in=-30, out=90, looseness=0.75] (3.center) to (4);
		\draw [in=-90, out=30, looseness=0.75] (4) to (1.center);
	\end{pgfonlayer}
\end{tikzpicture}
\ ,\hspace*{.2cm}
\left\llbracket
\tikz \draw (0,0) to[R=$r$] (0,2);
\ \right\rrbracket
=
\begin{tikzpicture}
	\begin{pgfonlayer}{nodelayer}
		\node [style=none] (0) at (21, 4.75) {};
		\node [style=none] (2) at (21, 2.75) {};
		\node [style=X] (4) at (21, 3.75) {$0,r$};
	\end{pgfonlayer}
	\begin{pgfonlayer}{edgelayer}
		\draw (0.center) to (4);
		\draw (4) to (2.center);
	\end{pgfonlayer}
\end{tikzpicture}
\ ,\hspace*{.2cm}
\left\llbracket
\tikz \draw (0,0) to[C=$c$] (0,2);
\ \right\rrbracket
=
\begin{tikzpicture}
	\begin{pgfonlayer}{nodelayer}
		\node [style=none] (0) at (21, 4.75) {};
		\node [style=none] (2) at (21, 2.75) {};
		\node [style=X] (4) at (21, 3.75) {$0,cs$};
	\end{pgfonlayer}
	\begin{pgfonlayer}{edgelayer}
		\draw (0.center) to (4);
		\draw (4) to (2.center);
	\end{pgfonlayer}
\end{tikzpicture}
\ ,\hspace*{.2cm}
\left\llbracket
\tikz \draw (0,0) to[L=$\ell$] (0,2);
\ \right\rrbracket
=
\begin{tikzpicture}
	\begin{pgfonlayer}{nodelayer}
		\node [style=none] (0) at (21, 4.75) {};
		\node [style=none] (2) at (21, 2.75) {};
		\node [style=Z] (4) at (21, 3.75) {$0,-1/(\ell s)$};
	\end{pgfonlayer}
	\begin{pgfonlayer}{edgelayer}
		\draw (0.center) to (4);
		\draw (4) to (2.center);
	\end{pgfonlayer}
\end{tikzpicture}
$$

A voltage source with a fixed voltage $v$ and current sources with fixed current $I$ are interpreted as:

$$
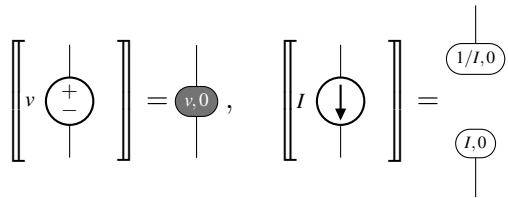
\left\llbracket
\begin{tikzpicture}
	\begin{pgfonlayer}{nodelayer}
		\node [style=none] (0) at (0, 2) {};
		\node [style=vsourceAMshape,rotate=-90,fill=white] (10) at (0, 1) {};
		\node [style=none] (1) at (0, 1) {};
		\node [style=none] (2) at (0, 0) {};
	\end{pgfonlayer}
	\begin{pgfonlayer}{edgelayer}
		\draw (2.center) to (1);
		\draw (1) to (0.center);
		\node [style=none] (3) at (-.7, 1) {$v$};
	\end{pgfonlayer}
\end{tikzpicture}
\hspace*{.3cm}
\right\rrbracket
=
\begin{tikzpicture}
	\begin{pgfonlayer}{nodelayer}
		\node [style=none] (0) at (21, 4.75) {};
		\node [style=none] (2) at (21, 2.75) {};
		\node [style=X] (4) at (21, 3.75) {$v, 0$};
	\end{pgfonlayer}
	\begin{pgfonlayer}{edgelayer}
		\draw (0.center) to (4);
		\draw (4) to (2.center);
	\end{pgfonlayer}
\end{tikzpicture}
\ ,\hspace*{.5cm}
\left\llbracket
\begin{tikzpicture}
	\begin{pgfonlayer}{nodelayer}
		\node [style=none] (0) at (0, 2) {};
		\node [style=isourceAMshape,rotate=-90,fill=white] (10) at (0, 1) {};
		\node [style=none] (1) at (0, 1) {};
		\node [style=none] (2) at (0, 0) {};
		\node [style=none] (3) at (-.7, 1) {$I$};
	\end{pgfonlayer}
	\begin{pgfonlayer}{edgelayer}
		\draw (2.center) to (1);
		\draw (1) to (0.center);
	\end{pgfonlayer}
\end{tikzpicture}
\hspace*{.3cm}
\right\rrbracket
=
\begin{tikzpicture}
	\begin{pgfonlayer}{nodelayer}
		\node [style=Z] (13) at (26.5, 4.25) {$1/I,0$};
		\node [style=Z] (14) at (26.5, 2.75) {$I,0$};
		\node [style=none] (15) at (26.5, 5.2) {};
		\node [style=none] (16) at (26.5, 1.8) {};
	\end{pgfonlayer}
	\begin{pgfonlayer}{edgelayer}
		\draw (16.center) to (14);
		\draw (13) to (15.center);
	\end{pgfonlayer}
\end{tikzpicture}
$$
\end{remark}

\section{String diagrams for measurement and coisotropic relations}
\label{sec:coisotrel}
\label{sec:coisot}
In this section we show that by only requiring that the morphisms are affine {\em coisotropic} subspaces   (subspaces $V$ so that $V^\omega \subseteq V$) instead of affine Lagrangian subspaces (where $V^\omega= V$), we can capture the maximally mixed state/discarding; with which we can recover state preparation and measurement compositionally.

\begin{remark}
The cozero linear relation $2n\to 0$ is an  isotropic subspace since:
$$
\left(
\begin{tikzpicture}
	\begin{pgfonlayer}{nodelayer}
		\node [style=X] (0) at (0.5, 0.5) {};
		\node [style=X] (1) at (1, 0.5) {};
		\node [style=none] (2) at (0.5, 0) {};
		\node [style=none] (3) at (1, 0) {};
	\end{pgfonlayer}
	\begin{pgfonlayer}{edgelayer}
		\draw (1) to (3.center);
		\draw (0) to (2.center);
	\end{pgfonlayer}
\end{tikzpicture}
\right)^\omega
=
\begin{tikzpicture}
	\begin{pgfonlayer}{nodelayer}
		\node [style=Z] (0) at (0.5, 0.5) {};
		\node [style=Z] (1) at (1, 0.5) {};
		\node [style=none] (2) at (0.5, 0) {};
		\node [style=none] (3) at (1, 0) {};
		\node [style=s] (4) at (1, 0) {};
		\node [style=none] (5) at (1, -0.75) {};
		\node [style=none] (7) at (0.5, -0.75) {};
	\end{pgfonlayer}
	\begin{pgfonlayer}{edgelayer}
		\draw (1) to (3.center);
		\draw (0) to (2.center);
		\draw [in=270, out=90] (7.center) to (4.center);
		\draw [in=90, out=-90] (2.center) to (5.center);
	\end{pgfonlayer}
\end{tikzpicture}
=
\begin{tikzpicture}
	\begin{pgfonlayer}{nodelayer}
		\node [style=Z] (0) at (0.5, 0.5) {};
		\node [style=Z] (1) at (1, 0.5) {};
		\node [style=none] (2) at (0.5, 0) {};
		\node [style=none] (3) at (1, 0) {};
	\end{pgfonlayer}
	\begin{pgfonlayer}{edgelayer}
		\draw (1) to (3.center);
		\draw (0) to (2.center);
	\end{pgfonlayer}
\end{tikzpicture}
\supset
\begin{tikzpicture}
	\begin{pgfonlayer}{nodelayer}
		\node [style=X] (0) at (0.5, 0.5) {};
		\node [style=X] (1) at (1, 0.5) {};
		\node [style=none] (2) at (0.5, 0) {};
		\node [style=none] (3) at (1, 0) {};
	\end{pgfonlayer}
	\begin{pgfonlayer}{edgelayer}
		\draw (1) to (3.center);
		\draw (0) to (2.center);
	\end{pgfonlayer}
\end{tikzpicture}
$$

And dually the codiscard linear relation $0\to 2n$ is coisotropic  since:
$$
\left(
\begin{tikzpicture}[yscale=-1]
	\begin{pgfonlayer}{nodelayer}
		\node [style=Z] (0) at (0.5, 0.5) {};
		\node [style=Z] (1) at (1, 0.5) {};
		\node [style=none] (2) at (0.5, 0) {};
		\node [style=none] (3) at (1, 0) {};
	\end{pgfonlayer}
	\begin{pgfonlayer}{edgelayer}
		\draw (1) to (3.center);
		\draw (0) to (2.center);
	\end{pgfonlayer}
\end{tikzpicture}
\right)^\omega
=
\begin{tikzpicture}[yscale=-1]
	\begin{pgfonlayer}{nodelayer}
		\node [style=X] (0) at (0.5, 0.5) {};
		\node [style=X] (1) at (1, 0.5) {};
		\node [style=none] (2) at (0.5, 0) {};
		\node [style=none] (3) at (1, 0) {};
		\node [style=s] (4) at (1, 0) {};
		\node [style=none] (5) at (1, -0.75) {};
		\node [style=none] (7) at (0.5, -0.75) {};
	\end{pgfonlayer}
	\begin{pgfonlayer}{edgelayer}
		\draw (1) to (3.center);
		\draw (0) to (2.center);
		\draw [in=270, out=90] (7.center) to (4.center);
		\draw [in=90, out=-90] (2.center) to (5.center);
	\end{pgfonlayer}
\end{tikzpicture}
=
\begin{tikzpicture}[yscale=-1]
	\begin{pgfonlayer}{nodelayer}
		\node [style=X] (0) at (0.5, 0.5) {};
		\node [style=X] (1) at (1, 0.5) {};
		\node [style=none] (2) at (0.5, 0) {};
		\node [style=none] (3) at (1, 0) {};
	\end{pgfonlayer}
	\begin{pgfonlayer}{edgelayer}
		\draw (1) to (3.center);
		\draw (0) to (2.center);
	\end{pgfonlayer}
\end{tikzpicture}
\supset
\begin{tikzpicture}[yscale=-1]
	\begin{pgfonlayer}{nodelayer}
		\node [style=Z] (0) at (0.5, 0.5) {};
		\node [style=Z] (1) at (1, 0.5) {};
		\node [style=none] (2) at (0.5, 0) {};
		\node [style=none] (3) at (1, 0) {};
	\end{pgfonlayer}
	\begin{pgfonlayer}{edgelayer}
		\draw (1) to (3.center);
		\draw (0) to (2.center);
	\end{pgfonlayer}
\end{tikzpicture}
$$

Given a matrix $A:n\to m$ in $\cb_k$, we can recover the image and kernel of $A$ diagrammatically:
$$
\begin{tikzpicture}
	\begin{pgfonlayer}{nodelayer}
		\node [style=map] (3) at (21, 6) {$\im A$};
		\node [style=none] (5) at (21, 7) {};
	\end{pgfonlayer}
	\begin{pgfonlayer}{edgelayer}
		\draw (3) to (5.center);
	\end{pgfonlayer}
\end{tikzpicture}
=
\begin{tikzpicture}
	\begin{pgfonlayer}{nodelayer}
		\node [style=map] (0) at (20, 6) {$A$};
		\node [style=Z] (1) at (20, 5) {};
		\node [style=none] (2) at (20, 7) {};
	\end{pgfonlayer}
	\begin{pgfonlayer}{edgelayer}
		\draw (1) to (0);
		\draw (0) to (2.center);
	\end{pgfonlayer}
\end{tikzpicture}
\ , \hspace*{1cm}
\begin{tikzpicture}
	\begin{pgfonlayer}{nodelayer}
		\node [style=map] (3) at (21, 6) {$\ker A $};
		\node [style=none] (5) at (21, 7) {};
	\end{pgfonlayer}
	\begin{pgfonlayer}{edgelayer}
		\draw (3) to (5.center);
	\end{pgfonlayer}
\end{tikzpicture}
=
\begin{tikzpicture}
	\begin{pgfonlayer}{nodelayer}
		\node [style=map] (0) at (20, 6) {$A^*$};
		\node [style=X] (1) at (20, 5) {};
		\node [style=none] (2) at (20, 7) {};
	\end{pgfonlayer}
	\begin{pgfonlayer}{edgelayer}
		\draw (1) to (0);
		\draw (0) to (2.center);
	\end{pgfonlayer}
\end{tikzpicture}
$$
Therefore, if we ask that $A$ is a Lagrangian relation rather than matrix, we find that  kernels of Lagrangian relations are isotropic subspaces and the images are coisotropic subspaces.
\end{remark}

All (co)isotropic subspaces are generated in this way:

\begin{theorem}[Symplectic Stinespring dilation]
\label{thm:symstine}
Every coisotropic subspace of $k^{2n}$ of dimension $n+m$ is the image of a Lagrangian isometry $m\to n$.
\end{theorem}
\begin{proof}
Suppose that we have a coisotropic subspace $V^\omega$ of $k^{2n}$ with dimension $n+m$. 
Then $V$ is an isotropic subspace of $k^{2n}$ with dimension $n-m$.
Applying Fourier transforms, we obtain a  subspace symplectomorphic to $V$ generated by a matrix whose pivots are all in the $Z$ block.  Therefore, we can row reduce this matrix to obtain one of the following form:
$$
\left[
\begin{array}{cc|cc}
I_{n-m} & Z_B & X_A & X_B 
\end{array}\right]
$$
By applying $C_a$ gates from the first $n-m$ wires to the last $m$ wires, we obtain an isotropic subspace $V'$ generated by a matrix of the following form:
$$
\left[\begin{array}{cc|cc}
I_{n-m} & 0 & X_A' & X_B' 
\end{array}\right]
$$

Therefore we have $V' = U V$ for a symplectomorphism $U$.
Since all of the rows of the basis for $V'$  are orthogonal with respect to the symplectic form, we have:
\begin{align*}
0 &=
\left[\begin{array}{cc|cc}
I_{n-m} & 0 & X_A' & X_B' 
\end{array}\right]
\omega
\left[\begin{array}{cc|cc}
I_{n-m} & 0 & X_A' & X_B' 
\end{array}\right]^T\\
&=
\left[\begin{array}{cc|cc}
I_{n-m} & 0 & X_A' & X_B' 
\end{array}\right]
\left[\begin{array}{cc|cc}
 -X_A' & -X_B'  & I_{n-m} & 0
\end{array}\right]^T\\
&=
I_{n-m}(-X_A')^T +  0( -X_B' )^T +X_A'I_{n-m} + X_B' 0 \\
&=
(-X_A')^T +X_A'
\end{align*}

So that  $X_A'=(X_A')^T$ is a symmetric matrix. 
Therefore, the following matrix generates a graph state, and thus a Lagrangian subspace of $k^{2(n+m)}$:
$$
\left[\begin{array}{ccc|ccc}
I_{n-m} & 0    & 0 & X_A'       & X_B' & 0\\
0           & I_m & 0 & (X_B')^T & 0     &  I_m \\
0           & 0    & I_m  & 0            & I_m & 0
\end{array}\right]
$$

Let $W\subseteq k^{2(n+m)}$ be the Lagrangian subspace generated by this matrix.  Then
$$
%
$$

Therefore, $V^\omega$ is the image of the Lagrangian relation:

$$
%
$$

Which is an isometry.

\end{proof}

\begin{corollary}
\label{cor:affsymstine}
Every affine coisotropic subspace of $k^{2n}$ of dimension $n+m$ is the image of a an affine Lagrangian isometry $m\to n$.
\end{corollary}

\begin{remark}
\label{rem:xdisc}
More concretely the prop is $\Co\Isot\Rel_k$ is generated by adding 
$%
$ 
to the image of the embedding $\Lag\Rel_k\hookrightarrow \LinRel_k$; and dually, $\Isot\Rel_k$  is generated by adding 
$%
$ to $\Lag\Rel_k\hookrightarrow \LinRel_k$. So that the symplectic complement extends to an isomorphism of props $(-)^\omega:\Co\Isot\Rel_k\cong \Isot\Rel_k$.

The props $\Aff\Co\Isot\Rel_k$ and $\Aff\Isot\Rel_k$ are generated by respectively adding 
$%
$ 
to the images of the embeddings $\Aff\Co\Isot\Rel_k,\Aff\Isot\Rel_k\hookrightarrow \Aff\Rel_k$.  However, in this case, they are not isomorphic as:
$$
%
$$

\end{remark}

\begin{theorem}
$\CPM(\Aff\Lag\Rel_{k},\bar{(-)})\cong \Aff\Co\Isot\Rel_k$
\end{theorem}

\begin{proof}
Consider the identity on objects map $\CPM(\Aff\Lag\Rel_k, \bar{(-)}) \to \Aff\Co\Isot\Rel_k$, sending:
$$
%
$$
This is clearly bijective on objects.  This is full by Corollary \ref{cor:affsymstine}.
So we just have to show that it is faithful and well defined.
Take some $\hat f \in \CPM(\Aff\Lag\Rel_k, \bar{(-)})(0,n)$ with purification $f$.
Then for $\vec z_0, \vec x_0, \vec z_1, \vec x_1 \in k^n$:

\begin{align*}
(\vec z_0, \vec x_0,\vec z_1,\vec x_1) \in \hat f &\iff 
%
\\
&\iff (\vec z_0-\vec z_1, \vec x_0+\vec x_1,z_1-\vec z_0,\vec x_0+\vec x_1) \in \hat f 
\end{align*}
Therefore, maps in  $\CPM(\Aff\Lag\Rel_k, \bar{(-)})(0,n)$ are determined by their diagonal elements $(-\vec z,\vec x,\vec z,\vec x)$.
Take:

$$
%
$$

Therefore the following biconditional:

$$
%
$$

is equivalent to the following biconditional:

$$
%
$$

We prove the latter biconditional. If $g$ is the identity, then for  all $(\vec z,\vec x)\in k^{2n}$:

$$
%
$$

Given any map $g:m\to n$ in  $\Aff\Lag\Rel_k$, regarded as a state, we know that $(\vec z,\vec x)$ is a stabilizer for $g$ precisely when $(-\vec z,\vec x)$ is a stabilizer for $\bar g$.  However, we know that the cup discards these diagonal stabilizers on the bottom of $\bar g$ and the codiscard map discards all Weyl operators.

\end{proof}

\begin{corollary}
\label{cor:stabcode}
For odd prime $p$:
$$\Aff\Co\Isot\Rel_{\F_p}\cong \CPM(\Aff\Lag\Rel_{\F_p})\cong \CPM(\Stab_p)$$

That is, adding the discard relation to $\Aff\Lag\Rel_{\F_p}$ gives a semantics for odd prime dimensional  {\bf mixed stabilizer circuits} and { mixed circuits in Spekkens' qubit toy model}.  Graphically:

$$
\left\llbracket \
\begin{tikzpicture}[yscale=-1]
	\begin{pgfonlayer}{nodelayer}
		\node [style=none] (0) at (0.25, 0) {};
		\node [ground] (1) at (0.25, -0.5) {};
	\end{pgfonlayer}
	\begin{pgfonlayer}{edgelayer}
		\draw (1) to (0.center);
	\end{pgfonlayer}
\end{tikzpicture}
\ \right\rrbracket
=
\left\{ 
\left(
\begin{pmatrix}
z\\x
\end{pmatrix},
*
\right)
:\forall z,x \in \F_p
\right\}
=
\left\llbracket \
\begin{tikzpicture}
	\begin{pgfonlayer}{nodelayer}
		\node [style=Z] (2) at (4, 0) {};
		\node [style=Z] (3) at (4.5, 0) {};
		\node [style=none] (4) at (4, -1) {};
		\node [style=none] (5) at (4.5, -1) {};
	\end{pgfonlayer}
	\begin{pgfonlayer}{edgelayer}
		\draw (4.center) to (2);
		\draw (3) to (5.center);
	\end{pgfonlayer}
\end{tikzpicture}
\ \right\rrbracket
$$

\end{corollary}
By mixed stabilizer circuits; we mean stabilizer circuits which can be obtained by discarding part of a pure stabilizer circuit, not an arbitrary convex combination of stabilizer circuits.  I point this out only because the  terms ``stabilizer mixed state/circuit'' and ``mixed stabilizer state/circuit'' can interchangeably refer to either notion.
Mixed stabilizer states are also known as {\bf stabilizer codes} for reasons that will become clear in Section \ref{sec:qec}.

This is conceptually very close to the result of \cite{huot} where they show that the affine completion of isometries between finite dimensional Hilbert spaces, which freely adds a discard map, yields completely positive trace-preserving maps.  And even more similar to the more refined result of   \cite{disc}, where they show that various fragments of the ZX-calculus can be augmented with quantum discarding by freely adding a generator which discards isometries. 

Although in our case the quantum discarding is interpreted as the literal {\em discard relation} of the corresponding stabilizers, therefore our semantics never departs from affine relations.
The discard relation on $n$ wires is {\em defined} to be the morphism $!_n$ such that for all morphisms $f:n\to m$,  $!_m\circ f \leq !_n$.

It was already known that stabilizer codes are in bijection with  affine isotropic subspaces, for example \cite[\S A]{gross}.  Indeed affine coisotropic subspaces are in bijection with affine isotropic subspace by  taking the symplectic complement of the linear component of the affine subspace; however, as noted in Remark  \ref{rem:xdisc}, $\Aff\Isot\Rel_{\F_p} \not \cong \Aff\Co\Isot\Rel_{\F_p}$, so their compositions as affine relations are different. The interpretation of the doubled cozero as the quantum discard map is not sound with respect to relational composition.   This is closely related to the discrete Wigner function of Gross which we will comment further on in Remark \ref{rem:wig}.

This formalizes the relationship between mixed stabilizer circuits and stabilizer tableaux with not-necessarily-full rank in  a compositional way.  It is standard, and indeed very useful to interpret a stabilizer state in terms of its stabilizer tableau:  which is an augmented  basis $L+a$ of an affine isotropic subspace.  In order to compose these tableaux, one must take the symplectic complement of the linear component of both spaces, and then compute  their relational composition as affine coisotropic relations.  Then the stabilizer tableau for this composite is obtained by taking the symplectic complement on the linear component once again.

{\em Absolutely remarkably}, and seemingly out of nowhere, the Weyl-free mixed odd prime dimensional stabilizer circuits modulo invertible scalars can be expressed in terms of an iterated $\CPM$ construction with respect to the orthogonal complement at the inner level, and the complex conjugation at the outer level:
\begin{corollary}
Given a prime $p$ and $k=\F_p$ or $k=\mathbb{Q}$:
$$\Isot\Rel_{k}\cong\Co\Isot\Rel_{k}\cong \CPM(\CPM(\LinRel_{k},(-)^\perp),\bar{(-)})$$
\end{corollary}
The astounding symmetry involved here begs the question if iterating the $\CPM$ construction more times yields anything physically interesting. Perhaps the work of \cite{CPMho} can shed some light on this question. 

In order to add measurement and state preparation in the symplectic setting we inspect the structure of the $Z$ and $X$ projectors in the symplectic setting:
\begin{definition}
The $Z$ and $X$ projectors are defined as follows in $\Aff\Co\Isot\Rel_{k}$:
$$
p_Z:=
\begin{tikzpicture}
	\begin{pgfonlayer}{nodelayer}
		\node [style=X] (0) at (0.5, -0.75) {};
		\node [style=none] (2) at (0.25, 0) {};
		\node [style=none] (4) at (1.25, 0.5) {};
		\node [style=Z] (5) at (0.75, 0) {};
		\node [style=Z] (6) at (1.75, 0) {};
		\node [style=Z] (7) at (1.5, -0.75) {};
		\node [style=none] (8) at (0.75, 0) {};
		\node [style=none] (9) at (1.25, 0) {};
		\node [style=none] (10) at (1.75, 0) {};
		\node [style=none] (11) at (0.5, -1.5) {};
		\node [style=none] (13) at (1.5, -1.5) {};
		\node [style=none] (14) at (0.25, 0.5) {};
	\end{pgfonlayer}
	\begin{pgfonlayer}{edgelayer}
		\draw [in=-90, out=120] (7) to (9.center);
		\draw (7) to (13.center);
		\draw [in=60, out=-90] (10.center) to (7);
		\draw [in=60, out=-90] (8.center) to (0);
		\draw [in=-90, out=120] (0) to (2.center);
		\draw (0) to (11.center);
		\draw (9.center) to (4.center);
		\draw (2.center) to (14.center);
	\end{pgfonlayer}
\end{tikzpicture}
=
\begin{tikzpicture}
	\begin{pgfonlayer}{nodelayer}
		\node [style=none] (4) at (1, 0.5) {};
		\node [style=Z] (5) at (0.25, 0) {};
		\node [style=none] (9) at (1, -1.25) {};
		\node [style=none] (11) at (0.25, -1.25) {};
		\node [style=none] (14) at (0.25, 0.5) {};
		\node [style=Z] (15) at (0.25, -0.75) {};
	\end{pgfonlayer}
	\begin{pgfonlayer}{edgelayer}
		\draw (9.center) to (4.center);
		\draw (14.center) to (5);
		\draw (11.center) to (15);
	\end{pgfonlayer}
\end{tikzpicture}
\hspace*{.5cm}
p_X:=
\begin{tikzpicture}
	\begin{pgfonlayer}{nodelayer}
		\node [style=Z] (0) at (0.5, -0.75) {};
		\node [style=none] (2) at (0.25, 0) {};
		\node [style=none] (4) at (1.25, 0.5) {};
		\node [style=Z] (5) at (0.75, 0) {};
		\node [style=Z] (6) at (1.75, 0) {};
		\node [style=X] (7) at (1.5, -0.75) {};
		\node [style=none] (8) at (0.75, 0) {};
		\node [style=none] (9) at (1.25, 0) {};
		\node [style=none] (10) at (1.75, 0) {};
		\node [style=none] (11) at (0.5, -1.5) {};
		\node [style=none] (13) at (1.5, -1.5) {};
		\node [style=none] (14) at (0.25, 0.5) {};
	\end{pgfonlayer}
	\begin{pgfonlayer}{edgelayer}
		\draw [in=-90, out=120] (7) to (9.center);
		\draw (7) to (13.center);
		\draw [in=60, out=-90] (10.center) to (7);
		\draw [in=60, out=-90] (8.center) to (0);
		\draw [in=-90, out=120] (0) to (2.center);
		\draw (0) to (11.center);
		\draw (9.center) to (4.center);
		\draw (2.center) to (14.center);
	\end{pgfonlayer}
\end{tikzpicture}
=
\begin{tikzpicture}[scale=-1]
	\begin{pgfonlayer}{nodelayer}
		\node [style=none] (4) at (1, 0.5) {};
		\node [style=Z] (5) at (0.25, 0) {};
		\node [style=none] (9) at (1, -1.25) {};
		\node [style=none] (11) at (0.25, -1.25) {};
		\node [style=none] (14) at (0.25, 0.5) {};
		\node [style=Z] (15) at (0.25, -0.75) {};
	\end{pgfonlayer}
	\begin{pgfonlayer}{edgelayer}
		\draw (9.center) to (4.center);
		\draw (14.center) to (5);
		\draw (11.center) to (15);
	\end{pgfonlayer}
\end{tikzpicture}
$$
\end{definition}
The $Z$ projector discards and then codiscards the $Z$-gradient: cutting the $Z$ gradient in two so that no information is preserved, while acting trivially on the $X$ gradient.  Dually for the $X$ projector.  Concretely, these are interpreted as the following relations:

$$
\left\llbracket\
\begin{tikzpicture}
	\begin{pgfonlayer}{nodelayer}
		\node [style=none] (4) at (1, 0.5) {};
		\node [style=Z] (5) at (0.25, 0) {};
		\node [style=none] (9) at (1, -1.25) {};
		\node [style=none] (11) at (0.25, -1.25) {};
		\node [style=none] (14) at (0.25, 0.5) {};
		\node [style=Z] (15) at (0.25, -0.75) {};
	\end{pgfonlayer}
	\begin{pgfonlayer}{edgelayer}
		\draw (9.center) to (4.center);
		\draw (14.center) to (5);
		\draw (11.center) to (15);
	\end{pgfonlayer}
\end{tikzpicture}
\ \right\rrbracket
=
\left\{
\left(
\begin{bmatrix}
z\\x
\end{bmatrix},
\begin{bmatrix}
z'\\x
\end{bmatrix}\right)
 \in k^{2+2} \ | \
\forall z,z',x \in k
\right\}
$$
$$
\left\llbracket\
\begin{tikzpicture}[xscale=-1]
	\begin{pgfonlayer}{nodelayer}
		\node [style=none] (4) at (1, 0.5) {};
		\node [style=Z] (5) at (0.25, 0) {};
		\node [style=none] (9) at (1, -1.25) {};
		\node [style=none] (11) at (0.25, -1.25) {};
		\node [style=none] (14) at (0.25, 0.5) {};
		\node [style=Z] (15) at (0.25, -0.75) {};
	\end{pgfonlayer}
	\begin{pgfonlayer}{edgelayer}
		\draw (9.center) to (4.center);
		\draw (14.center) to (5);
		\draw (11.center) to (15);
	\end{pgfonlayer}
\end{tikzpicture}
\ \right\rrbracket
=
\left\{
\left(
\begin{bmatrix}
z\\x
\end{bmatrix},
\begin{bmatrix}
z\\x'
\end{bmatrix}\right)
 \in k^{2+2} \ | \
\forall z,x,z' \in k
\right\}
$$

We split either of these projectors, following \cite{idempotent} to add classical types:

\begin{definition}
Let $\Aff\Co\Isot\Rel_k^M$ denote the two-coloured prop generated by splitting  $p_Z$ in $\Aff\Co\Isot\Rel_k$.

Concretely $\Aff\Co\Isot\Rel_k^M$  has:
\begin{description}
\item[\ \ Objects:]  pairs $(n, e)$ where $n \in \N$ and $e:n\to n$ is a finite tensor product of $1$ and $p_Z$.

\item[\ \ Maps:] $(n,e) \to (m, g)$ are maps $f:n\to m$  in  $\Aff\Co\Isot\Rel_k$ such that $gfe=f$.
\end{description}

The object $Q=(1,1_1)$ is interpreted as a quantum channel and the object $C=(1,p_Z)$ as a classical channel.
\end{definition}

We could have instead split $p_X$, or split both $p_X$ and $p_Z$; however, all three of these multicoloured props are equivalent.  This equivalence is witnessed via the Fourier transform. Indeed this suffices to split all nonzero projectors up to isomorphism because all projectors of the same dimension are affine symplectomorphic. 
It is important to remark that the choice of projectors which are split effects the code-distance, because code-distance is basis dependent, and not invariant under equivalence.

This category has a succinct presentation;  adding the affine relations to  $\Aff\Co\Isot\Rel_k$ obtained by cutting/splitting the $Z$ projector in two:
\begin{theorem}
\label{thm:twocoloured}
The full subcategory of $\Aff\Co\Isot\Rel_k^M$ generated by tensor powers of $C$ is isomorphic to $\Aff\Rel_k$.
Therefore $\Aff\Co\Isot\Rel_k^M$ is isomorphic to adding the following linear relations to the image of $\Aff\Co\Isot\Rel_k^M\to \Aff\Co\Isot\Rel_k$ in the way which makes this into a two-coloured prop:

$$
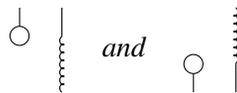
\begin{tikzpicture}
	\begin{pgfonlayer}{nodelayer}
		\node [style=none] (0) at (24.75, 0.5) {};
		\node [style=Z] (1) at (24, 0) {};
		\node [style=none] (2) at (24.75, 0) {};
		\node [style=none] (3) at (24, 0.5) {};
		\node [style=none] (4) at (24.75, -1) {};
	\end{pgfonlayer}
	\begin{pgfonlayer}{edgelayer}
		\draw (2.center) to (0.center);
		\draw (3.center) to (1);
		\draw [style=classicalwire] (4.center) to (2.center);
	\end{pgfonlayer}
\end{tikzpicture}
\hspace*{.5cm}\text{and}\hspace*{.5cm}
\begin{tikzpicture}
	\begin{pgfonlayer}{nodelayer}
		\node [style=none] (0) at (24.75, -1) {};
		\node [style=Z] (1) at (24, -0.5) {};
		\node [style=none] (2) at (24.75, -0.5) {};
		\node [style=none] (3) at (24, -1) {};
		\node [style=none] (4) at (24.75, 0.5) {};
	\end{pgfonlayer}
	\begin{pgfonlayer}{edgelayer}
		\draw (2.center) to (0.center);
		\draw (3.center) to (1);
		\draw [style=classicalwire] (4.center) to (2.center);
	\end{pgfonlayer}
\end{tikzpicture}
$$
\end{theorem}
These correspond to the state preparation and measurement processes in the $Z$ basis.
We draw the wire associated to $C$ as a coil to indicate the type (although this is just syntactic sugar).
In the quantum setting, the classical state ``lives'' on a single wire and the stabilizer state ``lives'' on the doubled wires.  The two coloured notation is more evident when we use the following:

\begin{definition}
There is a graphical calculus for two-sorted prop of classical and quantum types following \cite{pqp}.  Although, here it works for arbitrary field $k$.
The classical wires are drawn thin and the quantum wires are drawn thick.  The thin bordered string diagrams are generated by affine relations.  Moreover, the pure affine Lagrangian relations are drawn with a thick border, so that:
$$
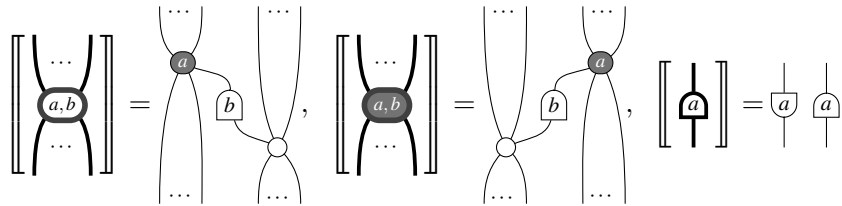
\left\llbracket\

$$
The state preparation and measurement are drawn as thin spiders interfacing between the quantum and classical systems:
$$
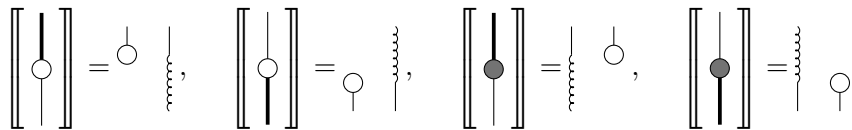
\left\llbracket \
%
$$

For example, preparing the state $|x\rangle$ corresponds to the following manipulation:

$$
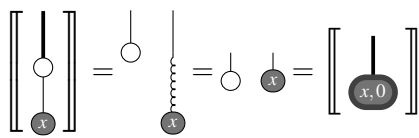
\left\llbracket \
%
\
\right\rrbracket 
$$

The state preparation and discarding in the $X$-basis are obtained by composition of these morphisms with the Fourier transform; yielding morphisms which discard the $X$ wire instead of the $Z$ wire:
$$
%
$$

The conjugation of pure spiders by state preparation and measurement maps creates {\bf bastard spiders} following \cite{pqp}.  Bastard spiders are drawn with a thin border; for example:
$$
%
$$

For example, the thin border on the white spider indicates that the state has been measured in the $Z$-basis and is a stochastic mixture of the $Z$-basis elements.
When a pure spider is connected to a bastard spider they both fuse into a bastard spider where the linear phase is discarded:
$$
%
$$

The discarding map is given by a bastard spider with one input and no outputs:
$$
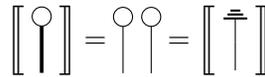
\left\llbracket\ 
%
\ \right\rrbracket
$$

Notice that this way of describing the discard map is independent of the choice of orthonormal basis, so that:

$$
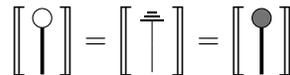
\left\llbracket\ 
%
\ \right\rrbracket
$$
\end{definition}

\begin{remark}
\label{rem:wig}
Compare this relational semantics of measurement to the discrete Wigner function of Gross \cite{gross}; which is the  discrete version of the Wigner's quasi-probabilty distribution \cite{Wigner1932}.  Gross shows that on odd prime dimensional stabilizer states, this is a probability distribution given by the indicator function of an affine Lagrangian subspace  \cite[Lemma 9]{gross} (actually it works just as well for stabilizer codes/affine coistropic subspaces):
$$
P_{G(L+a)}:\F_p^{2n}\to [0,1);\ (z,x)\mapsto
1/|L| \cdot \delta_{L+a}(z,x)
=
\begin{cases}
1/|L| & \text{ if } (z,x) \in L+a\\
0 & \text{ otherwise}
\end{cases}
$$
This probability distribution is uniform so that every outcome is equally likely.  For example, the measurement of a stabilizer state $G(L+a)$ in the $Z$ basis produces outcome $|x\rangle$ with probability:
$$
\sum_{z\in \F_p} P_{G(L+a)}(z,x)
$$
Moreover, this marginalization over $z$ acts backwards on the state; so that in accordance with Spekkens toy model, the observer gains at most one pit of information by sampling $x$, but injects at most one pit of uncertainty back into the ontic state.

In the conventional, functional approaches to measurement of stabilizer states/states in Spekkens' toy model, for example \cite{catani}, they construct ``measurement update rules'' to compute this back action on the state.  This is essentially the same kind of procedure performed to compute the effect of  measuring a qubit stabilizer states \cite{aaronson}.  However, in the relational paradigm, the measurement outcome, the measurement update, the state preparation and the unitary evolution of the quantum state are all computed in the same way: by taking the relational composition.
\end{remark}

We have a simple topological proof of the following well-known result:

\begin{remark}
Preparing a state in the $Z$ basis and measureing in the $X$ basis preserves no information:
$$
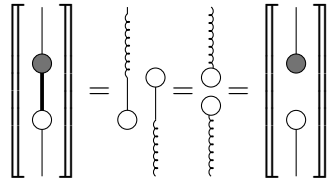
\left\llbracket\
\begin{tikzpicture}
	\begin{pgfonlayer}{nodelayer}
		\node [style=Z] (12) at (31, -5) {};
		\node [style=none] (13) at (31, -6) {};
		\node [style=none] (14) at (31, -3) {};
		\node [style=X] (15) at (31, -4) {};
	\end{pgfonlayer}
	\begin{pgfonlayer}{edgelayer}
		\draw (13.center) to (12);
		\draw (15) to (14.center);
		\draw [thick] (12) to (15);
	\end{pgfonlayer}
\end{tikzpicture}
\
\right\rrbracket
=
\begin{tikzpicture}
	\begin{pgfonlayer}{nodelayer}
		\node [style=Z] (4) at (60, 5.5) {};
		\node [style=none] (5) at (60, 6.25) {};
		\node [style=Z] (6) at (60.5, 6.25) {};
		\node [style=none] (7) at (60.5, 5.5) {};
		\node [style=none] (8) at (60, 7.5) {};
		\node [style=none] (9) at (60.5, 4.5) {};
	\end{pgfonlayer}
	\begin{pgfonlayer}{edgelayer}
		\draw (5.center) to (4);
		\draw (7.center) to (6);
		\draw [style=classicalwire] (5.center) to (8.center);
		\draw [style=classicalwire] (9.center) to (7.center);
	\end{pgfonlayer}
\end{tikzpicture}
=
\begin{tikzpicture}
	\begin{pgfonlayer}{nodelayer}
		\node [style=Z] (0) at (59, 6.25) {};
		\node [style=none] (1) at (59, 7.5) {};
		\node [style=Z] (2) at (59, 5.75) {};
		\node [style=none] (3) at (59, 4.5) {};
	\end{pgfonlayer}
	\begin{pgfonlayer}{edgelayer}
		\draw [style=classicalwire] (1.center) to (0);
		\draw [style=classicalwire] (3.center) to (2);
	\end{pgfonlayer}
\end{tikzpicture}
=
\left\llbracket\
\begin{tikzpicture}
	\begin{pgfonlayer}{nodelayer}
		\node [style=Z] (0) at (31, -5) {};
		\node [style=none] (1) at (31, -6) {};
		\node [style=none] (2) at (31, -3) {};
		\node [style=X] (3) at (31, -4) {};
	\end{pgfonlayer}
	\begin{pgfonlayer}{edgelayer}
		\draw (1.center) to (0);
		\draw (3) to (2.center);
	\end{pgfonlayer}
\end{tikzpicture}
\
\right\rrbracket
$$
\end{remark}
For this reason, we can prove the correctness of the prime dimensional quantum teleportation algorithm discussed previously using only spider fusion:
\begin{example}
\label{ex:teleportation}
Given any prime $p$, the following equations of string diagrams in $\Aff\Co\Isot\Rel_{\F_p}^M$ proves the correctness of the quantum teleportation protocol where Alice on the left teleports a qudit to Bob, on the right. They share a Bell state  (on the bottom of the diagram)  and two classical dits (drawn as coiled wires).  
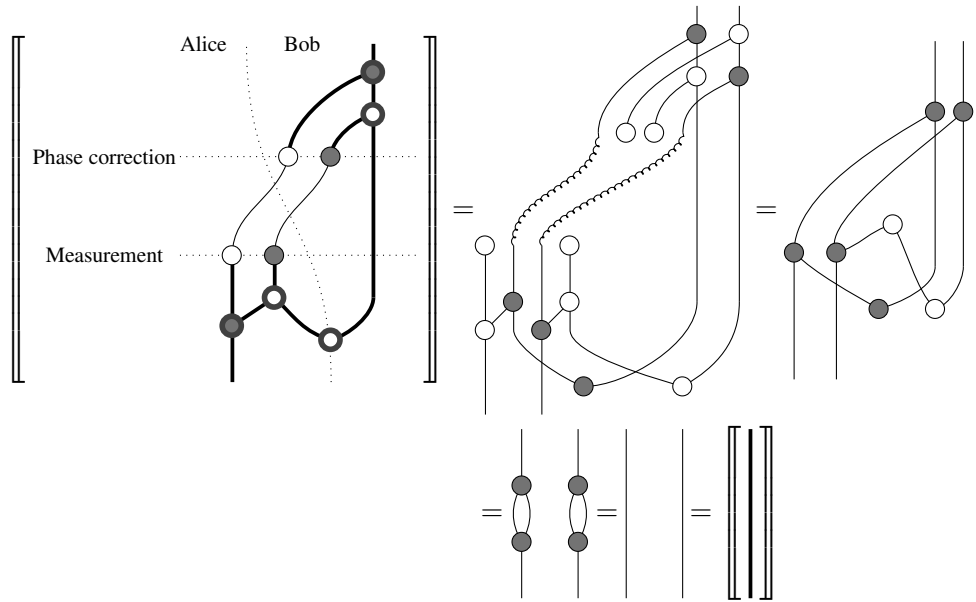
\begin{align*}
\left\llbracket\
\begin{tikzpicture}
	\begin{pgfonlayer}{nodelayer}
		\node [style=none] (20) at (380.75, 9.5) {};
		\node [style=none] (21) at (382.25, 3.5) {};
		\node [style=none] (22) at (380, 9.5) {Alice};
		\node [style=none] (23) at (381.75, 9.5) {Bob};
		\node [style=none] (24) at (379.5, 5.75) {};
		\node [style=none] (25) at (383.75, 5.75) {};
		\node [style=none] (26) at (379.5, 7.5) {};
		\node [style=none] (27) at (383.75, 7.5) {};
		\node [style=none] (28) at (378.25, 7.5) {Phase correction};
		\node [style=none] (29) at (378.25, 5.75) {Measurement};
		\node [style=Xthick] (30) at (380.5, 4.5) {};
		\node [style=Zthick] (31) at (381.25, 5) {};
		\node [style=X] (32) at (381.25, 5.75) {};
		\node [style=Z] (33) at (380.5, 5.75) {};
		\node [style=none] (34) at (380.5, 3.5) {};
		\node [style=Zthick] (35) at (382.25, 4.25) {};
		\node [style=none] (36) at (383, 5) {};
		\node [style=none] (37) at (383, 9.5) {};
		\node [style=Zthick] (38) at (383, 8.25) {};
		\node [style=Xthick] (39) at (383, 9) {};
		\node [style=X] (40) at (382.25, 7.5) {};
		\node [style=Z] (41) at (381.5, 7.5) {};
	\end{pgfonlayer}
	\begin{pgfonlayer}{edgelayer}
		\draw [style=dotted, in=-90, out=90, looseness=1.25] (21.center) to (20.center);
		\draw [style=dotted] (25.center) to (24.center);
		\draw [style=dotted] (27.center) to (26.center);
		\draw [style=thick] (34.center) to (30);
		\draw [style=thick] (30) to (31);
		\draw [style=thick] (30) to (33);
		\draw [style=thick] (31) to (32);
		\draw [style=thick, in=-45, out=165] (35) to (31);
		\draw [style=thick] (37.center) to (39);
		\draw [style=thick] (39) to (38);
		\draw [style=thick] (38) to (36.center);
		\draw [style=thick, in=15, out=-90, looseness=0.75] (36.center) to (35);
		\draw [style=thick, in=-165, out=90, looseness=0.75] (40) to (38);
		\draw [style=thick, in=90, out=-165, looseness=0.75] (39) to (41);
		\draw [in=-105, out=90] (32) to (40);
		\draw [in=90, out=-105] (41) to (33);
	\end{pgfonlayer}
\end{tikzpicture}
\ \right\rrbracket
=&
\begin{tikzpicture}
	\begin{pgfonlayer}{nodelayer}
		\node [style=none] (0) at (366.75, 1.75) {};
		\node [style=none] (1) at (367.25, 3.25) {};
		\node [style=none] (2) at (366.25, 3.25) {};
		\node [style=none] (3) at (365.75, 1.75) {};
		\node [style=Z] (4) at (367.25, 4.75) {};
		\node [style=Z] (5) at (365.75, 4.75) {};
		\node [style=none] (6) at (369.5, 3.75) {};
		\node [style=none] (7) at (370.25, 3.75) {};
		\node [style=none] (8) at (369.5, 9) {};
		\node [style=none] (9) at (370.25, 9) {};
		\node [style=Z] (10) at (365.75, 3.25) {};
		\node [style=X] (11) at (366.25, 3.75) {};
		\node [style=Z] (12) at (367.25, 3.75) {};
		\node [style=X] (13) at (366.75, 3.25) {};
		\node [style=none] (14) at (366.25, 4.75) {};
		\node [style=none] (15) at (366.75, 4.75) {};
		\node [style=X] (16) at (367.5, 2.25) {};
		\node [style=Z] (17) at (369.25, 2.25) {};
		\node [style=Z] (18) at (369.5, 7.75) {};
		\node [style=X] (19) at (370.25, 7.75) {};
		\node [style=X] (20) at (369.5, 8.5) {};
		\node [style=Z] (21) at (370.25, 8.5) {};
		\node [style=Z] (22) at (368.25, 6.75) {};
		\node [style=Z] (23) at (368.75, 6.75) {};
		\node [style=none] (24) at (369.25, 6.75) {};
		\node [style=none] (25) at (367.75, 6.75) {};
	\end{pgfonlayer}
	\begin{pgfonlayer}{edgelayer}
		\draw (11) to (10);
		\draw (13) to (12);
		\draw (1.center) to (12);
		\draw (12) to (4);
		\draw (3.center) to (10);
		\draw (10) to (5);
		\draw (2.center) to (11);
		\draw (0.center) to (13);
		\draw (13) to (15.center);
		\draw (11) to (14.center);
		\draw [in=0, out=-90, looseness=0.75] (6.center) to (16);
		\draw [in=-90, out=165, looseness=0.50] (16) to (2.center);
		\draw [in=165, out=-90, looseness=0.50] (1.center) to (17);
		\draw [in=-90, out=30] (17) to (7.center);
		\draw (21) to (19);
		\draw (20) to (18);
		\draw (6.center) to (18);
		\draw (20) to (8.center);
		\draw (9.center) to (21);
		\draw (19) to (7.center);
		\draw [in=-150, out=90, looseness=0.75] (22) to (21);
		\draw [in=-150, out=90] (23) to (18);
		\draw [style=classicalwire, in=-90, out=90, looseness=0.50] (15.center) to (24.center);
		\draw [style=classicalwire, in=-90, out=90] (14.center) to (25.center);
		\draw [in=210, out=90, looseness=0.75] (25.center) to (20);
		\draw [in=-150, out=90] (24.center) to (19);
	\end{pgfonlayer}
\end{tikzpicture}
=
\begin{tikzpicture}
	\begin{pgfonlayer}{nodelayer}
		\node [style=none] (0) at (82, 2.25) {};
		\node [style=none] (2) at (81.25, 2.25) {};
		\node [style=none] (3) at (83.75, 4.25) {};
		\node [style=none] (4) at (84.25, 4.25) {};
		\node [style=none] (5) at (83.75, 8.25) {};
		\node [style=none] (6) at (84.25, 8.25) {};
		\node [style=X] (7) at (81.25, 4.5) {};
		\node [style=X] (8) at (82.75, 3.5) {};
		\node [style=Z] (9) at (83.75, 3.5) {};
		\node [style=X] (10) at (84.25, 7) {};
		\node [style=X] (11) at (83.75, 7) {};
		\node [style=Z] (12) at (83, 5) {};
		\node [style=X] (13) at (82, 4.5) {};
	\end{pgfonlayer}
	\begin{pgfonlayer}{edgelayer}
		\draw [in=15, out=-90, looseness=0.75] (3.center) to (8);
		\draw [in=-90, out=15, looseness=0.75] (9) to (4.center);
		\draw (11) to (5.center);
		\draw (10) to (4.center);
		\draw (10) to (6.center);
		\draw (3.center) to (11);
		\draw (7) to (2.center);
		\draw [in=-150, out=90, looseness=0.75] (7) to (11);
		\draw [in=225, out=120, looseness=0.50] (13) to (10);
		\draw [in=15, out=-165] (12) to (13);
		\draw [in=270, out=90] (0.center) to (13);
		\draw [in=-60, out=135, looseness=0.75] (9) to (12);
		\draw [in=-45, out=150, looseness=0.75] (8) to (7);
	\end{pgfonlayer}
\end{tikzpicture}\\
&=
\begin{tikzpicture}
	\begin{pgfonlayer}{nodelayer}
		\node [style=none] (15) at (86.75, 5) {};
		\node [style=none] (16) at (86.75, 8) {};
		\node [style=X] (18) at (86.75, 6) {};
		\node [style=X] (21) at (86.75, 7) {};
		\node [style=none] (22) at (87.75, 5) {};
		\node [style=none] (23) at (87.75, 8) {};
		\node [style=X] (24) at (87.75, 6) {};
		\node [style=X] (25) at (87.75, 7) {};
	\end{pgfonlayer}
	\begin{pgfonlayer}{edgelayer}
		\draw (21) to (16.center);
		\draw (18) to (15.center);
		\draw [bend left] (21) to (18);
		\draw [bend left] (18) to (21);
		\draw (25) to (23.center);
		\draw (24) to (22.center);
		\draw [bend left] (25) to (24);
		\draw [bend left] (24) to (25);
	\end{pgfonlayer}
\end{tikzpicture}
=
\begin{tikzpicture}
	\begin{pgfonlayer}{nodelayer}
		\node [style=none] (26) at (88.75, 5) {};
		\node [style=none] (27) at (88.75, 8) {};
		\node [style=none] (30) at (89.75, 5) {};
		\node [style=none] (31) at (89.75, 8) {};
	\end{pgfonlayer}
	\begin{pgfonlayer}{edgelayer}
		\draw (26.center) to (27.center);
		\draw (30.center) to (31.center);
	\end{pgfonlayer}
\end{tikzpicture}
=
\left\llbracket\
\begin{tikzpicture}
	\begin{pgfonlayer}{nodelayer}
		\node [style=none] (30) at (89.75, 5) {};
		\node [style=none] (31) at (89.75, 8) {};
	\end{pgfonlayer}
	\begin{pgfonlayer}{edgelayer}
		\draw [thick] (30.center) to (31.center);
	\end{pgfonlayer}
\end{tikzpicture}
\ \right\rrbracket
\end{align*}
\end{example}

Compare this with the graphical proof contained in \cite[p. 706]{pqp} which uses the Hopf-algebra structure.

Because $\Aff\Co\Isot\Rel_{\F_p}^M$ is a subcategory of relations, composable maps are ordered by subspace inclusion (ie, it is poset-enriched). Moreover, since all possible outcomes are equally likely we can identify when the measurement statistics of one process arise from the marginalization of the measurement statistics of another process:

\begin{definition}
Take two odd prime dimensional stabilizer circuits with state preparations and measurement $f,g$ interpreted as parallel maps in  $\Aff\Co\Isot\Rel_{\F_p}^M$.
Then $f$ is a {\bf coarse-graining} of $g$ when $f \subset g$ is a (strict)  affine subspace.
\end{definition}
\begin{example}
For an extreme example, the circuit obtained by preparing in the $Z$-basis and measuring in the $X$-basis  is a coarse-graining of the identity circuit on a classical wire:
$$
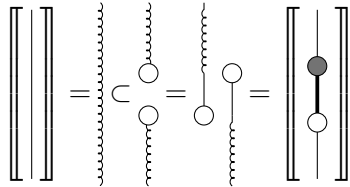
\left\llbracket\
\begin{tikzpicture}
	\begin{pgfonlayer}{nodelayer}
		\node [style=none] (13) at (31, -6) {};
		\node [style=none] (14) at (31, -3) {};
	\end{pgfonlayer}
	\begin{pgfonlayer}{edgelayer}
		\draw (13.center) to (14.center);
	\end{pgfonlayer}
\end{tikzpicture}
\
\right\rrbracket
=
\begin{tikzpicture}
	\begin{pgfonlayer}{nodelayer}
		\node [style=none] (0) at (59, 8.25) {};
		\node [style=none] (1) at (59, 5) {};
	\end{pgfonlayer}
	\begin{pgfonlayer}{edgelayer}
		\draw [style=classicalwire] (0.center) to (1.center);
	\end{pgfonlayer}
\end{tikzpicture}
\subset
\begin{tikzpicture}
	\begin{pgfonlayer}{nodelayer}
		\node [style=Z] (2) at (60, 7) {};
		\node [style=none] (3) at (60, 8.25) {};
		\node [style=Z] (4) at (60, 6.25) {};
		\node [style=none] (5) at (60, 5) {};
	\end{pgfonlayer}
	\begin{pgfonlayer}{edgelayer}
		\draw [style=classicalwire] (3.center) to (2);
		\draw [style=classicalwire] (5.center) to (4);
	\end{pgfonlayer}
\end{tikzpicture}
=
\begin{tikzpicture}
	\begin{pgfonlayer}{nodelayer}
		\node [style=Z] (6) at (61, 6.25) {};
		\node [style=none] (7) at (61, 7) {};
		\node [style=Z] (8) at (61.5, 7) {};
		\node [style=none] (9) at (61.5, 6.25) {};
		\node [style=none] (10) at (61, 8.25) {};
		\node [style=none] (11) at (61.5, 5) {};
	\end{pgfonlayer}
	\begin{pgfonlayer}{edgelayer}
		\draw (7.center) to (6);
		\draw (9.center) to (8);
		\draw [style=classicalwire] (11.center) to (9.center);
		\draw [style=classicalwire] (7.center) to (10.center);
	\end{pgfonlayer}
\end{tikzpicture}
=
\left\llbracket\
\begin{tikzpicture}
	\begin{pgfonlayer}{nodelayer}
		\node [style=Z] (12) at (31, -5) {};
		\node [style=none] (13) at (31, -6) {};
		\node [style=none] (14) at (31, -3) {};
		\node [style=X] (15) at (31, -4) {};
	\end{pgfonlayer}
	\begin{pgfonlayer}{edgelayer}
		\draw (13.center) to (12);
		\draw (15) to (14.center);
		\draw [thick] (12) to (15);
	\end{pgfonlayer}
\end{tikzpicture}
\
\right\rrbracket
$$
This is because, given any input state, the circuit on the right hand side can produce any output state; however, the identity circuit forces the inputs to be the same as the outputs.  
\end{example}
\begin{example}
Similarly,  the decoherence map is a coarse-graining of  the identity on a quantum wire :
$$
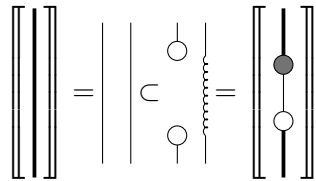
\left\llbracket\
\begin{tikzpicture}
	\begin{pgfonlayer}{nodelayer}
		\node [style=none] (13) at (31, -6) {};
		\node [style=none] (14) at (31, -3) {};
	\end{pgfonlayer}
	\begin{pgfonlayer}{edgelayer}
		\draw [thick] (13.center) to (14.center);
	\end{pgfonlayer}
\end{tikzpicture}
\
\right\rrbracket
=
\begin{tikzpicture}
	\begin{pgfonlayer}{nodelayer}
		\node [style=none] (0) at (69.25, 7.25) {};
		\node [style=none] (1) at (69.25, 4.75) {};
		\node [style=none] (2) at (69.75, 7.25) {};
		\node [style=none] (3) at (69.75, 4.75) {};
	\end{pgfonlayer}
	\begin{pgfonlayer}{edgelayer}
		\draw (3.center) to (2.center);
		\draw (1.center) to (0.center);
	\end{pgfonlayer}
\end{tikzpicture}
\subset
\begin{tikzpicture}
	\begin{pgfonlayer}{nodelayer}
		\node [style=Z] (4) at (70.75, 6.75) {};
		\node [style=none] (5) at (70.75, 7.25) {};
		\node [style=Z] (6) at (70.75, 5.25) {};
		\node [style=none] (7) at (70.75, 4.75) {};
		\node [style=none] (8) at (71.25, 7.25) {};
		\node [style=none] (9) at (71.25, 4.75) {};
		\node [style=none] (10) at (71.25, 6.75) {};
		\node [style=none] (11) at (71.25, 5.25) {};
	\end{pgfonlayer}
	\begin{pgfonlayer}{edgelayer}
		\draw (5.center) to (4);
		\draw (7.center) to (6);
		\draw [style=classicalwire] (11.center) to (10.center);
		\draw (9.center) to (11.center);
		\draw (10.center) to (8.center);
	\end{pgfonlayer}
\end{tikzpicture}
=
\left\llbracket\
\begin{tikzpicture}
	\begin{pgfonlayer}{nodelayer}
		\node [style=Z] (12) at (31, -5) {};
		\node [style=none] (13) at (31, -6) {};
		\node [style=none] (14) at (31, -3) {};
		\node [style=X] (15) at (31, -4) {};
	\end{pgfonlayer}
	\begin{pgfonlayer}{edgelayer}
		\draw [thick] (13.center) to (12);
		\draw [thick] (15) to (14.center);
		\draw  (12) to (15);
	\end{pgfonlayer}
\end{tikzpicture}
\
\right\rrbracket
$$
\end{example}

\subsection{Electrical circuits with control and measurement}
This multicoloured view of things recaptures the  semantics for controlled voltage and current sources as well as ammeters and voltmenters given in \cite{impedence}; however the situation is more nuanced than for quantum circuits.
\label{rem:electrical}
$\Aff\Co\Isot\Rel_{\R(x)}^M$ gives a semantics for all of the electrical circuit components we have discussed so far in addition to the controlled voltage and current sources:
$$
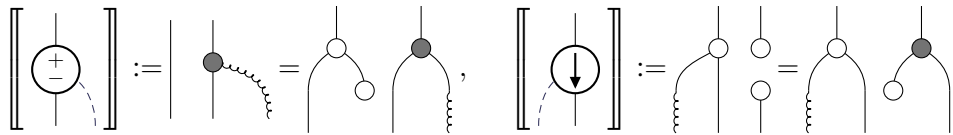
\left\llbracket\
\begin{tikzpicture}
	\begin{pgfonlayer}{nodelayer}
		\node [style=none] (0) at (0.5, 2) {};
		\node [style=vsourceAMshape, rotate=-90,fill=white] (10) at (0.5, 1) {};
		\node [style=none] (1) at (0.5, 1) {};
		\node [style=none] (2) at (0.5, 0) {};
		\node [style=none] (3) at (1.2, 0) {};
	\end{pgfonlayer}
	\begin{pgfonlayer}{edgelayer}
		\draw (2.center) to (1);
		\draw (1) to (0.center);
		\draw [style=electricalcontrol, in=-15, out=90] (3.center) to (1);
	\end{pgfonlayer}
\end{tikzpicture}
\ \right\rrbracket
:=
\begin{tikzpicture}
	\begin{pgfonlayer}{nodelayer}
		\node [style=none] (0) at (6, 0.5) {};
		\node [style=none] (1) at (6, -1.25) {};
		\node [style=none] (2) at (6.75, -1.25) {};
		\node [style=none] (3) at (6.75, 0.5) {};
		\node [style=none] (4) at (7.75, -1.25) {};
		\node [style=X] (5) at (6.75, -0.25) {};
	\end{pgfonlayer}
	\begin{pgfonlayer}{edgelayer}
		\draw [style=classicalwire, in=-15, out=90, looseness=1.25] (4.center) to (5);
		\draw (5) to (3.center);
		\draw (0.center) to (1.center);
		\draw (2.center) to (5);
	\end{pgfonlayer}
\end{tikzpicture}
=
\begin{tikzpicture}
	\begin{pgfonlayer}{nodelayer}
		\node [style=none] (6) at (2.5, 0.75) {};
		\node [style=none] (7) at (2, -0.75) {};
		\node [style=none] (8) at (3.5, -0.75) {};
		\node [style=none] (9) at (4, 0.75) {};
		\node [style=none] (10) at (4.5, -1.5) {};
		\node [style=X] (11) at (4, 0) {};
		\node [style=Z] (12) at (2.5, 0) {};
		\node [style=Z] (13) at (3, -0.75) {};
		\node [style=none] (14) at (4.5, -0.75) {};
		\node [style=none] (15) at (2, -1.5) {};
		\node [style=none] (17) at (3.5, -1.5) {};
	\end{pgfonlayer}
	\begin{pgfonlayer}{edgelayer}
		\draw (11) to (9.center);
		\draw [in=-150, out=90] (8.center) to (11);
		\draw [in=-150, out=90] (7.center) to (12);
		\draw (12) to (6.center);
		\draw (15.center) to (7.center);
		\draw [in=-30, out=90] (13) to (12);
		\draw (8.center) to (17.center);
		\draw [in=-30, out=90] (14.center) to (11);
		\draw [style=classicalwire] (10.center) to (14.center);
	\end{pgfonlayer}
\end{tikzpicture}
 \ ,
\hspace*{.5cm}
\left\llbracket\
\begin{tikzpicture}
	\begin{pgfonlayer}{nodelayer}
		\node [style=none] (4) at (2.9, 2) {};
		\node [style=isourceAMshape, rotate=-90,fill=white] (50) at (2.9, 1) {};
		\node [style=none] (5) at (2.9, 1) {};
		\node [style=none] (6) at (2.9, 0) {};
		\node [style=none] (7) at (2.25, 0) {};
	\end{pgfonlayer}
	\begin{pgfonlayer}{edgelayer}
		\draw (6.center) to (5);
		\draw (5) to (4.center);
		\draw [style=electricalcontrol, in=-165, out=90] (7.center) to (5);
	\end{pgfonlayer}
\end{tikzpicture}
\ \right\rrbracket
:=
\begin{tikzpicture}
	\begin{pgfonlayer}{nodelayer}
		\node [style=none] (0) at (7, -2.75) {};
		\node [style=none] (2) at (7.75, -2.75) {};
		\node [style=Z] (3) at (7, -3.5) {};
		\node [style=none] (4) at (7, -5) {};
		\node [style=none] (5) at (7.75, -5) {};
		\node [style=Z] (6) at (7.75, -4.25) {};
		\node [style=none] (7) at (6.25, -4.25) {};
		\node [style=none] (8) at (6.25, -5) {};
		\node [style=Z] (9) at (7.75, -3.5) {};
	\end{pgfonlayer}
	\begin{pgfonlayer}{edgelayer}
		\draw (3) to (0.center);
		\draw [in=-150, out=90] (7.center) to (3);
		\draw [style=classicalwire] (8.center) to (7.center);
		\draw (9) to (2.center);
		\draw (5.center) to (6);
		\draw (3) to (4.center);
	\end{pgfonlayer}
\end{tikzpicture}
=
\begin{tikzpicture}
	\begin{pgfonlayer}{nodelayer}
		\node [style=none] (19) at (2.75, -2.75) {};
		\node [style=none] (20) at (3.25, -4.25) {};
		\node [style=none] (21) at (4.75, -4.25) {};
		\node [style=none] (22) at (4.25, -2.75) {};
		\node [style=X] (24) at (4.25, -3.5) {};
		\node [style=Z] (25) at (2.75, -3.5) {};
		\node [style=none] (27) at (3.75, -4.25) {};
		\node [style=none] (28) at (3.25, -5) {};
		\node [style=none] (29) at (4.75, -5) {};
		\node [style=Z] (30) at (3.75, -4.25) {};
		\node [style=none] (31) at (2.25, -4.25) {};
		\node [style=none] (32) at (2.25, -5) {};
	\end{pgfonlayer}
	\begin{pgfonlayer}{edgelayer}
		\draw (24) to (22.center);
		\draw [in=-30, out=90] (21.center) to (24);
		\draw [in=-30, out=90] (20.center) to (25);
		\draw (25) to (19.center);
		\draw (28.center) to (20.center);
		\draw (21.center) to (29.center);
		\draw [in=-150, out=90] (27.center) to (24);
		\draw [in=-150, out=90] (31.center) to (25);
		\draw [style=classicalwire] (32.center) to (31.center);
	\end{pgfonlayer}
\end{tikzpicture}
$$
Or in bastard spider notation, the controlled voltage and current sources have the following form:
$$
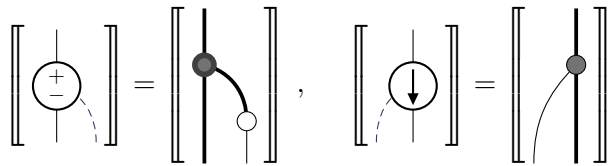
\left\llbracket\
\begin{tikzpicture}
	\begin{pgfonlayer}{nodelayer}
		\node [style=none] (0) at (0.5, 2) {};
		\node [style=vsourceAMshape, rotate=-90,fill=white] (10) at (0.5, 1) {};
		\node [style=none] (1) at (0.5, 1) {};
		\node [style=none] (2) at (0.5, 0) {};
		\node [style=none] (3) at (1.2, 0) {};
	\end{pgfonlayer}
	\begin{pgfonlayer}{edgelayer}
		\draw (2.center) to (1);
		\draw (1) to (0.center);
		\draw [style=electricalcontrol, in=-15, out=90] (3.center) to (1);
	\end{pgfonlayer}
\end{tikzpicture}
\ \right\rrbracket
=
\left\llbracket\
\begin{tikzpicture}
	\begin{pgfonlayer}{nodelayer}
		\node [style=Xthick] (47) at (3, 3) {};
		\node [style=none] (48) at (3, 1.25) {};
		\node [style=none] (49) at (3, 4) {};
		\node [style=none] (50) at (3.75, 1.25) {};
		\node [style=none] (51) at (3, 4) {};
		\node [style=Z] (52) at (3.75, 2) {};
	\end{pgfonlayer}
	\begin{pgfonlayer}{edgelayer}
		\draw [thick] (47) to (51.center);
		\draw [thick] (47) to (48.center);
		\draw [thick, in=345, out=90] (52) to (47);
		\draw (50.center) to (52);
	\end{pgfonlayer}
\end{tikzpicture}
\ \right\rrbracket\ ,
\hspace*{.5cm}
\left\llbracket\
\begin{tikzpicture}
	\begin{pgfonlayer}{nodelayer}
		\node [style=none] (4) at (2.9, 2) {};
		\node [style=isourceAMshape, rotate=-90,fill=white] (50) at (2.9, 1) {};
		\node [style=none] (5) at (2.9, 1) {};
		\node [style=none] (6) at (2.9, 0) {};
		\node [style=none] (7) at (2.25, 0) {};
	\end{pgfonlayer}
	\begin{pgfonlayer}{edgelayer}
		\draw (6.center) to (5);
		\draw (5) to (4.center);
		\draw [style=electricalcontrol, in=-165, out=90] (7.center) to (5);
	\end{pgfonlayer}
\end{tikzpicture}
\ \right\rrbracket
=
\left\llbracket\
\begin{tikzpicture}
	\begin{pgfonlayer}{nodelayer}
		\node [style=X] (53) at (6, 3) {};
		\node [style=none] (54) at (6, 1.25) {};
		\node [style=none] (55) at (6, 4) {};
		\node [style=none] (56) at (5.25, 1.25) {};
	\end{pgfonlayer}
	\begin{pgfonlayer}{edgelayer}
		\draw [thick] (54.center) to (53);
		\draw [thick] (53) to (55.center);
		\draw [in=-135, out=90] (56.center) to (53);
	\end{pgfonlayer}
\end{tikzpicture}
\ \right\rrbracket
$$
On the other hand the dual multicoloured prop $\Aff\Isot\Rel_{\R(x)}^M$ gives a semantics for all of the uncontrolled components we have discussed so far as well as voltmeters and ammeters:
$$
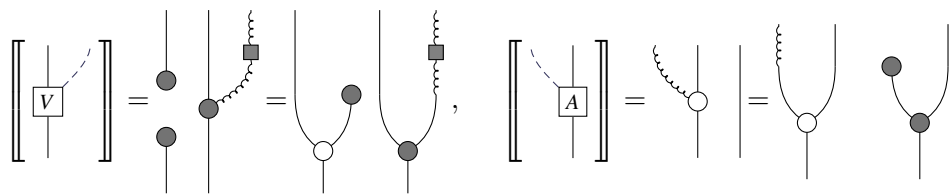
\left\llbracket\
\begin{tikzpicture}
	\begin{pgfonlayer}{nodelayer}
		\node [style=none] (58) at (9.25, 1) {};
		\node [style=none] (59) at (9.25, -1) {};
		\node [style=none] (60) at (10, 1) {};
		\node [style=map] (61) at (9.25, 0) {$V$};
	\end{pgfonlayer}
	\begin{pgfonlayer}{edgelayer}
		\draw [style=electricalcontrol, in=-90, out=45, looseness=0.75] (61) to (60.center);
		\draw (59.center) to (61);
		\draw (61) to (58.center);
	\end{pgfonlayer}
\end{tikzpicture}
\ \right\rrbracket
=
\begin{tikzpicture}
	\begin{pgfonlayer}{nodelayer}
		\node [style=X] (62) at (12.25, 0) {};
		\node [style=X] (63) at (11.5, -0.5) {};
		\node [style=X] (64) at (11.5, 0.5) {};
		\node [style=s] (65) at (13, 1) {};
		\node [style=none] (66) at (13, 1.75) {};
		\node [style=none] (67) at (11.5, 1.75) {};
		\node [style=none] (68) at (11.5, -1.5) {};
		\node [style=none] (69) at (12.25, -1.5) {};
		\node [style=none] (70) at (12.25, 1.75) {};
	\end{pgfonlayer}
	\begin{pgfonlayer}{edgelayer}
		\draw (67.center) to (64);
		\draw (68.center) to (63);
		\draw (69.center) to (62);
		\draw (70.center) to (62);
		\draw [style=classicalwire, in=-90, out=30] (62) to (65);
		\draw [style=classicalwire] (65) to (66.center);
	\end{pgfonlayer}
\end{tikzpicture}
=
\begin{tikzpicture}
	\begin{pgfonlayer}{nodelayer}
		\node [style=s] (107) at (27, 1.75) {};
		\node [style=none] (108) at (27, 2.5) {};
		\node [style=none] (109) at (24.5, 1) {};
		\node [style=none] (110) at (25, -0.75) {};
		\node [style=none] (111) at (26.5, -0.75) {};
		\node [style=none] (112) at (26, 1) {};
		\node [style=Z] (113) at (25, 0) {};
		\node [style=X] (114) at (26.5, 0) {};
		\node [style=X] (115) at (25.5, 1) {};
		\node [style=none] (116) at (27, 1) {};
		\node [style=none] (117) at (26, 2.5) {};
		\node [style=none] (118) at (24.5, 2.5) {};
	\end{pgfonlayer}
	\begin{pgfonlayer}{edgelayer}
		\draw [style=classicalwire] (107) to (108.center);
		\draw [style=classicalwire] (116.center) to (107);
		\draw [in=-90, out=30] (114) to (116.center);
		\draw [in=-90, out=30] (113) to (115);
		\draw (114) to (111.center);
		\draw (113) to (110.center);
		\draw [in=-90, out=150] (114) to (112.center);
		\draw [in=-90, out=150] (113) to (109.center);
		\draw (112.center) to (117.center);
		\draw (109.center) to (118.center);
	\end{pgfonlayer}
\end{tikzpicture}
\ ,
\hspace*{.5cm}
\left\llbracket\
\begin{tikzpicture}
	\begin{pgfonlayer}{nodelayer}
		\node [style=none] (84) at (17.5, 1) {};
		\node [style=none] (85) at (17.5, -1) {};
		\node [style=none] (86) at (16.75, 1) {};
		\node [style=map] (87) at (17.5, 0) {$A$};
	\end{pgfonlayer}
	\begin{pgfonlayer}{edgelayer}
		\draw [style=electricalcontrol, in=-90, out=135, looseness=0.75] (87) to (86.center);
		\draw (85.center) to (87);
		\draw (87) to (84.center);
	\end{pgfonlayer}
\end{tikzpicture}
\ \right\rrbracket
=
\begin{tikzpicture}
	\begin{pgfonlayer}{nodelayer}
		\node [style=Z] (13) at (26.75, -0.25) {};
		\node [style=none] (15) at (26, 0.75) {};
		\node [style=none] (16) at (26.75, 0.75) {};
		\node [style=none] (18) at (27.5, 0.75) {};
		\node [style=none] (19) at (26.75, -1.25) {};
		\node [style=none] (20) at (27.5, -1.25) {};
	\end{pgfonlayer}
	\begin{pgfonlayer}{edgelayer}
		\draw (19.center) to (13);
		\draw (13) to (16.center);
		\draw [style=classicalwire, in=-90, out=150] (13) to (15.center);
		\draw (18.center) to (20.center);
	\end{pgfonlayer}
\end{tikzpicture}
=
\begin{tikzpicture}
	\begin{pgfonlayer}{nodelayer}
		\node [style=Z] (0) at (22.5, -0.25) {};
		\node [style=X] (1) at (24.5, -0.25) {};
		\node [style=none] (2) at (22, 0.75) {};
		\node [style=none] (3) at (23, 0.75) {};
		\node [style=none] (4) at (24, 0.75) {};
		\node [style=none] (5) at (25, 0.75) {};
		\node [style=none] (6) at (22.5, -1.25) {};
		\node [style=none] (7) at (24.5, -1.25) {};
		\node [style=X] (8) at (24, 0.75) {};
		\node [style=none] (9) at (22, 1.5) {};
		\node [style=none] (10) at (22, 1.5) {};
		\node [style=none] (11) at (23, 1.5) {};
		\node [style=none] (12) at (25, 1.5) {};
	\end{pgfonlayer}
	\begin{pgfonlayer}{edgelayer}
		\draw (6.center) to (0);
		\draw [in=-90, out=30] (0) to (3.center);
		\draw (7.center) to (1);
		\draw [in=-90, out=30] (1) to (5.center);
		\draw [in=-90, out=150] (1) to (4.center);
		\draw [in=-90, out=150] (0) to (2.center);
		\draw [style=classicalwire] (2.center) to (10.center);
		\draw (3.center) to (11.center);
		\draw (5.center) to (12.center);
	\end{pgfonlayer}
\end{tikzpicture}
$$
These are respectively the orthogonal complements of the interepretations of  following bastard spider diagrams in linear relations:
$$
\left\llbracket\
\begin{tikzpicture}
	\begin{pgfonlayer}{nodelayer}
		\node [style=none] (58) at (9.25, 1) {};
		\node [style=none] (59) at (9.25, -1) {};
		\node [style=none] (60) at (10, 1) {};
		\node [style=map] (61) at (9.25, 0) {$V$};
	\end{pgfonlayer}
	\begin{pgfonlayer}{edgelayer}
		\draw [style=electricalcontrol, in=-90, out=45, looseness=0.75] (61) to (60.center);
		\draw (59.center) to (61);
		\draw (61) to (58.center);
	\end{pgfonlayer}
\end{tikzpicture}
\ \right\rrbracket
=
\left\llbracket\
\begin{tikzpicture}
	\begin{pgfonlayer}{nodelayer}
		\node [style=Z] (21) at (28.5, 2.25) {};
		\node [style=none] (22) at (28.5, 4) {};
		\node [style=none] (23) at (28.5, 1.25) {};
		\node [style=none] (24) at (29.25, 4) {};
		\node [style=s] (25) at (29.25, 3) {};
	\end{pgfonlayer}
	\begin{pgfonlayer}{edgelayer}
		\draw [thick] (22.center) to (21);
		\draw [thick] (21) to (23.center);
		\draw [in=-90, out=30] (21) to (25);
		\draw (25) to (24.center);
	\end{pgfonlayer}
\end{tikzpicture}
\ \right\rrbracket^\perp
\ ,
\hspace*{.5cm}
\left\llbracket\
\begin{tikzpicture}
	\begin{pgfonlayer}{nodelayer}
		\node [style=none] (84) at (17.5, 1) {};
		\node [style=none] (85) at (17.5, -1) {};
		\node [style=none] (86) at (16.75, 1) {};
		\node [style=map] (87) at (17.5, 0) {$A$};
	\end{pgfonlayer}
	\begin{pgfonlayer}{edgelayer}
		\draw [style=electricalcontrol, in=-90, out=135, looseness=0.75] (87) to (86.center);
		\draw (85.center) to (87);
		\draw (87) to (84.center);
	\end{pgfonlayer}
\end{tikzpicture}
\ \right\rrbracket
=
\left\llbracket\
\begin{tikzpicture}
	\begin{pgfonlayer}{nodelayer}
		\node [style=Zthick] (26) at (31, 2.25) {};
		\node [style=none] (27) at (31, 4) {};
		\node [style=none] (28) at (31, 1.25) {};
		\node [style=none] (29) at (30.25, 4) {};
		\node [style=none] (30) at (31, 1.25) {};
		\node [style=X] (31) at (30.25, 3.25) {};
	\end{pgfonlayer}
	\begin{pgfonlayer}{edgelayer}
		\draw [thick] (26) to (30.center);
		\draw [thick] (26) to (27.center);
		\draw [thick, in=165, out=-90] (31) to (26);
		\draw (29.center) to (31);
	\end{pgfonlayer}
\end{tikzpicture}
\ \right\rrbracket^\perp
$$
The controlled voltage source is the electrical circuit analogue of a $X$-phase correction in quantum circuits.  Similarly, the controlled current source is the transpose of a measurement in the $X$-basis. However the voltmeters and ammeters have no quantum analogue. This is because, the ground is the zero-relation:

$$
\left\llbracket \
\begin{tikzpicture}[yscale=-1]
	\begin{pgfonlayer}{nodelayer}
		\node [style=none] (0) at (0.25, 0) {};
		\node [ground] (1) at (0.25, -0.5) {};
	\end{pgfonlayer}
	\begin{pgfonlayer}{edgelayer}
		\draw (1) to (0.center);
	\end{pgfonlayer}
\end{tikzpicture}\
\right\rrbracket
=
\begin{tikzpicture}
	\begin{pgfonlayer}{nodelayer}
		\node [style=X] (2) at (4, 0) {};
		\node [style=X] (3) at (4.5, 0) {};
		\node [style=none] (4) at (4, -1) {};
		\node [style=none] (5) at (4.5, -1) {};
	\end{pgfonlayer}
	\begin{pgfonlayer}{edgelayer}
		\draw (4.center) to (2);
		\draw (3) to (5.center);
	\end{pgfonlayer}
\end{tikzpicture}
=
\left\llbracket\
\begin{tikzpicture}
	\begin{pgfonlayer}{nodelayer}
		\node [style=Z] (0) at (31, 2.25) {};
		\node [style=none] (2) at (31, 1.25) {};
		\node [style=none] (4) at (31, 1.25) {};
	\end{pgfonlayer}
	\begin{pgfonlayer}{edgelayer}
		\draw [thick] (0) to (4.center);
	\end{pgfonlayer}
\end{tikzpicture}
\ \right\rrbracket^\perp
$$
 This is in contrast to the  quantum discard which is literally the discard relation; despite the fact that the quantum discard and ground are denoted by the same symbol in their respective disciplines, their semantics are orthogonal!

Just as Spekkens' toy model can be interpreted as an epistemically restricted toy theory of quantum circuits; by dualizing things, we could also ask what properties of quantum mechanics hold in $\Aff\Isot\Rel_{\F_p}^M$.  This  is an ``ontologically restricted'' toy theory of quantum mechanics; or equivalently an epistemically ``co-restricted'' toy theory. In other words, where there is a {\em minimum} amount of knowledge the observer has about the ontic state so that discarding a state imposes extra equations on the epistemic state.
We  showed how affine isotropic relations are not compatible with discarding; however this view of measurement is not compatible with quantum theory for the following reason. If we added the quantum analogue of  ammeters to our relational semantics of pure stabilizer circuits, we could compose the ammeter  with another ammeter conjugated by the symplectic Fourier transform as follows:
$$
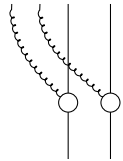
\begin{tikzpicture}
	\begin{pgfonlayer}{nodelayer}
		\node [style=Z] (0) at (12.75, -9.25) {};
		\node [style=none] (1) at (11.75, -7.5) {};
		\node [style=none] (2) at (12.75, -10.25) {};
		\node [style=none] (3) at (13.5, -10.25) {};
		\node [style=Z] (4) at (13.5, -9.25) {};
		\node [style=none] (5) at (12.25, -7.5) {};
		\node [style=none] (6) at (13.5, -7.5) {};
		\node [style=none] (7) at (12.75, -7.5) {};
	\end{pgfonlayer}
	\begin{pgfonlayer}{edgelayer}
		\draw (2.center) to (0);
		\draw [style=classicalwire, in=-90, out=135] (0) to (1.center);
		\draw (4) to (6.center);
		\draw [style=classicalwire, in=-90, out=135] (4) to (5.center);
		\draw (4) to (3.center);
		\draw (0) to (7.center);
	\end{pgfonlayer}
\end{tikzpicture}
$$
This would allow us to simultaneously measure the  $Z$ and $X$ observables which is not possible in quantum mechanics due to the uncertainty principle.
\section{Error correction}
\label{sec:qec}
In this section, we show how to implement quantum error correction protocols for stabilizer codes using the string diagrams we developed in the previous section.  See \cite{gottesman} or \cite{nielsen} for reference on stabilizers codes and error correction.  Nothing in this section is particularly novel from a technical point of view; however, it is conceptually different from the way that stabilizer codes are usually explained.

The graphical algebra approach to error correction which we employ in this section strictly generalizes that of \cite{grok}; where they use linear relations over $\F_2$ as the semantics for CSS-codes. Although we give much more elementary examples.

Fix an odd prime local dimension  $p$.
Consider an affine coisotropic subspace $S=L+a \subseteq \F_p^{2n}$ where $L$ has dimension $n+k$ (or for qubits take an affine coisotropic subspace $S=L+a \subseteq \F_p^{2n}$ without linear phase).  Then the associated projector on $n$-qudits is called a $[n,k]$-stabilizer code, as it encodes $k$ logical qudits into $n$ physical qudits. 
The relationship between logical and physical qudits can be understood in terms of pictures.  We will draw the doubled string diagrams for calculation; accompanied by the qudit stabilizer bastard spider diagrams to give a less cluttered presentation.

Recalling Corollary \ref{cor:affsymstine}, fix a unitary dilation $U$ of $S$:
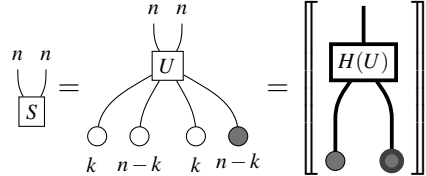
$$

\ \right\rrbracket
$$
The induced isometry, called the {\bf encoder}, embeds $k$ logical qudits into $n$ physical qudits:
$$
%
\ \right\rrbracket
$$
Splitting this projector fixes a basis $\{b_1,\ldots, b_{n-k}\}$ for $L^\omega$  which fixes the possible measurement outcomes:
$$
%
\ \right\rrbracket
$$
Suppose that Alice encodes a state on $k$ logical qudits into  $n$ physical qudits  using this isometry and then sends it to Bob on a noisy quantum channel.
Because Weyl operators form a unitary operator basis,  given any error on a quantum channel we can decompose it as a linear combination of Weyl operators  ${\mathcal W}(e)$ (see \cite[\S 10.3.1]{nielsen}).

To detect the error, Bob applies the non-destructive measurement with respect to $H(U)(1_{k}\otimes{\cal Z}^{\otimes (n-k)})H(U)^\dag$, which we draw in our calculus as follows:

$$
%
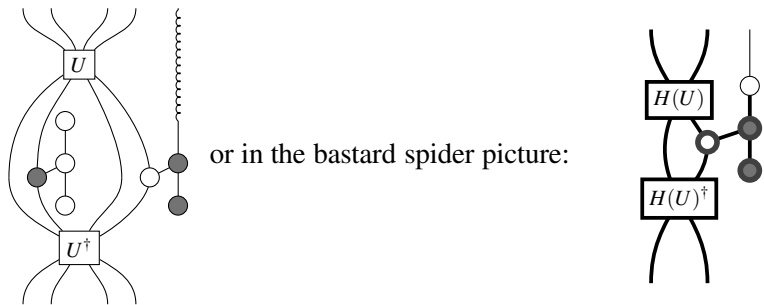
\text{\ \ \ or in the bastard spider picture: }
%
$$

 Because stabilizers preserve Weyl-operators under conjugation, in the symplectic picture, we have: $$W((a,b),(c,d)) = U^\dag W(e)U$$ Therefore:
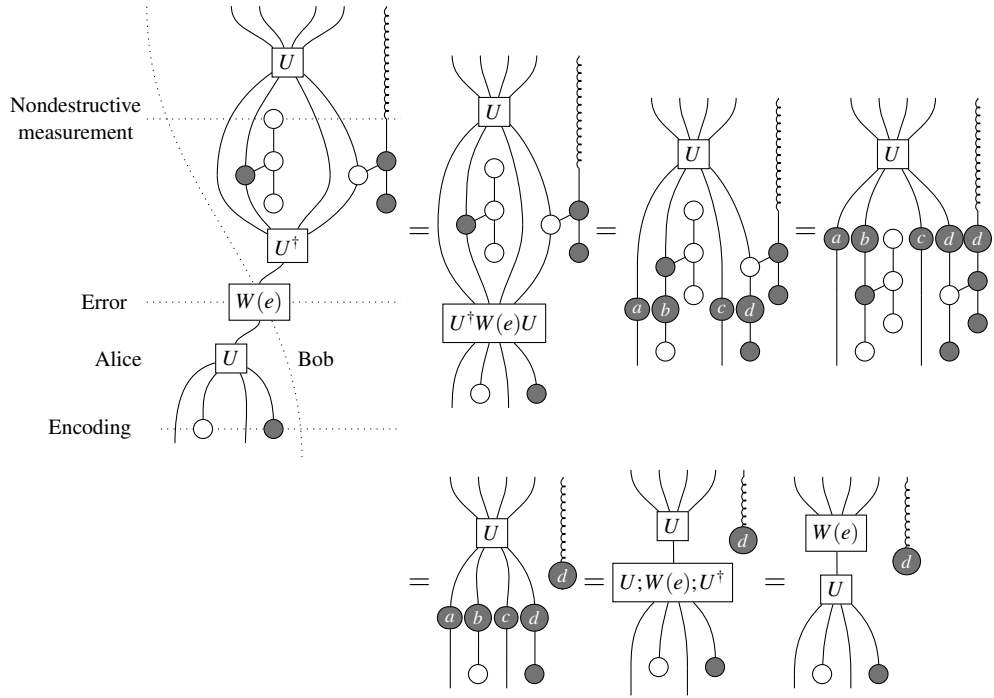
\begin{align*}
%
\end{align*}
In the bastard spider picture, that is:
$$
%
$$
The tuple $d \in \F_p^{n-k}$ is called the {\bf syndrome}. The syndrome measures the displacement of the basis elements $b_i$ by errors.
An error $W(e)$ is {\bf undetectable} if and only if the syndrome is the zero vector; this is because $e$ commutes with everything in $L+a$ meaning that $e \in L^\omega+a$.  In particular, the trivial error is undetectable; so undetectable errors are indistinguishable from having no errors at all.

To correct for errors, given any nonzero syndrome measurement $d\in \F_p^{n-k}$, Bob picks an error ${\mathcal W}(e)$ which he wishes to correct. The choice of errors which Bob chooses to correct for determines a function $f:\F_p^{n-k}\to\F_p^{2n}$ sending $d\mapsto e$.
Given syndrome $d$, Bob applies the operation ${\mathcal W}(-f(d))$ to his $n$ qudits. Finally, Bob applies  $U^\dag$  to the quantum channel and then discards the last $n-k$ qudits.

If $f$ is an affine transformation, then there is a classically controlled operation
$c_f:C^{\otimes(n-k)} \otimes Q^{\otimes n}\to Q^{\otimes n}$ which imeplements this controlled operation, so that:
$$

$$
Or in the bastard spider picture:
$$
%
$$
So that if $e=f(d)$ then this reduces to the identity channel:
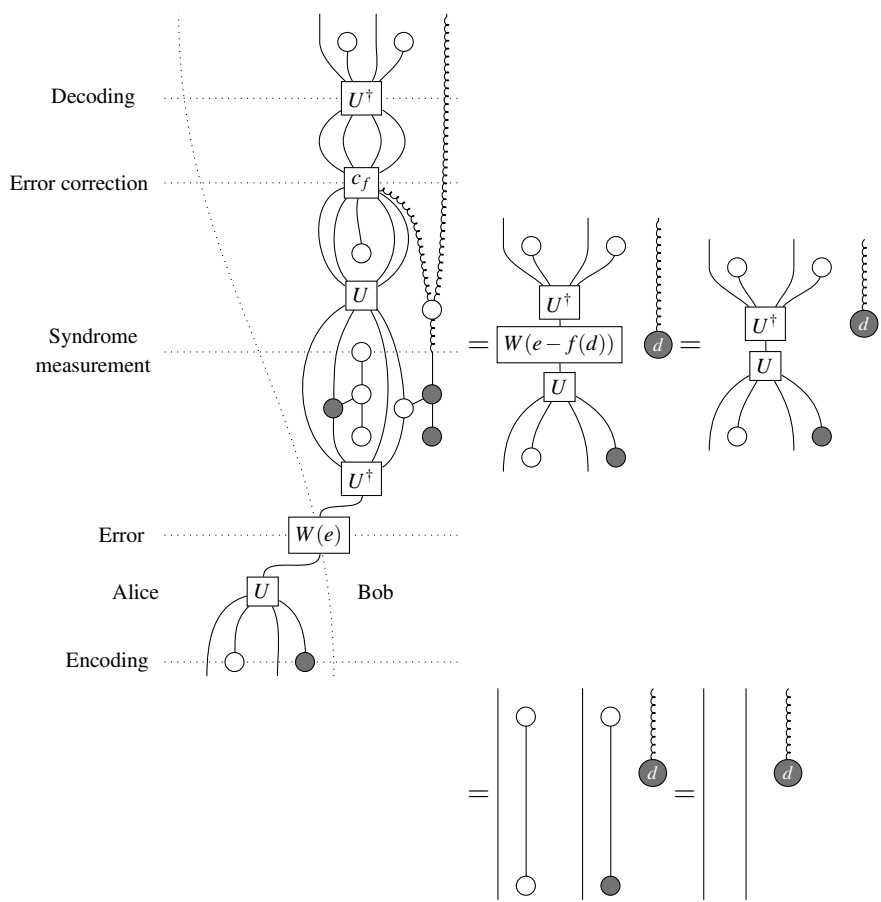
\begin{align*}
%
\end{align*}
Or in the bastard spider picture:
$$
%
$$
We use the threefold qubit repetition code (see \cite[\S 10.1.1]{nielsen}) as an example (which is permissible within our calculus because it is a CSS code/it has trivial linear phase):
\begin{example}
\label{ex:rep}
Consider the Linear subspace:
$$
S = \{ ((z_1,z_2,z_3),(x_1,x_2,x_3)) : x_1=x_2=x_3 \} \subseteq \F_2^{2(3)}
$$
Which can be written in the form of a circuit:
$$
%
\ \right\rrbracket
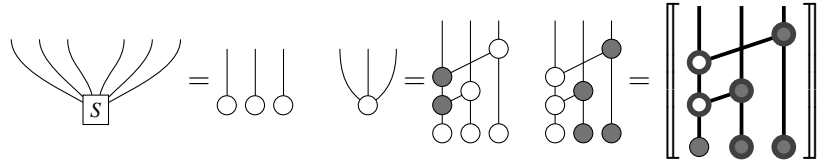
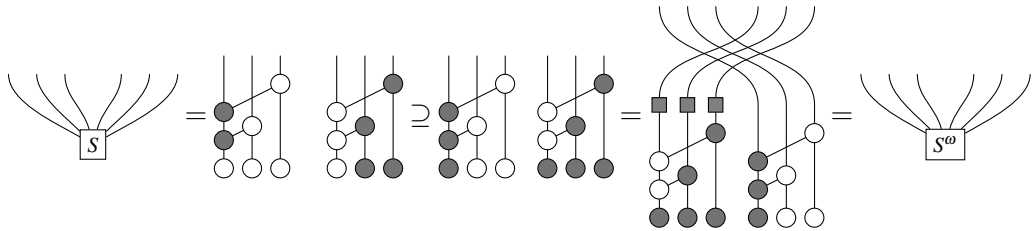
$$
$S$ is coisotropic because:
$$
%
$$
Chopping off the maximally mixed state gives us an encoding map:
$$
%
$$
Also, we will choose to measure in the $Z$-basis, so we split the projector:
$$
%
$$
Suppose there is an error $W((a,b,c),(d,e,f))$, then we have the following error detection circuit:
$$
%
$$
Suppose we want to correct for single  $X$ errors, then we find that:
\begin{itemize}
\item An $X$ error $(d,e,f) = (1,0,0)$ yields syndrome $(e+d,f+d) = (1,1)$
\item An $X$ error $(d,e,f) = (0,1,0)$ yields syndrome $(e+d,f+d) = (1,0)$
\item An $X$ error$(d,e,f) = (0,0,1)$ yields syndrome $(e+d,f+d) = (0,1)$
\item An $X$ error $(d,e,f) = (0,0,0)$ yields syndrome $(e+d,f+d) = (0,0)$
\end{itemize}
Therefore, we want to apply the correction $(s,t) \mapsto (s t, s (t+1),(s+1) t)$. This is a nonlinear function, so we have to leave category $\Aff\Co\Isot\Rel_{\F_p}^M$.  The error correction protocol then has the following form:
$$
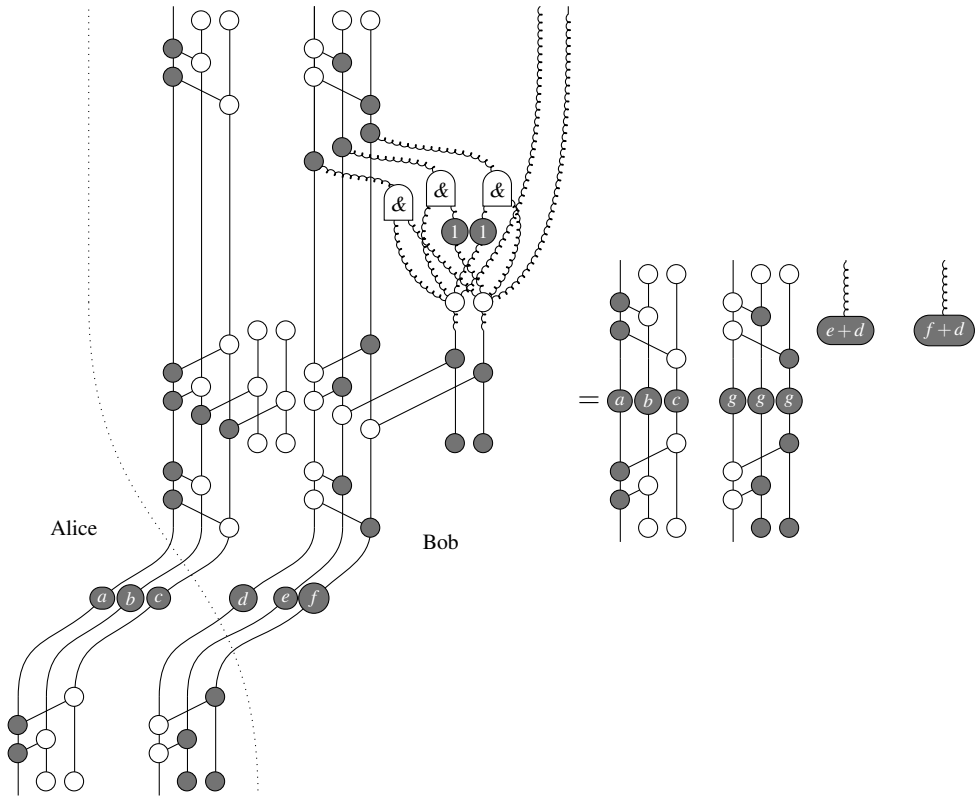
\begin{tikzpicture}
	\begin{pgfonlayer}{nodelayer}
		\node [style=none] (2375) at (319, -25.5) {};
		\node [style=none] (2376) at (319.5, -25.25) {};
		\node [style=none] (2377) at (320, -25.25) {};
		\node [style=Z] (2378) at (319, -24.75) {};
		\node [style=Z] (2379) at (319, -24.25) {};
		\node [style=X] (2380) at (319.5, -24.5) {};
		\node [style=X] (2381) at (320, -23.75) {};
		\node [style=X] (2382) at (320, -25.25) {};
		\node [style=X] (2383) at (319.5, -25.25) {};
		\node [style=none] (2384) at (316.5, -25.5) {};
		\node [style=none] (2385) at (317, -25.25) {};
		\node [style=none] (2386) at (317.5, -25.25) {};
		\node [style=Z] (2387) at (317, -24.5) {};
		\node [style=X] (2388) at (316.5, -24.75) {};
		\node [style=X] (2389) at (316.5, -24.25) {};
		\node [style=Z] (2390) at (317.5, -23.75) {};
		\node [style=none] (2391) at (319, -23.75) {};
		\node [style=none] (2392) at (319.5, -23.75) {};
		\node [style=none] (2393) at (320, -23.75) {};
		\node [style=none] (2394) at (316.5, -23.75) {};
		\node [style=none] (2395) at (317, -23.75) {};
		\node [style=none] (2396) at (317.5, -23.75) {};
		\node [style=Z] (2397) at (317, -25.25) {};
		\node [style=Z] (2398) at (317.5, -25.25) {};
		\node [style=none] (2399) at (321.75, -20.75) {};
		\node [style=none] (2400) at (322.25, -20.75) {};
		\node [style=none] (2401) at (322.75, -20.75) {};
		\node [style=none] (2402) at (319.25, -20.75) {};
		\node [style=none] (2403) at (319.75, -20.75) {};
		\node [style=none] (2404) at (320.25, -20.75) {};
		\node [style=X] (2405) at (318, -22) {$a$};
		\node [style=X] (2406) at (318.5, -22) {$b$};
		\node [style=X] (2407) at (319, -22) {$c$};
		\node [style=X] (2408) at (320.5, -22) {$d$};
		\node [style=X] (2409) at (321.25, -22) {$e$};
		\node [style=X] (2410) at (321.75, -22) {$f$};
		\node [style=none] (2411) at (321.75, -20.75) {};
		\node [style=none] (2412) at (322.25, -20.75) {};
		\node [style=none] (2413) at (322.75, -20.75) {};
		\node [style=none] (2414) at (319.25, -20.75) {};
		\node [style=none] (2415) at (319.75, -20.75) {};
		\node [style=none] (2416) at (320.25, -20.75) {};
		\node [style=Z] (2417) at (321.75, -18.5) {};
		\node [style=Z] (2418) at (321.75, -18) {};
		\node [style=X] (2419) at (322.25, -18.25) {};
		\node [style=X] (2420) at (322.75, -17.5) {};
		\node [style=Z] (2421) at (319.75, -18.25) {};
		\node [style=X] (2422) at (319.25, -18.5) {};
		\node [style=X] (2423) at (319.25, -18) {};
		\node [style=Z] (2424) at (320.25, -17.5) {};
		\node [style=none] (2425) at (321.75, -11.5) {};
		\node [style=none] (2426) at (322.25, -11.75) {};
		\node [style=none] (2427) at (322.75, -11.75) {};
		\node [style=none] (2428) at (319.25, -11.5) {};
		\node [style=none] (2429) at (319.75, -11.75) {};
		\node [style=none] (2430) at (320.25, -11.75) {};
		\node [style=Z] (2431) at (321.75, -19.75) {};
		\node [style=Z] (2432) at (321.75, -20.25) {};
		\node [style=X] (2433) at (322.25, -20) {};
		\node [style=X] (2434) at (322.75, -20.75) {};
		\node [style=Z] (2435) at (319.75, -20) {};
		\node [style=X] (2436) at (319.25, -19.75) {};
		\node [style=X] (2437) at (319.25, -20.25) {};
		\node [style=Z] (2438) at (320.25, -20.75) {};
		\node [style=none] (2439) at (321.75, -20.75) {};
		\node [style=none] (2440) at (322.25, -20.75) {};
		\node [style=none] (2441) at (322.75, -20.75) {};
		\node [style=none] (2442) at (319.25, -20.75) {};
		\node [style=none] (2443) at (319.75, -20.75) {};
		\node [style=none] (2444) at (320.25, -20.75) {};
		\node [style=X] (2445) at (324.25, -17.75) {};
		\node [style=X] (2446) at (324.75, -18) {};
		\node [style=Z] (2447) at (320.75, -18.25) {};
		\node [style=Z] (2448) at (321.25, -18.5) {};
		\node [style=Z] (2449) at (320.75, -19.25) {};
		\node [style=Z] (2450) at (321.25, -19.25) {};
		\node [style=X] (2451) at (324.25, -19.25) {};
		\node [style=X] (2452) at (324.75, -19.25) {};
		\node [style=Z] (2453) at (320.75, -17.25) {};
		\node [style=Z] (2454) at (321.25, -17.25) {};
		\node [style=none] (2455) at (324.25, -17.25) {};
		\node [style=none] (2456) at (324.75, -17.25) {};
		\node [style=Z] (2457) at (322.25, -18.75) {};
		\node [style=Z] (2458) at (322.75, -19) {};
		\node [style=X] (2459) at (319.75, -18.75) {};
		\node [style=X] (2460) at (320.25, -19) {};
		\node [style=none] (2461) at (324.75, -16.75) {};
		\node [style=none] (2462) at (317.75, -16.75) {};
		\node [style=none] (2463) at (320.75, -25.5) {};
		\node [style=none] (2464) at (317.5, -20.75) {Alice};
		\node [style=none] (2465) at (324, -21) {Bob};
		\node [style=Z] (2466) at (324.25, -16.75) {};
		\node [style=Z] (2467) at (324.75, -16.75) {};
		\node [style=X] (2468) at (322.75, -13.75) {};
		\node [style=X] (2469) at (322.25, -14) {};
		\node [style=X] (2470) at (321.75, -14.25) {};
		\node [style=none] (2471) at (323.25, -15) {};
		\node [style=none] (2472) at (325.75, -11.5) {};
		\node [style=none] (2473) at (326.25, -11.5) {};
		\node [style=none] (2474) at (317.75, -11.5) {};
		\node [style=andin] (2475) at (325, -14.75) {};
		\node [style=Z] (2476) at (321.75, -12.25) {};
		\node [style=Z] (2477) at (321.75, -12.75) {};
		\node [style=X] (2478) at (322.25, -12.5) {};
		\node [style=X] (2479) at (322.75, -13.25) {};
		\node [style=Z] (2480) at (319.75, -12.5) {};
		\node [style=X] (2481) at (319.25, -12.25) {};
		\node [style=X] (2482) at (319.25, -12.75) {};
		\node [style=Z] (2483) at (320.25, -13.25) {};
		\node [style=Z] (2484) at (319.75, -11.75) {};
		\node [style=Z] (2485) at (322.25, -11.75) {};
		\node [style=Z] (2486) at (322.75, -11.75) {};
		\node [style=Z] (2487) at (320.25, -11.75) {};
		\node [style=none] (2488) at (324, -14.75) {};
		\node [style=andin] (2489) at (323.25, -15) {};
		\node [style=none] (2490) at (325, -14.75) {};
		\node [style=andin] (2491) at (324, -14.75) {};
		\node [style=X] (2492) at (324.25, -15.5) {$1$};
		\node [style=X] (2493) at (324.75, -15.5) {$1$};
	\end{pgfonlayer}
	\begin{pgfonlayer}{edgelayer}
		\draw (2375.center) to (2378);
		\draw (2376.center) to (2380);
		\draw (2377.center) to (2381);
		\draw (2378) to (2380);
		\draw (2381) to (2379);
		\draw (2378) to (2379);
		\draw (2389) to (2390);
		\draw (2388) to (2387);
		\draw (2387) to (2385.center);
		\draw (2384.center) to (2388);
		\draw (2388) to (2389);
		\draw (2386.center) to (2390);
		\draw (2392.center) to (2380);
		\draw (2379) to (2391.center);
		\draw (2395.center) to (2387);
		\draw (2389) to (2394.center);
		\draw [in=-135, out=90, looseness=1.25] (2394.center) to (2405);
		\draw [in=45, out=-90] (2402.center) to (2405);
		\draw [in=-135, out=90, looseness=1.25] (2395.center) to (2406);
		\draw [in=270, out=45] (2406) to (2403.center);
		\draw [in=45, out=-90] (2404.center) to (2407);
		\draw [in=90, out=-135, looseness=1.25] (2407) to (2396.center);
		\draw [in=45, out=-90] (2399.center) to (2408);
		\draw [in=90, out=-135, looseness=1.25] (2408) to (2391.center);
		\draw [in=-135, out=90, looseness=1.25] (2392.center) to (2409);
		\draw [in=270, out=45] (2409) to (2400.center);
		\draw [in=45, out=-90] (2401.center) to (2410);
		\draw [in=90, out=-135, looseness=1.25] (2410) to (2393.center);
		\draw (2417) to (2419);
		\draw (2420) to (2418);
		\draw (2417) to (2418);
		\draw (2423) to (2424);
		\draw (2422) to (2421);
		\draw (2422) to (2423);
		\draw (2420) to (2427.center);
		\draw (2426.center) to (2419);
		\draw (2418) to (2425.center);
		\draw (2424) to (2430.center);
		\draw (2429.center) to (2421);
		\draw (2423) to (2428.center);
		\draw (2431) to (2433);
		\draw (2434) to (2432);
		\draw (2431) to (2432);
		\draw (2437) to (2438);
		\draw (2436) to (2435);
		\draw (2436) to (2437);
		\draw (2440.center) to (2433);
		\draw (2432) to (2439.center);
		\draw (2443.center) to (2435);
		\draw (2437) to (2442.center);
		\draw (2431) to (2417);
		\draw (2450) to (2448);
		\draw (2448) to (2454);
		\draw (2453) to (2447);
		\draw (2447) to (2449);
		\draw (2451) to (2445);
		\draw (2446) to (2452);
		\draw (2445) to (2455.center);
		\draw (2446) to (2456.center);
		\draw (2438) to (2460);
		\draw (2460) to (2424);
		\draw (2459) to (2421);
		\draw (2435) to (2459);
		\draw (2436) to (2422);
		\draw (2457) to (2433);
		\draw (2420) to (2458);
		\draw (2458) to (2434);
		\draw (2457) to (2419);
		\draw (2460) to (2448);
		\draw (2459) to (2447);
		\draw (2457) to (2445);
		\draw (2458) to (2446);
		\draw [style=classicalwire] (2456.center) to (2461.center);
		\draw [style=dotted, in=90, out=-90, looseness=1.25] (2462.center) to (2463.center);
		\draw (2425.center) to (2470);
		\draw [style=classicalwire, in=-60, out=135] (2467) to (2471.center);
		\draw [style=classicalwire, in=-105, out=165] (2466) to (2471.center);
		\draw [style=classicalwire, in=90, out=-45] (2470) to (2471.center);
		\draw [style=classicalwire, in=-90, out=45] (2466) to (2472.center);
		\draw [style=classicalwire, in=270, out=15, looseness=0.75] (2467) to (2473.center);
		\draw (2476) to (2478);
		\draw (2479) to (2477);
		\draw (2482) to (2483);
		\draw (2481) to (2480);
		\draw [style=classicalwire] (2455.center) to (2466);
		\draw [style=dotted] (2462.center) to (2474.center);
		\draw [style=classicalwire, in=-135, out=135] (2466) to (2488.center);
		\draw [style=classicalwire, in=-90, out=90] (2467) to (2492);
		\draw [style=classicalwire, in=-45, out=90] (2492) to (2488.center);
		\draw [style=classicalwire, in=-105, out=90] (2466) to (2493);
		\draw [style=classicalwire, in=-135, out=90] (2493) to (2490.center);
		\draw [style=classicalwire, in=45, out=-30] (2490.center) to (2467);
		\draw [style=classicalwire, in=-30, out=90] (2490.center) to (2468);
		\draw [style=classicalwire, in=-30, out=90] (2488.center) to (2469);
	\end{pgfonlayer}
\end{tikzpicture}
=
\begin{tikzpicture}
	\begin{pgfonlayer}{nodelayer}
		\node [style=none] (0) at (242.5, -21.25) {};
		\node [style=none] (1) at (243, -21) {};
		\node [style=none] (2) at (243.5, -21) {};
		\node [style=Z] (3) at (242.5, -20.5) {};
		\node [style=Z] (4) at (242.5, -20) {};
		\node [style=X] (5) at (243, -20.25) {};
		\node [style=X] (6) at (243.5, -19.5) {};
		\node [style=X] (7) at (243.5, -21) {};
		\node [style=X] (8) at (243, -21) {};
		\node [style=none] (9) at (240.5, -21.25) {};
		\node [style=none] (10) at (241, -21) {};
		\node [style=none] (11) at (241.5, -21) {};
		\node [style=Z] (12) at (241, -20.25) {};
		\node [style=X] (13) at (240.5, -20.5) {};
		\node [style=X] (14) at (240.5, -20) {};
		\node [style=Z] (15) at (241.5, -19.5) {};
		\node [style=Z] (16) at (241, -21) {};
		\node [style=Z] (17) at (241.5, -21) {};
		\node [style=none] (18) at (242.5, -18) {};
		\node [style=none] (19) at (243, -18) {};
		\node [style=none] (20) at (243.5, -18) {};
		\node [style=none] (21) at (240.5, -18) {};
		\node [style=none] (22) at (241, -18) {};
		\node [style=none] (23) at (241.5, -18) {};
		\node [style=X] (24) at (243, -18.75) {$g$};
		\node [style=X] (25) at (243.5, -18.75) {$g$};
		\node [style=X] (26) at (242.5, -18.75) {$g$};
		\node [style=X] (27) at (244.5, -17.5) {$e+d$};
		\node [style=X] (28) at (246.25, -17.5) {$f+d$};
		\node [style=X] (29) at (240.5, -18.75) {$a$};
		\node [style=X] (30) at (241, -18.75) {$b$};
		\node [style=X] (31) at (241.5, -18.75) {$c$};
		\node [style=none] (32) at (244.5, -16.25) {};
		\node [style=none] (33) at (246.25, -16.25) {};
		\node [style=Z] (34) at (242.5, -17) {};
		\node [style=Z] (35) at (242.5, -17.5) {};
		\node [style=X] (36) at (243, -17.25) {};
		\node [style=X] (37) at (243.5, -18) {};
		\node [style=Z] (38) at (241, -17.25) {};
		\node [style=X] (39) at (240.5, -17) {};
		\node [style=X] (40) at (240.5, -17.5) {};
		\node [style=Z] (41) at (241.5, -18) {};
		\node [style=Z] (42) at (241, -16.5) {};
		\node [style=Z] (43) at (241.5, -16.5) {};
		\node [style=Z] (44) at (243, -16.5) {};
		\node [style=Z] (45) at (243.5, -16.5) {};
		\node [style=none] (46) at (242.5, -16.25) {};
		\node [style=none] (47) at (240.5, -16.25) {};
	\end{pgfonlayer}
	\begin{pgfonlayer}{edgelayer}
		\draw (0.center) to (3);
		\draw (1.center) to (5);
		\draw (2.center) to (6);
		\draw (3) to (5);
		\draw (6) to (4);
		\draw (3) to (4);
		\draw (14) to (15);
		\draw (13) to (12);
		\draw (12) to (10.center);
		\draw (9.center) to (13);
		\draw (13) to (14);
		\draw (11.center) to (15);
		\draw (12) to (22.center);
		\draw (15) to (23.center);
		\draw (4) to (26);
		\draw (26) to (18.center);
		\draw (19.center) to (24);
		\draw (24) to (5);
		\draw (6) to (25);
		\draw (25) to (20.center);
		\draw (14) to (29);
		\draw (29) to (21.center);
		\draw [style=classicalwire] (27) to (32.center);
		\draw [style=classicalwire] (28) to (33.center);
		\draw (34) to (36);
		\draw (37) to (35);
		\draw (34) to (35);
		\draw (40) to (41);
		\draw (39) to (38);
		\draw (39) to (40);
		\draw (39) to (47.center);
		\draw (22.center) to (42);
		\draw (43) to (41);
		\draw (34) to (46.center);
		\draw (44) to (36);
		\draw (37) to (45);
		\draw (36) to (19.center);
		\draw (18.center) to (35);
		\draw (40) to (21.center);
	\end{pgfonlayer}
\end{tikzpicture}
$$
Where 
$$
g:=d+(e+d)(f+d) = e+(e+d)(f+d+1) =f+(e+d+1)(f+d) = de+ef+fd
\mod 2$$
If no more than one of $d,e,f$ is $1$ then $g=0$.  If furthermore $a=b=c=0$, then:
$$
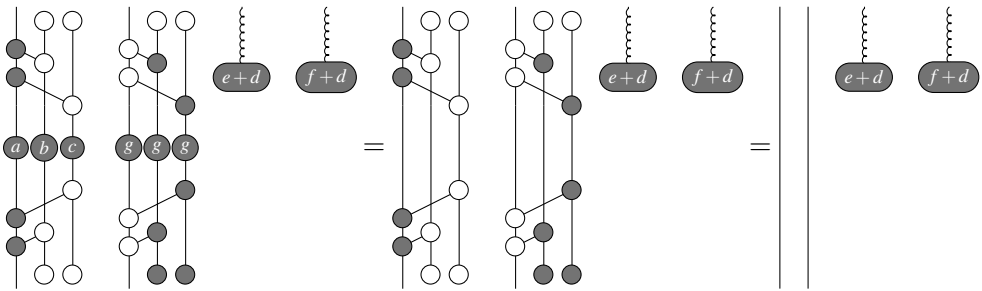
\begin{tikzpicture}
	\begin{pgfonlayer}{nodelayer}
		\node [style=none] (50) at (260.75, -21.25) {};
		\node [style=none] (51) at (261.25, -21) {};
		\node [style=none] (52) at (261.75, -21) {};
		\node [style=Z] (53) at (260.75, -20.5) {};
		\node [style=Z] (54) at (260.75, -20) {};
		\node [style=X] (55) at (261.25, -20.25) {};
		\node [style=X] (56) at (261.75, -19.5) {};
		\node [style=X] (57) at (261.75, -21) {};
		\node [style=X] (58) at (261.25, -21) {};
		\node [style=none] (59) at (258.75, -21.25) {};
		\node [style=none] (60) at (259.25, -21) {};
		\node [style=none] (61) at (259.75, -21) {};
		\node [style=Z] (62) at (259.25, -20.25) {};
		\node [style=X] (63) at (258.75, -20.5) {};
		\node [style=X] (64) at (258.75, -20) {};
		\node [style=Z] (65) at (259.75, -19.5) {};
		\node [style=Z] (66) at (259.25, -21) {};
		\node [style=Z] (67) at (259.75, -21) {};
		\node [style=none] (68) at (260.75, -18) {};
		\node [style=none] (69) at (261.25, -18) {};
		\node [style=none] (70) at (261.75, -18) {};
		\node [style=none] (71) at (258.75, -18) {};
		\node [style=none] (72) at (259.25, -18) {};
		\node [style=none] (73) at (259.75, -18) {};
		\node [style=X] (74) at (261.25, -18.75) {$g$};
		\node [style=X] (75) at (261.75, -18.75) {$g$};
		\node [style=X] (76) at (260.75, -18.75) {$g$};
		\node [style=X] (77) at (262.75, -17.5) {$e+d$};
		\node [style=X] (78) at (264.25, -17.5) {$f+d$};
		\node [style=X] (79) at (258.75, -18.75) {$a$};
		\node [style=X] (80) at (259.25, -18.75) {$b$};
		\node [style=X] (81) at (259.75, -18.75) {$c$};
		\node [style=none] (82) at (262.75, -16.25) {};
		\node [style=none] (83) at (264.25, -16.25) {};
		\node [style=Z] (84) at (260.75, -17) {};
		\node [style=Z] (85) at (260.75, -17.5) {};
		\node [style=X] (86) at (261.25, -17.25) {};
		\node [style=X] (87) at (261.75, -18) {};
		\node [style=Z] (88) at (259.25, -17.25) {};
		\node [style=X] (89) at (258.75, -17) {};
		\node [style=X] (90) at (258.75, -17.5) {};
		\node [style=Z] (91) at (259.75, -18) {};
		\node [style=Z] (92) at (259.25, -16.5) {};
		\node [style=Z] (93) at (259.75, -16.5) {};
		\node [style=Z] (94) at (261.25, -16.5) {};
		\node [style=Z] (95) at (261.75, -16.5) {};
		\node [style=none] (96) at (260.75, -16.25) {};
		\node [style=none] (97) at (258.75, -16.25) {};
	\end{pgfonlayer}
	\begin{pgfonlayer}{edgelayer}
		\draw (50.center) to (53);
		\draw (51.center) to (55);
		\draw (52.center) to (56);
		\draw (53) to (55);
		\draw (56) to (54);
		\draw (53) to (54);
		\draw (64) to (65);
		\draw (63) to (62);
		\draw (62) to (60.center);
		\draw (59.center) to (63);
		\draw (63) to (64);
		\draw (61.center) to (65);
		\draw (62) to (72.center);
		\draw (65) to (73.center);
		\draw (54) to (76);
		\draw (76) to (68.center);
		\draw (69.center) to (74);
		\draw (74) to (55);
		\draw (56) to (75);
		\draw (75) to (70.center);
		\draw (64) to (79);
		\draw (79) to (71.center);
		\draw [style=classicalwire] (77) to (82.center);
		\draw [style=classicalwire] (78) to (83.center);
		\draw (84) to (86);
		\draw (87) to (85);
		\draw (84) to (85);
		\draw (90) to (91);
		\draw (89) to (88);
		\draw (89) to (90);
		\draw (89) to (97.center);
		\draw (72.center) to (92);
		\draw (93) to (91);
		\draw (84) to (96.center);
		\draw (94) to (86);
		\draw (87) to (95);
		\draw (86) to (69.center);
		\draw (68.center) to (85);
		\draw (90) to (71.center);
	\end{pgfonlayer}
\end{tikzpicture}
=
\begin{tikzpicture}
	\begin{pgfonlayer}{nodelayer}
		\node [style=none] (8) at (254.25, -21.25) {};
		\node [style=none] (9) at (254.75, -21) {};
		\node [style=none] (10) at (255.25, -21) {};
		\node [style=Z] (11) at (254.25, -20.5) {};
		\node [style=Z] (12) at (254.25, -20) {};
		\node [style=X] (13) at (254.75, -20.25) {};
		\node [style=X] (14) at (255.25, -19.5) {};
		\node [style=X] (15) at (255.25, -21) {};
		\node [style=X] (16) at (254.75, -21) {};
		\node [style=none] (17) at (252.25, -21.25) {};
		\node [style=none] (18) at (252.75, -21) {};
		\node [style=none] (19) at (253.25, -21) {};
		\node [style=Z] (20) at (252.75, -20.25) {};
		\node [style=X] (21) at (252.25, -20.5) {};
		\node [style=X] (22) at (252.25, -20) {};
		\node [style=Z] (23) at (253.25, -19.5) {};
		\node [style=Z] (24) at (252.75, -21) {};
		\node [style=Z] (25) at (253.25, -21) {};
		\node [style=none] (26) at (254.25, -18) {};
		\node [style=none] (27) at (254.75, -18) {};
		\node [style=none] (28) at (255.25, -18) {};
		\node [style=none] (29) at (252.25, -18) {};
		\node [style=none] (30) at (252.75, -18) {};
		\node [style=none] (31) at (253.25, -18) {};
		\node [style=X] (32) at (256.25, -17.5) {$e+d$};
		\node [style=X] (33) at (257.75, -17.5) {$f+d$};
		\node [style=none] (34) at (256.25, -16.25) {};
		\node [style=none] (35) at (257.75, -16.25) {};
		\node [style=Z] (36) at (254.25, -17) {};
		\node [style=Z] (37) at (254.25, -17.5) {};
		\node [style=X] (38) at (254.75, -17.25) {};
		\node [style=X] (39) at (255.25, -18) {};
		\node [style=Z] (40) at (252.75, -17.25) {};
		\node [style=X] (41) at (252.25, -17) {};
		\node [style=X] (42) at (252.25, -17.5) {};
		\node [style=Z] (43) at (253.25, -18) {};
		\node [style=Z] (44) at (252.75, -16.5) {};
		\node [style=Z] (45) at (253.25, -16.5) {};
		\node [style=Z] (46) at (254.75, -16.5) {};
		\node [style=Z] (47) at (255.25, -16.5) {};
		\node [style=none] (48) at (254.25, -16.25) {};
		\node [style=none] (49) at (252.25, -16.25) {};
	\end{pgfonlayer}
	\begin{pgfonlayer}{edgelayer}
		\draw (8.center) to (11);
		\draw (9.center) to (13);
		\draw (10.center) to (14);
		\draw (11) to (13);
		\draw (14) to (12);
		\draw (11) to (12);
		\draw (22) to (23);
		\draw (21) to (20);
		\draw (20) to (18.center);
		\draw (17.center) to (21);
		\draw (21) to (22);
		\draw (19.center) to (23);
		\draw (20) to (30.center);
		\draw (23) to (31.center);
		\draw [style=classicalwire] (32) to (34.center);
		\draw [style=classicalwire] (33) to (35.center);
		\draw (36) to (38);
		\draw (39) to (37);
		\draw (36) to (37);
		\draw (42) to (43);
		\draw (41) to (40);
		\draw (41) to (42);
		\draw (41) to (49.center);
		\draw (30.center) to (44);
		\draw (45) to (43);
		\draw (36) to (48.center);
		\draw (46) to (38);
		\draw (39) to (47);
		\draw (38) to (27.center);
		\draw (26.center) to (37);
		\draw (42) to (29.center);
		\draw (22) to (29.center);
		\draw (13) to (27.center);
		\draw (12) to (26.center);
		\draw (14) to (39);
	\end{pgfonlayer}
\end{tikzpicture}
=
\begin{tikzpicture}
	\begin{pgfonlayer}{nodelayer}
		\node [style=none] (0) at (248.75, -21) {};
		\node [style=none] (1) at (248.25, -21) {};
		\node [style=X] (2) at (249.75, -17.25) {$e+d$};
		\node [style=X] (3) at (251.25, -17.25) {$f+d$};
		\node [style=none] (4) at (249.75, -16) {};
		\node [style=none] (5) at (251.25, -16) {};
		\node [style=none] (6) at (248.75, -16) {};
		\node [style=none] (7) at (248.25, -16) {};
	\end{pgfonlayer}
	\begin{pgfonlayer}{edgelayer}
		\draw [style=classicalwire] (2) to (4.center);
		\draw [style=classicalwire] (3) to (5.center);
		\draw (1.center) to (7.center);
		\draw (6.center) to (0.center);
	\end{pgfonlayer}
\end{tikzpicture}
$$
Therefore this error correction protocol corrects for at most one $X$-error.
\end{example}

\subsubsection*{Acknowledgments}
The author thanks Alex Cowtan for useful discussions about quantum error correction, and various people who have pointed out minor errors and typos.

\bibliographystyle{eptcs}

\bibliography{lagrel}

\end{document}